\documentclass[11pt,a4paper]{article}

\usepackage[top=3cm, bottom=4cm, left=3cm, right=3cm]{geometry}

\usepackage{amsmath,amssymb}
\usepackage{amsthm}
\usepackage[utf8]{inputenc}
\usepackage[T1]{fontenc}
\usepackage{xspace}

\newcommand{\emi}{({\em i}\,)\xspace}
\newcommand{\emii}{({\em ii}\,)\xspace}
\newcommand{\emiii}{({\em iii}\,)\xspace}

\newcommand{\N}{\mathbb{N}}

\newcommand{\R}{\mathbb{R}}
\newcommand{\C}{\mathbb{C}}

\newcommand{\ab}{{\mathbf{a}}}
\newcommand{\bb}{{\mathbf{b}}}
\newcommand{\thetabf}{\boldsymbol{\theta}}
\newcommand{\wt}{{\tilde{w}}}
\newcommand{\Ft}{{\tilde{F}}}

\newcommand{\la}{\lambda}
\newcommand{\eps}{\epsilon}

\newcommand{\Ac}{{\cal A}}
\newcommand{\Sc}{{\cal S}}
\newcommand{\Hc}{{\cal H}}
\newcommand{\Bc}{{\cal B}}
\newcommand{\Mc}{{\cal M}}
\newcommand{\Qc}{{\cal Q}}
\newcommand{\Kc}{{\cal K}}
\newcommand{\Hf}{\mathfrak{H}}
\newcommand{\Kf}{\mathfrak{K}}
\newcommand{\Psit}{\tilde{\Psi}}
\newcommand{\Wt}{\tilde{W}}
\newcommand{\Ht}{\tilde{H}}
\newcommand{\Ct}{\tilde{\mathbb{C}}}

\newcommand{\scal}[2]{\langle #1| #2\rangle}
\newcommand{\bigscal}[2]{\big\langle #1 \big| #2\big\rangle}
\newcommand{\scl}[2]{( #1| #2 )}
\newcommand{\pder}[2]{\frac{\partial #1}{\partial #2}}

\newcommand{\ot}{\otimes}

\DeclareFontFamily{U}{mathx}{\hyphenchar\font45}
\DeclareFontShape{U}{mathx}{m}{n}{
      <5> <6> <7> <8> <9> <10>
      <10.95> <12> <14.4> <17.28> <20.74> <24.88>
      mathx10
      }{}
\DeclareSymbolFont{mathx}{U}{mathx}{m}{n}
\DeclareMathSymbol{\bigtimes}{1}{mathx}{"91}

\DeclareMathOperator{\Diff}{{\rm Diff}}

\DeclareMathOperator{\id}{{\rm id}}
\DeclareMathOperator{\Int}{{\rm Int}}
\DeclareMathOperator{\supp}{{\rm supp}}

\newcounter{mnotecount}[section]

\newtheorem{thr}{Theorem}

\newtheorem{df}[thr]{Definition}
\newtheorem{lm}[thr]{Lemma}

\numberwithin{equation}{section}
\numberwithin{thr}{section}

\begin{document}
\title{Hilbert spaces built over metrics of fixed signature}
\author{Andrzej Oko{\l}\'ow}
\date{February 23, 2022}
\maketitle
\begin{center}
{\it  Institute of Theoretical Physics, Warsaw University\\ ul. Pasteura 5, 02-093 Warsaw, Poland\smallskip\\
oko@fuw.edu.pl}
\end{center}
\medskip

\begin{abstract}
We construct two Hilbert spaces over the set of all metrics of arbitrary but fixed signature, defined on a manifold. Every state in one of the Hilbert spaces is built of an uncountable number of wave functions representing some elementary quantum degrees of freedom, while every state in the other space is built of a countable number of them. Each Hilbert space is unique up to natural isomorphisms and carries a unitary representation of the diffeomorphism group of the underlying manifold. The Hilbert spaces constructed in the case of signature $(3,0)$ may be possibly used for {canonical} quantization of the ADM formulation of general relativity.  
\end{abstract}\medskip

\begin{center}
\fbox{\begin{minipage}{0.9\textwidth}
\small This is the Accepted Manuscript version of an article accepted for publication in Classical and Quantum Gravity.  IOP Publishing Ltd is not responsible for any errors or omissions in this version of the manuscript or any version derived from it.  The Version of Record is available online at https://doi.org/10.1088/1361-6382/ac4b96 . This Accepted Manuscript is available for reuse under a CC BY-NC-ND licence after the 12 month embargo period provided that all the terms and conditions of the licence are adhered to.
\end{minipage}}
\end{center}\medskip

\section{Introduction}

In \cite{qs-metr} we constructed a space of quantum states and an algebra of quantum observables over the set of all metrics of arbitrary but fixed signature, defined on a manifold. This space and this algebra were obtained by means of the Kijowski's projective method \cite{kpt,q-nonl,proj-lt-I,proj-lt-II,mod-proj}. The motivation for this construction was a desire to find a space of quantum states, which could be used for quantization of the ADM formulation \cite{adm} of general relativity (GR), as the ``position'' part of the ADM phase space is the set of all Riemannian metrics defined on a three-dimensional manifold.

The space of quantum states built in \cite{qs-metr} is not a Hilbert space, but a convex set of mixed states. It turns out, however, that a structural component of that space can be used to obtain two distinct Hilbert spaces related to metrics of arbitrary signature. 

To outline the construction of these Hilbert spaces, which will be denoted by $\Hf$ and $\Kf$, let us first describe that structural component. To this end consider a manifold $\Mc$ and fix a metric signature $(p,p')$ such that $p+p'=\dim\Mc$. Given point $x\in\Mc$, denote by $\Gamma_x$ the set of values at $x$ of all metrics of signature $(p,p')$ defined on $\Mc$. The structural component is a {\em diffeomorphism invariant} assignment
\[
x\mapsto d\mu_x,
\]
where $d\mu_x$ is a natural measure on $\Gamma_x$. 

This assignment allows us to define for every $x\in \Mc$ a Hilbert space $H_x$ being the space of all complex functions on $\Gamma_x$ square integrable with respect to the measure $d\mu_x$. Each Hilbert space $H_x$ thus defined will be treated as an elementary quantum degree of freedom (d.o.f.).

To arrange the Hilbert spaces $\{H_x\}_{x\in\Mc}$ into the Hilbert space $\Hf$, we will proceed as follows. First, we will associate with every point $x$ the set $\Ht_x$  of all half-densities over the tangent space $T_x\Mc$ valued in $H_x$---a section of the bundle-like set $\bigcup_{x\in\Mc}\Ht_x$ will be a half-density on $\Mc$ valued in the Hilbert spaces $\{H_x\}$. Given two such half-densities $\Psit$ and $\Psit'$, we can pair their values $\Psit(x)$ and $\Psit'(x)$ using the inner product on $H_x$, obtaining as a result a complex-valued density over $T_x\Mc$. Doing this point by  point gives us a scalar density on $\Mc$, which can be naturally integrated over the manifold. This procedure, that is, the pairing and the integration, defines an inner product on a set of sections of $\bigcup_{x\in\Mc}\Ht_x$. This set equipped with the product will form a Hilbert space $\Hc_1$. 

The Hilbert space $\Hc_1$ alone will not be suitable for quantization of physical theories since it is more like an uncountable orthogonal sum of the Hilbert spaces $\{H_x\}$ and therefore it does not contain any tensor product of them. But these Hilbert spaces represent independent d.o.f. and a physically acceptable Hilbert space should contain tensor products of $\{H_x\}$. We will include these tensor products as follows.

We will fix a natural number $N\geq 2$ and will consider the set $\Mc_N$ of all $N$-element subsets of $\Mc$. Next we will equip $\Mc_N$ with a differential structure and will associate with each element $\{x_1,\ldots,x_N\}\in\Mc_N$ the Hilbert space  $H_{x_1}\ot\ldots\ot H_{x_N}$. Then, following the construction of the Hilbert space $\Hc_1$ outlined above, we will build a Hilbert space $\Hc_N$.

Finally, the Hilbert space $\Hf$ will be defined as an orthogonal sum of all the spaces $\{\Hc_N\}$.       

Regarding the Hilbert space $\Kf$, its construction will be similar to that of $\Hf$, but simpler: we will consider sections of $\bigcup_{x\in\Mc}H_x$, which are non-zero merely on countable subsets of $\Mc$, and will construct a Hilbert space $\Kc_1$ using these sections and the inner products on $\{H_x\}$. In an analogous way, we will obtain a Hilbert space $\Kc_N$ for $N\geq 2$ using the manifold $\Mc_N$ and the tensor products $\{H_{x_1}\ot\ldots\ot H_{x_N}\}$ assigned to points of the manifold. All the spaces $\{\Kc_N\}$ will be then merged into $\Kf$ by means of an orthogonal sum.          

Thus both Hilbert spaces $\Hf$ and $\Kf$, will be constructed of the same elementary quantum d.o.f. being the Hilbert spaces $\{H_x\}$. A difference between the spaces will be that a state in $\Hf$ will be built of an {\em uncountable} number of wave functions belonging to the spaces $\{H_x\}$ (and their finite tensor products), while a state in $\Kf$ will be built of a {\em countable} number of them.       
      
The diffeomorphism invariance of the assignment $x\mapsto d\mu_x$ used to build both $\Hf$ and $\Kf$ will allow us to define unitary representations of the diffeomorphism group of $\Mc$ on the Hilbert spaces. Moreover, as shown in \cite{qs-metr}, the assignment is unique up to a positive multiplicative constant. This will imply that each Hilbert space $\Hf$ and $\Kf$ is unique up to distinguished unitary maps. We will also show that the two Hilbert spaces $\{\Hf\}$ built over $\Mc=\R$ (for signature $(1,0)$ and $(0,1)$) are separable and that all the Hilbert spaces $\{\Kf\}$ are non-separable.

{With regard to the original motivation underlying this research, that is, to quantization of the ADM formalism: let us note that in the case of signature $(3,0)$ both Hilbert spaces $\Hf$ and $\Kf$ are constructed over the ``position'' part of the ADM phase space. Therefore one may try to apply these spaces to canonical quantization of the formalism. More precisely, since the formalism describes a constrained Hamiltonian system, an appropriate method to use here is the Dirac procedure of quantization of such systems. This procedure (see e.g. \cite{q-gauge-s,modern-cqg}) requires to construct so-called kinematical Hilbert space, which corresponds to} {the unconstrained phase space of the system, and the spaces $\Hf$ and $\Kf$ could possibly serve as such a space for the ADM formalism.}

{According to the Dirac procedure, constraints on the phase space are to be taken into account by \emi defining corresponding constraint operators on the kinematical Hilbert space and \emii by finding states, which are annihilated by the operators. The so-called vector constraints on the ADM phase space generate gauge transformations which coincide with the action of spatial diffeomorphisms on the canonical variables. On the other hand, in the case of signature $(3,0)$, the manifold $\Mc$ can be interpreted as a three-dimensional spatial slice of a spacetime and then the unitary representations of the diffeomorphism group of $\Mc$ on $\Hf$ and $\Kf$, are representations of spatial diffeomorphisms. As it is in loop quantum gravity (see e.g. \cite{cq-diff,rev}), these representations may be helpful in taking into account the vector constraints at the quantum level.}   

{A quantum model resulting from canonical quantization of the ADM formalism is called {\em quantum geometrodynamics} (see e.g. \cite{qgd-rev} and references therein). To the best of our knowledge, so far no kinematical Hilbert space for quantum geometrodynamics equipped with a (non-trivial) unitary representation of the diffeomorphism group, has been known \cite{qgd-rev,rev-1,wdw-rov}. Thus the spaces $\Hf$ and $\Kf$ obtained in the case of signature $(3,0)$ seem to be the first known spaces of this sort. It has to be however emphasized that although the constructions of $\Hf$ and $\Kf$ appear to be fairly natural, neither space consists of square integrable functions on the set of metrics of signature $(3,0)$. Therefore it is not obvious whether $\Hf$ or $\Kf$ can be actually useful for quantum geometrodynamics---a further research is needed to answer this question. As the first step towards this goal, in the forthcoming paper \cite{prep} we will define some operators on $\Hf$ and $\Kf$ related to the ADM canonical variables.}

The paper is organized as follows: Section 2 contains preliminaries---there we recall first of all some necessary notions and facts from \cite{qs-metr}. In Section 3 we construct the Hilbert space $\Hc_1$, and in Section 4 we define a unitary representation of diffeomorphisms of $\Mc$ on $\Hc_1$. Then, in Section 5, we build the Hilbert space $\Hf$, and in Section 6 the Hilbert space $\Kf$. Section 7 contains a summary and an outlook for future research. In Appendix A we define the differential structure on the set $\Mc_N$, and in Appendices B and C we present proofs of some lemmas.  

\section{Preliminaries}

\subsection{Manifold of scalar products of fixed signature}

Suppose that $V$ is a real vector space of non-zero finite dimension. Let us fix a pair $(p,p')$ of non-negative integers such that $p+p'=\dim V$ and denote by $\Gamma$ the set of all scalar products of signature $(p,p')$ defined on $V$. As shown in \cite{qs-metr}, $\Gamma$ is a noncompact connected real-analytic manifold of dimension $\dim V(\dim V+1)/2$. If $(e_i)_{i=1,2,\ldots,\dim V}$ is a basis of $V$, then the following map
\begin{equation}
\Gamma\ni\gamma \mapsto \chi(\gamma):=\big(\gamma(e_i,e_j)\big)_{i\leq j}\in \R^{\dim \Gamma}
\label{map-lin-coor}
\end{equation}
defines\footnote{To treat numbers $\big(\gamma(e_i,e_j)\big)_{i\leq j}$ as an element of $\R^{\dim \Gamma}$ it is necessary to choose an ordering of the numbers. In this case and in other similar cases considered in this paper, we will tacitly assume that an ordering is chosen and, if it is necessary or convenient (see e.g. the function \eqref{f-S} below), then two such orderings are compatible.} a global coordinate system on $\Gamma$. This system will be called here {\em linear coordinate system on} $\Gamma$ and denoted $(\gamma_{i\leq j})$, where 
\[
\Gamma\ni\gamma \mapsto\gamma_{i\leq j}(\gamma):=\gamma(e_i,e_j)\in \R.
\]
Occasionally we will use a single index $\alpha$ to label the coordinates: $(\gamma_{i\leq j})\equiv (\gamma^\alpha)_{\alpha=1,\ldots,\dim\Gamma}$.    

Consider now another real vector space $\check{V}$ such that $\dim\check{V}=\dim V$ and the set $\check{\Gamma}$ of all scalar products of signature $(p,p')$ defined on $\check{V}$. Obviously, $\dim\check{\Gamma}=\dim\Gamma$. Let $\check{\chi}:\check{\Gamma}\to\R^{\dim \Gamma}$ be a map given by a basis $(\check{e}_i)$ of $\check{V}$ via \eqref{map-lin-coor}. Then
\begin{equation}
\chi(\Gamma)=\check{\chi}(\check{\Gamma}).
\label{GG}
\end{equation}

To see that the equality above holds, we will use dual bases $(\omega^i)$ and $(\check{\omega}^i)$ to, respectively, $(e_i)$ and $(\check{e}_i)$. Assume that $(\gamma_{i\leq j})\in \chi(\Gamma)$---this holds if and only if 
\[
\gamma=\sum_i\gamma_{ii}\,\omega^i\ot\omega^i+\sum_{i<j}\gamma_{ij}(\omega^i\ot\omega^j+\omega^j\ot\omega^i)
\]   
is a scalar product on $V$ of signature $(p,p')$. Replacing in the formula above each covector $\omega^i$ by $\check{\omega}^i$ we obtain a scalar product on $\check{V}$ of the same signature, which means that $(\gamma_{i\leq j})\in\check{\chi}(\check{\Gamma})$. Consequently, $\chi(\Gamma)\subset\check{\chi}(\check{\Gamma})$. By virtue of an analogous reasoning, $\chi(\Gamma)\supset\check{\chi}(\check{\Gamma})$ and \eqref{GG} follows.

Thus $\Gamma$ is diffeomorphic to the open subset $\chi(\Gamma)$ of $\R^{\dim \Gamma}$, which since now will be denoted by $\Gamma_\R$:
\[
\chi(\Gamma)\equiv \Gamma_\R.
\]

In practice one often expresses a function on a manifold in terms of a coordinate system on the manifold. Here we will often express a function $f$ on $\Gamma$ in terms of a linear coordinate system i.e. we will use the pull-back $\chi^{-1*}f:\Gamma_\R\to\C$ instead of $f$. However, using coordinates $(\gamma_{i\leq j})$ with the restriction $i\leq j$ is a bit cumbersome and therefore we would like to express functions on $\Gamma_\R$ in terms of all components $(\gamma_{ij})$. Formally this can be achieved in the following way.        

Let $(t_{ij})_{i,j=1,\ldots,\dim V}$ be an element of $\R^{(\dim V)^2}$. The following map
\begin{equation}
\R^{(\dim V)^2}\ni (t_{ij})\mapsto S(t_{ij}):=\Big(\frac{t_{ij}+t_{ji}}{2}\Big)_{i\leq j}\in \R^{\dim\Gamma}
\label{f-S}
\end{equation}
restricted to the set
\[
A:=\{\ (t_{ij})\in \R^{(\dim V)^2} \ | \ t_{ij}=t_{ji} \ \}\cap S^{-1}(\Gamma_\R).
\]
is a bijection onto $\Gamma_{\R}$. In this paper we {\em will not distinguish} between a function $f:\Gamma_{\R}\to \C$ and the pull-back $(S|_A)^\star f:A\to\C$, denoting at the same time elements $(t_{ij})$ of $A$ by $(\gamma_{ij})$. Moreover, we will usually abuse slightly the notation of elements of $\Gamma_\R$ and the notation of a linear coordinate system on $\Gamma$ by dropping the restriction $i\leq j$ in $(\gamma_{i\leq j})$ and  will write simply $(\gamma_{ij})$.

\subsection{Invariant measure on the homogeneous space of scalar products \label{inv-dmu}}

The group $GL(V)$ of all linear automorphisms of $V$ acts naturally on the set $\Gamma$ via pull-back:
\begin{equation}
\begin{aligned}
&GL(V)\times\Gamma \ni(g,\gamma)\mapsto g\gamma:=g^{-1*}\gamma\in \Gamma,\\
&(g^{-1*}\gamma)(v,v')=\gamma(g^{-1}v,g^{-1}v'),\quad v,v'\in V. 
\end{aligned}
\label{g-gamma}
\end{equation}
The pair $(GL(V),\Gamma)$ together with the action \eqref{g-gamma} is a homogeneous space \cite{qs-metr} isomorphic to $GL(\dim V,\R)/O(p,p')$, where $O(p,p')$ is the (pseudo-)orthogonal group, consisting of all those elements of $GL(\dim V,\R)$, which preserve the matrix 
\[
{\rm diag}(\,\overbrace{1,\ldots,1}^p,\overbrace{-1,\ldots,-1}^{p'}\,).
\]

$\Gamma$ is a locally compact Hausdorff (l.c.H.) space being homeomorphic to the open set $\Gamma_\R\subset\R^{\dim\Gamma}$. Since now, unless stated otherwise, a measure will mean a regular Borel measure on a l.c.H.  space (see e.g. \cite{cohn}). The symbol $C^c(Y)$ will denote here the linear space of all real-valued continuous functions of compact support defined on a l.c.H.  space $Y$. If $Y'$ is another such a space, $\alpha:Y\to Y'$ a homeomorphism and $d\mu$ a (regular Borel) measure on $Y$, then there exists\footnote{The existence of $\alpha_\star d\mu$ follows from the Riesz representation theorem (see e.g. \cite{cohn}).} a unique (regular Borel) measure $\alpha_\star d\mu$ on $Y'$ called {\em push-forward measure} such that for every $h\in C^c(Y')$, 
\begin{equation}
\int_{Y'}h \,(\alpha_\star d\mu)=\int_Y (\alpha^\star h)\, d\mu,
\label{*dmu}
\end{equation}
where $\alpha^\star h$ denotes the pull-back of the function $h$: $(\alpha^\star h)(y)=h(\alpha(y))$.

Let $\bar{g}$ be the diffeomorphism
\[
\Gamma\ni \gamma\mapsto g\gamma\in \Gamma
\]
given by $g\in GL(V)$ and the action \eqref{g-gamma}. We say that a measure $d\mu$ on $\Gamma$ is {\em invariant} if for every $g\in GL(V)$ 
\begin{equation}
\bar{g}_\star d\mu=d\mu.
\label{inv-df}
\end{equation}
In \cite{qs-metr} we showed that on $\Gamma$ there exists a (non-zero) invariant measure and that it is unique up to a positive multiplicative constant. 
    
$\Gamma$ is second countable (i.e. $\Gamma$ has a countable base for its topology) being homeomorphic to the open subset $\Gamma_\R$ of $\R^{\dim \Gamma}$. On the other hand, each regular measure on second countable l.c.H. space is $\sigma$-finite\footnote{A measure $d\nu$ on $Y$ is $\sigma$-finite if $Y$ is a union of a sequence $(Y_n)$ of its subsets such that each $Y_n$ has a finite measure under $d\nu$.} \cite{cohn}, which means that every invariant measure on $\Gamma$ is $\sigma$-finite.

Consider now real vector spaces $V_0$, $V_1$ and $V_2$ of the same dimension, and suppose that $\Gamma_i$ ($i=0,1,2$) is the homogeneous space of all scalar products of signature $(p,p')$ on $V_i$ (the signature is fixed and does not depend on $i$). Every linear isomorphism $l_{ij}:V_j\to V_i$ defines a pull-back $l^*_{ij}:\Gamma_i\to \Gamma_j$, being a diffeomorphism between the manifolds.

\begin{lm}
If $d\mu_0$ is an invariant measure on $\Gamma_0$, then
\begin{enumerate}
\item $(l^*_{01})_\star d\mu_0$ is an invariant measure on $\Gamma_1$, which is independent of the choice of linear isomorphism $l_{01}$;     
\item for every triplet of linear isomorphisms $l_{01}$, $l_{02}$ and $l_{12}$ 
\[
(l^*_{12})_\star(l^*_{01})_\star d\mu_0=(l^*_{02})_\star d\mu_0.
\]   
\end{enumerate}
\label{inv-meas}
\end{lm}
\noindent For a proof of these statements see \cite{qs-metr}.

We have shown in \cite{qs-metr} that on every homogeneous space $\Gamma$ of scalar products of signature $(p,p')$, there exists a special\footnote{The metric is invariant with respect to the group action on $\Gamma$.} metric $Q$ called in \cite{qs-metr} {\em natural metric on} $\Gamma$. Let us denote by $(Q_{\alpha\beta})$ components of the metric in a linear coordinate system $(\gamma_{ij})\equiv (\gamma^\alpha)$ on $\Gamma$ given by a map $\chi:\Gamma\to\Gamma_\R$ (see \eqref{map-lin-coor}). It turns out that the natural metric looks the same in every linear coordinate system \cite{qs-metr}. More precisely, there exist smooth functions 
\[
\Delta _{\alpha\beta}:\Gamma_\R\to\R, \quad \alpha,\beta=1,\ldots,\dim\Gamma, 
\]
such that for every space $\Gamma$ as above and for every linear coordinate system on $\Gamma$ the pull-back 
\[
\chi^{-1\star} Q_{\alpha\beta}=\Delta _{\alpha\beta}
\]  
(note that the components $(Q_{\alpha\beta})$ are functions on $\Gamma$).

Moreover, the natural metric $Q$ defines a  measure $d\mu_Q$ on $\Gamma$--- for every continuous (real or complex) function $\Psi$ of compact support on $\Gamma$ \cite{qs-metr}    
\begin{equation}
\int_\Gamma \Psi \,d\mu_Q := \int_{\Gamma_\R} \chi^{-1\star}\Big(\Psi \sqrt{|\det Q_{\alpha\beta}|}\Big)\,d\mu_L =\int_{\Gamma_\R} (\chi^{-1\star}\Psi)\,\Delta \,d\mu_L,
\label{natme}
\end{equation}
where $d\mu_L$ is the Lebesgue measure on $\R^{\dim\Gamma}\supset \Gamma_\R$ and 
\[
\Delta \equiv \sqrt{|\det \Delta _{\alpha\beta}|}=\chi^{-1\star}\Big(\sqrt{|\det Q_{\alpha\beta}|}\Big)
\] 
is a {\em positive} function on $\Gamma_\R$. 

Consider two real vector spaces $V$ and $\check{V}$ of the same dimension and the corresponding spaces $\Gamma$ and $\check{\Gamma}$ of scalar products of the same signature $(p,p')$. Let $d\mu_Q$ and $d\mu_{\check{Q}}$ be measures on, respectively, $\Gamma$ and $\check{\Gamma}$ defined by the corresponding natural metrics $Q$ and $\check{Q}$. If $l:\check{V}\to V$ is a linear isomorphism, then the pull-back $l^*:\Gamma\to \check{\Gamma}$ is a diffeomorphism. It was shown in \cite{qs-metr} that the push-forward measure 
\begin{equation}
l^*_\star d\mu_Q=d\mu_{\check{Q}},
\label{muQ-muQ'}
\end{equation}
which means in particular (i.e. in the case $V=\check{V}$) that $d\mu_Q$ is an {\em invariant} measure on the homogeneous space $\Gamma$. 

Since an invariant measure on $\Gamma$ is unique up to a positive multiplicative constant, for every invariant measure $d\mu$ on the homogeneous space, there exist a number $c>0$ such that
\begin{equation}
d\mu=c\,d\mu_Q.
\label{mu-cmuQ}
\end{equation}

\begin{lm}
Let $\Psi:\Gamma\to \C$ be continuous and $d\mu$ be an invariant measure on $\Gamma$. If
\begin{equation}
\int_\Gamma \bar{\Psi}\Psi\,d\mu = 0,
\label{zero}
\end{equation}
then $\Psi=0$.  
\label{lm-psi-zero}
\end{lm}

\begin{proof}
Suppose that $\Psi(\gamma_0)\neq 0$ for some $\gamma_0\in\Gamma$. It follows from continuity of $\Psi$ that there exists a non-negative compactly supported continuous function $h$ on $\Gamma$ such that $h(\gamma_0)>0$ and $\bar{\Psi}\Psi\geq h$. Then     
\[
\int_\Gamma \bar{\Psi}\Psi\,d\mu \geq \int_\Gamma h\,d\mu =
c\int_\Gamma h\,d\mu_Q = c \int_{\Gamma_\R} (\chi^{-1\star}h)\,\Delta \,d\mu_L>0
\]
---here the second equality holds by virtue of \eqref{mu-cmuQ}, and the last inequality does by virtue of properties of the Lebesgue measure $d\mu_L$ and the fact that the function $\Delta $ is positive everywhere on $\Gamma$. 

Clearly, the inequality above shows that the only continuous function, which satisfies \eqref{zero}, is the constant function of zero value.
\end{proof}

\subsection{Diffeomorphism invariant field of invariant measures}

Let $\Mc$ be a smooth connected paracompact manifold. We fix a pair of non-negative integers $(p,p')$ such that $p+p'=\dim \Mc$ and denote by $\Qc(\Mc)$ the set of all (smooth) metrics of signature $(p,p')$ defined on $\Mc$. Let us denote by $\Gamma_x$ the space of all scalar products on $T_x\Mc$, of signature $(p,p')$. Obviously, for every $x$ the pair $(GL(T_x\Mc),\Gamma_x)$ is a homogeneous space. Moreover, as shown in \cite{qs-metr}, if $\Qc(\Mc)$ is non-empty, then 
\[
\Gamma_x=\{\ q_x \ | \ q\in \Qc(\Mc) \ \},
\]
where $q_x$ denotes the value of the metric $q$ at $x\in\Mc$.
   
An assignment $x\mapsto d\mu_x$, where $d\mu_x$ is a measure on $\Gamma_x$, will be called a {\em field of measures} or a {\em measure field} on the manifold $\Mc$. 

Let $x_0$ be any point of $\Mc$ and $d\mu_{x_0}$ an {\em invariant} measure on $\Gamma_{x_0}$. In \cite{qs-metr} we introduced the following measure field on $\Mc$:
\begin{equation}
x\mapsto d\mu_{x}:=(l^*_{x_0x})_{\star} \,d\mu_{x_0},
\label{diff-inv}
\end{equation}
where $l_{x_0x}:T_x\Mc\to T_{x_0}\Mc$ is a linear isomorphism and $l^*_{x_0x}:\Gamma_{x_0}\to \Gamma_x$ the corresponding pull-back. By virtue of Lemma \ref{inv-meas}, \emi $d\mu_x$ is an invariant measure on $\Gamma_x$, which does not depend on the choice of the map $l_{x_0x}$, and \emii for every two points $x,x'\in\Mc$ and for every linear isomorphism $l_{xx'}:T_{x'}\Mc\to T_{x}\Mc$
\begin{equation}
d\mu_{x'}=(l^*_{xx'})_\star \,d\mu_x.
\label{l*xx}
\end{equation}
The latter property means that the measure field \eqref{diff-inv} is diffeomorphism invariant since $l_{xx'}$ above can be the tangent map $\theta^t$ defined by any diffeomorphism $\theta$ on $\Mc$, which maps $x'$ to $x$. Thus \eqref{diff-inv} is a {\em diffeomorphism invariant field of invariant measures}.      

In \cite{qs-metr} we showed moreover, that the measure field \eqref{diff-inv} is unique up to a positive multiplicative constant, i.e., for any two measure fields $x\mapsto d\mu_{x}$ and $x\mapsto d\check{\mu}_{x}$ constructed according to \eqref{diff-inv}, there exists a number $c>0$ such that for every $x\in\Mc$      
\begin{equation}
d\check{\mu}_{x}=c\,d{\mu_{x}}.
\label{mu-cmu}
\end{equation}

Let $d\mu_{Qx}$ be the invariant measure on $\Gamma_x$ given by the natural metric on the homogeneous space. It follows from \eqref{muQ-muQ'} that the measure field
\[
x\mapsto d\mu_{Qx}
\]   
can be obtained via the formula \eqref{diff-inv}. Taking into account \eqref{mu-cmu}, we see that every measure field \eqref{diff-inv} is of the form
\begin{equation}
x\mapsto c\,d\mu_{Qx}
\label{c-dmuQx}
\end{equation}
for some (independent of $x$) positive number $c$. 

\subsection{Pseudo Hilbert space of half-densities}

Let $W$ be a (possibly infinite dimensional) complex vector space, $V$ a finite dimensional real vector space and $\alpha$ a real number. Denote by $B$ the set of all bases of $V$. If $e=(e_i)_{i=1,\ldots,\dim V}$ is a basis of $V$ and $\Lambda=(\Lambda^j{}_i)_{i,j=1,\ldots,\dim V}$ a non-singular real matrix, then the symbol $\Lambda e$ will represent the basis $(\Lambda^j{}_ie_j)$. 

An $\alpha$-density over $V$ valued in $W$ is a map $\wt:B\to W$ of the following property: for every two bases $e$ and $\Lambda e$ of $V$,
\begin{equation}
\wt(\Lambda e)=|\det\Lambda\,|^\alpha\,\wt(e),
\label{w-Lw'}
\end{equation}
where $\det\Lambda\equiv\det(\Lambda^j{}_i)\neq 0$.
 
We will denote by $\Wt$ the set of all $\alpha$-densities over $V$ valued in $W$. This set possesses a natural complex vector space structure: if $z\in\C$ and $\wt,\wt'\in\Wt$, then
\begin{align*}
(z\wt)(e)&:=z\,\wt(e), & (\wt+\wt')(e)&:=\wt(e)+\wt'(e).
\end{align*}

Denote by $\Ct$ the vector space of one-densities over $V$ valued in complex numbers. The complex conjugate $\overline{\wt}$ of $\wt\in\Ct$ is an element of $\Ct$ such that 
\[
\overline{\wt}(e)=\overline{\wt(e)}
\]
for a basis $e$ of $V$ (if the equality above holds for $e$, then it does for every basis of $V$). 

Let $\wt,\wt'\in\Ct$ be real-valued. We will say that $\wt'$ is greater than or equal to $\wt$ and write $\wt'\geq \wt$ if 
\begin{equation}
\wt'(e)\geq\wt(e)
\label{w'>w}
\end{equation}
for a basis $e$ of $V$ (if \eqref{w'>w} holds for $e$, then it does for every basis of $V$).

Let $H$ be a (complex) Hilbert space with an inner product $\scal{\cdot}{\cdot}$. Consider a vector space $\Ht $ of half-densities (that is, $\frac{1}{2}$-densities) over $V$ valued in $H$. Let us define the following map:
\begin{equation}
\begin{aligned}
&\Ht \times \Ht  \ni (\wt',\wt)\mapsto \scl{\wt'}{\wt}\in\Ct,\\
&\scl{\wt'}{\wt}(e):=\scal{\wt'(e)}{\wt(e)}.
\end{aligned}
\label{den-prod}
\end{equation}
This map satisfies what follows: 
\begin{align}
\forall \,z_1,z_2\in\C, \,\forall\, \wt,\wt_1,\wt_2\in\Ht  \quad & \scl{\wt}{z_1\wt_1+z_2\wt_2}=z_1\scl{\wt}{\wt_1}+z_2\scl{\wt}{\wt_2},\nonumber\\
\forall\, \wt,\wt'\in\Ht \quad & \overline{\scl{\wt'}{\wt}}= \scl{\wt}{\wt'},\nonumber\\
\forall\,\wt\in \Ht\quad & \scl{\wt}{\wt}\geq 0,\nonumber\\
& \scl{\wt}{\wt}=0\Rightarrow \wt=0, \label{wt=0}
\end{align}
where, abusing slightly the notation, we used the symbol $0$ to denote both the zero of $\Ct$ and the zero of $\Ht$. Therefore the map \eqref{den-prod} will be called {\em density product on} $\Ht $. 

The space $\Ht $ equipped with the density product \eqref{den-prod} will be called {\em pseudo-Hilbert space of half-densities over $V$ valued in $H$}.
  
\section{Construction of the Hilbert space $\Hc_1$ \label{constr-H}}

Let us recall that $\Mc$ is a smooth connected paracompact manifold. For the sake of the construction of the Hilbert space $\Hc_1$, let us fix a pair $(p,p')$ of non-negative integers such that $p+p'=\dim \Mc$ and treat it as a metric signature---all objects used to construct $\Hc_1$, which need a metric signature to be chosen, will be given by this $(p,p')$.    

\subsection{Hilbert half-densities and scalar densities on $\Mc$}

\subsubsection{Definition \label{Hc-df}}

Let $x\mapsto d\mu_x$ be a diffeomorphism invariant field of invariant measures given by \eqref{diff-inv}. It allows to define a separable \cite{qs-metr} Hilbert space for every $x\in\Mc$:  
\begin{equation}
H_x:=L^2(\Gamma_x,d\mu_x).
\label{Hx}
\end{equation}
We will use the symbol $\scal{\cdot}{\cdot}_x$ to represent the inner product on $H_x$.
  
Given a point $x\in\Mc$, let $\Ct_x$ stands for the vector space of all one-densities over the tangent space $T_x\Mc$ valued in $\C$. Denote by $\Ht_x$ the pseudo-Hilbert space of half-densities over $T_x\Mc$ valued in $H_x$, and by $\scl{\cdot}{\cdot}_x$ the density product \eqref{den-prod} on $\Ht _x$ valued in $\Ct_x$.

Let 
\begin{equation}
\tilde{\mathbf{H}}:=\bigcup_{x\in\Mc}\Ht_x.
\label{t-bfH}
\end{equation}
A map $\Psit:\Mc\to\tilde{\mathbf{H}}$ such that  $\Psit(x)\in \Ht _x$ for every $x\in\Mc$, will be called {\em Hilbert half-density on} $\Mc$ (this name comes from a shortening of the precise but inconvenient term ``half-density on $\Mc$ valued in the Hilbert spaces $\{H_x\}$''). In other words, a Hilbert half-density is a section of the bundle-like set $\tilde{\mathbf{H}}$. All such half-densities form a complex vector space with multiplication by complex numbers and addition defined point by point:  
\begin{align}
(z\Psit)(x)&:=z\,\Psit(x), & (\Psit+\Psit')(x)&:=\Psit(x)+\Psit'(x),
\label{linear-str}
\end{align}
where $z\in\C$, and $\Psit$, $\Psit'$ are two Hilbert half-densities.
   
The $\Mc$-support of a Hilbert half-density $\Psit$ on $\Mc$ is the closure of the following set:
\[
\{\ x \in \Mc \ | \ \Psit(x)\neq 0 \ \}.
\]

Let $\tilde{\mathbf{C}}:=\bigcup_{x\in\Mc}\Ct_x$. A map $\tilde{F}:\Mc\to\tilde{\mathbf{C}}$, such that $\tilde{F}(x)\in \Ct_x$ for every $x\in \Mc$, is a complex scalar density on $\Mc$. The set of all such densities is a complex vector space with multiplication by complex numbers and addition defined point by point by formulas analogous to \eqref{linear-str}. Complex conjugate $\overline{\Ft}$ of a density $\Ft$ is defined naturally point by point:
\[
\overline{\Ft}(x):=\overline{\Ft(x)}.
\]
The support of a scalar density $\Ft$ on $\Mc$ is the closure of 
\[
\{\ x \in \Mc \ | \ \Ft(x)\neq 0 \ \}.
\]

\subsubsection{Regularity conditions \label{reg-con}}

As it is said in the introduction to this paper, we are going to define an inner product on a set of sections of $\tilde{\mathbf{H}}$, that is, on a set of Hilbert half-densities. To this end we will pair the values of two such half-densities point by point using the density products $\{\scl{\cdot}{\cdot}_x\}_{x\in\Mc}$, obtaining thereby a scalar density on $\Mc$. This density, once integrated over the manifold, will yield a complex number, which by definition will be the value of the inner product of these two half-densities.      

An important question is how to choose that set of Hilbert half-densities, which together with this inner product, will form the desired Hilbert space. It seems that an obvious answer to this question is that one should choose the set of all half-densities of finite norm with respect to the inner product. However, it has to be proven that this set equipped with the product is indeed a Hilbert space.

Moreover, in the case of signature $(3,0)$, we would like to use the resulting Hilbert space for quantization of GR and this means in particular that we will try to define on the space some operators, which will represent physical observables. To define such operators it is often very convenient to have a dense linear subspace of sufficiently regular (continuous, smooth, of compact support etc.) wave functions. Thus if we defined the desired Hilbert space using all the half-densities of finite norm, then we would have to find within it a linear subspace of sufficiently regular half-densities and prove that this subspace is dense.

To avoid having to carry out these two proofs, we will choose a linear space of sufficiently regular Hilbert half-densities and will simply define the desired Hilbert space as a completion of the former space in the norm defined by the inner product. The issue of a relation between this Hilbert space and that space given by all the half-densities of finite norm, will be postponed for the future, being not very essential at this moment. 

There is, however, a related issue: is this space of sufficiently regular Hilbert half-densities ``large enough'' from a physical point of view? Again, we will postpone this question in its generality for the future and will limit ourselves to examining only a simple particular case---the result of this examination (see Section \ref{Hf-McR}) will suggest that the answer to this (general) question is in affirmative.    

Below we will introduce a notion of continuous Hilbert half-densities (which can be easily modified to a notion of smooth half-densities). Moreover, we will impose on the  continuous half-densities an additional regularity condition, which will guarantee that the half-densities, once paired point by point by means of the density products $\{\scl{\cdot}{\cdot}_x\}$, yield continuous densities on $\Mc$---continuity in combination with a compact support of a density, will ensure that the density is integrable over the manifold.

Let us emphasize finally that the present section is quite technical and it may be skipped on the first reading. 

\paragraph{Continuous scalar densities} Let us begin by recalling the notion of continuous scalar density, which will be a model for introducing the notion of continuous Hilbert half-density. 

Let $U$ be an open subset of $\Mc$ and $\varphi:U\to\R^{\dim \Mc}$ a map defining a coordinate system $(x^i)$ on $U$. Given a scalar density $\Ft$ and  a point $x\in U$, the value $\Ft(x)$ is a one-density over $T_x\Mc$ valued in $\C$. Since $(\partial_{x^k})$ is a basis of the tangent space, then $\Ft\big(x,(\partial_{x^k})\big)$ is a complex number. Every $x\in U$ can be expressed in terms of the coordinate system $(x^i)$, which allows us to define {\em coordinate representation of $\Ft$ in the system $(x^i)$} as a function       
\begin{equation}
\varphi(U)\ni (x^i)\mapsto f(x^i):=\Ft\big(\varphi^{-1}(x^i),(\partial_{x^k})\big)\in \C.
\label{F-cont}
\end{equation}
If $(x^{\prime i})$ is an other coordinate system of the domain $U$, then the corresponding coordinate representation  $f'$ satisfies  
\begin{equation}
f'(x^{\prime i})=\bigg|\det\Big(\pder{x^k}{x^{\prime l}}\Big)\bigg|(x^{\prime i})\,f\big(x^k(x^{\prime i})\big).
\label{f'-f}
\end{equation}
This property means that the function $f$ is continuous if and only if $f'$ is continuous.

We will say that the scalar density $\Ft$ is {\em continuous} if for every local coordinate system $(x^i)$ the corresponding coordinate representation  \eqref{F-cont} is continuous.

\paragraph{Coordinate representation of a Hilbert half-density} Let $U$ be again an open subset of $\Mc$ and $\varphi:U\to\R^{\dim \Mc}$ a map defining a coordinate system $(x^i)$ on $U$. Given Hilbert half-density $\Psit$ and  a point $x\in U$, the value $\Psit(x)$ is a half-density over $T_x\Mc$ valued in $H_x$. Thus $\Psit(x,(\partial_{x^k}))$ is an element of $H_x$ being an equivalence class\footnote{{The Hilbert space $H_x=L^2(\Gamma_x,d\mu_x)$ consists of equivalence classes of those functions on $\Gamma_x$, which are square integrable with respect to $d\mu_x$---two functions belong to the same class if they are equal almost everywhere.}} of a function
\begin{equation}
\Gamma_x\ni\gamma\mapsto \Psi\big(x,(\partial_{x^k}),\gamma\big)\in\C.
\label{psi-repr}
\end{equation}
Now, every $x\in U$ can be expressed in terms of the coordinate system $(x^i)$ and the scalar product $\gamma$ can be expressed in terms of components $(\gamma_{ij})$ given by the basis $(\partial_{x^i})$ of $T_x\Mc$. This allows us to define {\em coordinate representation of $\Psit$ in the system $(x^i)$} as the following function:       
\begin{equation}
\varphi(U)\times \Gamma_\R\ni (x^i,\gamma_{ij})\mapsto \psi(x^i,\gamma_{ij}):=\Psi\big(\varphi^{-1}(x^i),(\partial_{x^k}),\gamma_{ij}\,dx^i\ot dx^j\big)\in \C.
\label{psi-cont}
\end{equation}
Note that since $\Psit(x,(\partial_{x^k}))$ is an equivalence class of functions on $\Gamma_x$, the coordinate representation above is not unique, even if the system $(x^i)$ is fixed. 

To reconstruct the Hilbert half density $\Psit$ on $U$ from its coordinate representative $\psi$, it is enough to observe that the function
\begin{equation}
\Gamma_x\ni\gamma\mapsto\Psi\big(x,(\partial_{x^k}),\gamma\big)=\psi\big(\varphi(x),\gamma(\partial_{x^i},\partial_{x^j})\big)\in\C
\label{P-from-p}
\end{equation}
is a representative of the equivalence class $\Psit\big(x,(\partial_{x^i})\big)\in H_x$.
 
\paragraph{Continuous Hilbert half-densities}
\begin{df}
Let $(x^i)$ be a coordinate system on an open set $U\subset \Mc$. We will say that a Hilbert half-density $\Psit$ is continuous on $U$ in the coordinate system $(x^i)$ if for every $x\in U$ the representative \eqref{psi-repr} of $\Psit(x,(\partial_{x^k}))\in H_x$ can be chosen in such a way that the coordinate representation \eqref{psi-cont} is a continuous map.        
\end{df}

It follows from Lemma \ref{lm-psi-zero} that if a Hilbert half-density $\Psit$ is continuous on $U$ in the coordinate system $(x^i)$, then the choice of the representatives \eqref{psi-repr}, which give the continuous function \eqref{psi-cont}, is unique. In other words, given a coordinate system $(x^i)$, a continuous coordinate representation of $\Psit$ in the system is {\em unique} (provided it exists).  

\begin{lm}
Let $(x^i)$ and $(x^{\prime i})$ be coordinate systems on an open set $U\subset \Mc$. A Hilbert half-density $\Psit$ is continuous on $U$ in the coordinate system $(x^i)$ if and only if it is continuous on $U$ in the coordinate system $(x^{\prime i})$.
\label{lm-cont-coor}
\end{lm}

\begin{proof}
Suppose that for every $x\in U$, a representative of $\Psit(x,(\partial_{x^j}))\in H_x$ is chosen, which results in a coordinate representative $\psi$ of $\Psit$ in the coordinates $(x^j)$.

For every $x\in U$
\[
\Psit(x,(\partial_{x^{\prime i}}))=\bigg|\det\Big(\pder{x^k}{x^{\prime l}}\Big)\bigg|^{1/2}\!\!\!\!\!\!(x)\,\,\Psit(x,(\partial_{x^j})).
\] 
Therefore if $\gamma\mapsto\Psit(x,(\partial_{x^j}),\gamma)$ is the selected representative of $\Psit(x,(\partial_{x^j}))$, then the function 
\[
\gamma\to \bigg|\det\Big(\pder{x^k}{x^{\prime l}}\Big)\bigg|^{1/2}\Psit(x,(\partial_{x^j}),\gamma)
\]
is a representative of $\Psit(x,(\partial_{x^{\prime i}}))\in H_x$. Hence
\[
(x^{\prime i},\gamma'_{ij})\mapsto \psi'(x^{\prime i},\gamma'_{ij})=\bigg|\det\Big(\pder{x^k}{x^{\prime l}}\Big)\bigg|^{1/2}\!\!\!\!\!\!(x^{\prime i})\,\,\Psit\big(\varphi^{\prime -1}(x^{\prime i}),(\partial_{x^j}),\gamma'_{ij}\,dx^{\prime i}\ot dx^{\prime j}\big),
\]
where $\varphi':U\to\R^{\dim\Mc}$ is the map defining the system $(x^{\prime i})$, is a coordinate representation  of $\Psit$ in the system $(x^{\prime i})$. Denoting $(\varphi\circ \varphi^{\prime -1})(x^{\prime i})\equiv (x^j(x^{\prime i}))$ we obtain from the formula above
\begin{multline}
\psi'(x^{\prime i},\gamma'_{ij})=\bigg|\det\Big(\pder{x^k}{x^{\prime l}}\Big)\bigg|^{1/2}\Psit\Big(\varphi^{-1}(x^j(x^{\prime i})),(\partial_{x^j}),\gamma'_{ij}\pder{x^{\prime i}}{x^{k}}\pder{x^{\prime j}}{x^{l}}\,dx^{k}\ot dx^{l}\Big)=\\=\bigg|\det\Big(\pder{x^k}{x^{\prime l}}\Big)\bigg|^{1/2}\psi\Big(x^j(x^{\prime i}),\gamma'_{ij}\pder{x^{\prime i}}{x^{k}}\pder{x^{\prime j}}{x^{l}}\Big),
\label{psi'}
\end{multline}
where $\psi$ is the coordinate representation of $\Psit$ in the system $(x^i)$ introduced at the very beginning of the proof, and the derivatives $\partial {x^k}/\partial{x^{\prime l}}$ and $\partial{x^{\prime l}}/\partial {x^k}$ are treated as functions of $(x^{\prime i})$. 

We conclude that if $\psi$ is a coordinate representative of $\Psit$ in the coordinates $(x^i)$, then $\psi'$ given by \eqref{psi'} is a representative of $\Psit$ in the coordinates $(x^{\prime i})$.  

Taking into account that the transition map $(x^{\prime i})\mapsto \big(x^j(x^{\prime i})\big)$ is smooth, we see from \eqref{psi'} that the coordinate representation $\psi'$ is continuous if and only if the coordinate representation $\psi$ is continuous (the ``only if'' part of these statement comes from the fact that the dependence of $\psi'$ on $\psi$ given by \eqref{psi'} can be inverted to a dependence of $\psi$ on $\psi'$ of an analogous form).
\end{proof}

We will say that a Hilbert half-density $\Psit$ is {\em continuous} if for every local coordinate system, there exists a continuous coordinate representation  \eqref{psi-cont}.

Since now in the case of a continuous Hilbert half-density we will use exclusively its continuous coordinate representations.

\paragraph{Hilbert half-densities of compact and slowly changing $\Gamma_\R$-support} Let us consider again the map $\varphi:U\to\R^{\dim\Mc}$ and the corresponding coordinate system $(x^i)$.

\begin{df}
Suppose that a Hilbert half-density $\Psit$ is continuous, and $\psi$ is its coordinate representation in the system $(x^i)$. We will say that the $\Gamma_\R$-support of $\Psit$ around $x_0\in U$ is compact and slowly changing in the coordinate system $(x^i)$, if there exist an open neighborhood $U_0\subset U$ of $x_0$, and  a compact set $K\subset\Gamma_\R$, such that for every value $(x^i)\in \varphi(U_0)$, the support of the function
\begin{equation}
\Gamma_{\R}\ni(\gamma_{ij})\mapsto \psi_{(x^i)}(\gamma_{ij}):=\psi(x^i,\gamma_{ij})\in\C
\label{psi_xi}
\end{equation}
is contained in $K$. 
\label{slowly-df}
\end{df}
\noindent Let us emphasize that if the support of $\psi_{(x^i)}$ is contained in a compact set, then the support is compact itself (since each closed subset of a compact set is compact). Thus the definition above implies that for every $(x^i)\in\varphi(U_0)$ the support of $\psi_{(x^i)}$ is compact.

\begin{lm}
Let $(x^i)$ and $(x^{\prime i})$ be coordinate systems on an open set $U\subset \Mc $, and let $\Psit$ be a continuous Hilbert half-density. The $\Gamma_\R$-support of $\Psit$ around $x_0\in U$ is compact and slowly changing in the coordinate system $(x^i)$ if and only if the $\Gamma_\R$-support of $\Psit$ around $x_0$ is compact and slowly changing in the coordinate system $(x^{\prime i})$.
\label{slow-coor}
\end{lm}

\begin{proof}
Let $\varphi:U\to \R^{\dim \Mc}$  and $\varphi':U\to\R^{\dim\Mc}$ be maps defining the coordinate systems, respectively, $(x^i)$ and $(x^{\prime i})$. Denote $\big(x^i(x^{\prime j})\big)\equiv\varphi(\varphi^{\prime -1}(x^{\prime j}))$. 

Suppose that the $\Gamma_\R$-support of $\Psit$ around $x_0\in U$ is compact and slowly changing in the coordinate system $(x^i)$. Let $U_0$ and $K$ be the sets introduced in Definition \ref{slowly-df} for the system $(x^i)$. Choose a compact set $U_1\subset U_0$ of non-empty interior $\Int U_1$ such that $x_0\in \Int U_1$. The map
\begin{equation}
\varphi(U)\times \Gamma_\R \ni(x^i,\gamma_{ij})\mapsto \pder{x^i}{x^{\prime k}}\pder{x^j}{x^{\prime l}}\gamma_{ij}\in\Gamma_{\R}
\label{xg-g'}
\end{equation}
is continuous (here we treat the derivative $\partial{x^i}/\partial {x^{\prime k}}$ as a function of $(x^i)$). Therefore the set
\[
K':=\text{the image of $\varphi(U_1)\times K$ under the map \eqref{xg-g'}}
\]
is compact. 

Consider now the following function 
\[
\Gamma_{\R}\ni(\gamma'_{kl})\mapsto \psi'_{(x^{\prime j})}(\gamma'_{kl}):=\psi'(x^{\prime j},\gamma'_{kl})\in\C,
\]
where $\psi'(x^{\prime j},\gamma'_{kl})$ represents the half-density $\Psit$ in the coordinate system $(x^{\prime j})$ (see \eqref{psi-cont}). Let us fix a value $(x^{\prime j})\in \varphi'(U)$ of the coordinates and the corresponding value $\big(x^i(x^{\prime j})\big)\in \varphi(U)$. By virtue of Equation \eqref{psi'}, $(\gamma'_{kl})$ belongs to the support of $\psi'_{(x^{\prime j})}$ if and only if   
\begin{equation}
\gamma_{ij}=\pder{x^{\prime k}}{x^{i}}\pder{x^{\prime l}}{x^{j}} \gamma'_{kl}
\label{gamma2}
\end{equation}
belongs to the support of $\psi_{(x^i)}$ (here $\partial{x^{\prime k}}/\partial{x^{i}}$ is the value of the derivative at $(x^i)$). 

Suppose now that the value $(x^{\prime j})\in \varphi'(\Int U_1)$. Then the corresponding value $(x^{i})\in \varphi(\Int U_1)\subset \varphi(U_0)$. In this situation, if $(\gamma'_{kl})$ belongs to the support of $\psi'_{(x^{\prime j})}$, then $(\gamma_{ij})$ given by \eqref{gamma2} belongs to $\supp \psi_{(x^i)}\subset K$. Moreover, $(\gamma'_{kl})$ is the value of the map \eqref{xg-g'} at $(x^i,\gamma_{ij})\in\varphi(U_1)\times K$. Thus $(\gamma'_{kl})$ is an element of $K'$ by definition of the latter set. 

We thus see that for every value $(x^{\prime j})\in \varphi'(\Int U_1)$ the support of $\psi'_{(x^{\prime j})}$ is contained in the compact set $K'$. Therefore, if the $\Gamma_\R$-support of $\Psit$ around $x_0$ is compact and slowly changing in the coordinate system $(x^i)$, then it is the same in the coordinate system $(x^{\prime i})$. But these coordinate systems are arbitrary and can be swapped in the previous statement. Thus the lemma follows.
\end{proof}

Since now we will say that the $\Gamma_\R$-support around $x_0\in\Mc$ of a continuous Hilbert half-density $\Psit$, is compact and slowly changing if it is compact and slowly changing in every coordinate system defined on a neighborhood of $x_0$. Finally, we will say that the $\Gamma_\R$-support of $\Psit$ is compact and slowly changing if the $\Gamma_\R$-support is such around every point of $\Mc$.  

\begin{lm}
If $\Psit$ and $\Psit'$ are continuous Hilbert half-densities of compact and slowly changing $\Gamma_\R$-support, then any (finite) linear combination of them is a continuous Hilbert half-density of compact and slowly changing $\Gamma_\R$-support.   
\label{PP-lin-comb}
\end{lm}
   
\begin{proof}
Suppose that $U$ is an open subset of $\Mc$ and that a map $\varphi:U\mapsto \R^{\dim\Mc}$ defines a coordinate system $(x^i)$. 

Consider now a linear combination 
\[
z\Psit+z'\Psit'\equiv \tilde{\Xi}, \quad z,z'\in \C
\] 
and suppose that $\psi$ and $\psi'$ are (continuous) coordinate representations of, respectively, $\Psit$ and $\Psit'$ in the system $(x^i)$. Then, for every $x\in U$, the function
\[
z\psi+z'\psi'\equiv \xi
\]     
substituted to the reconstruction formula \eqref{P-from-p} gives us an appropriate linear combination of functions on $\Gamma_x$ representing the equivalence class $z\Psit\big(x,(\partial_{x^i})\big)+z'\Psit'\big(x,(\partial_{x^i})\big)\in H_x$. $\xi$ is thus a {\em continuous} coordinate representation of $\tilde{\Xi}$ in the system $(x^i)$. Since $U$ is an arbitrary open subset of $\Mc$, then $\tilde{\Xi}$ is continuous.

Fix an arbitrary point $x_0\in U$. Let $K$  be a compact subset of $\Gamma_\R$ and  $U_0\subset U$ be an open neighborhood of $x_0$ such that for every value $(x^i)\in \varphi(U_0)$, the support of the function $\psi_{(x^i)}$ related to $\Psit$ via the formulas \eqref{psi_xi} and \eqref{psi-cont}, is contained in $K$. In the same way, let $K'$ be a compact subset of $\Gamma_\R$ and $U'_0\subset U$ be an open neighborhood of $x_0$ such that for every value $(x^i)\in \varphi(U'_0)$, the support of the function $\psi'_{(x^i)}$ related to $\Psit'$ via the formulas \eqref{psi_xi} and \eqref{psi-cont}, is contained in $K'$.

Obviously, if $\xi_{(x^i)}$ is related via \eqref{psi_xi} to $\xi$, then     
\[
\xi_{(x^i)}=z\psi_{(x^i)}+z'\psi'_{(x^i)}.
\] 
Therefore for every value $(x^i)\in \varphi(U_0\cap U'_0)$ the support of $\xi_{(x^i)}$ is contained in $K\cup K'$. Since $U_0\cap U'_0$ is open and $K\cup K'$ compact, the $\Gamma_\R$-support of $\tilde{\Xi}$ around $x_0$ is compact and slowly changing in the coordinate system $(x^i)$. But $x_0$ is an arbitrary point in $\Mc$, and $(x^i)$ an arbitrary local coordinate system. Therefore  the $\Gamma_\R$-support of $\tilde{\Xi}$ is compact and slowly changing. 
\end{proof}

\subsection{Pairing of Hilbert half-densities into scalar densities}

Consider now two Hilbert half-densities $\Psit$ and $\Psit'$ on $\Mc$. Then the map
\begin{equation}
\Mc\ni x\mapsto \scl{\Psit'}{\Psit}(x):=\scl{\Psit'(x)}{\Psit(x)}_x\in \tilde{\mathbf{C}}
\label{hden-pair}
\end{equation}
is a scalar density on $\Mc$. The lemma below ensures that if Hilbert half-densities satisfy the regularity conditions introduced in the previous section, then the resulting density is sufficiently regular for our purposes.

\begin{lm}
Suppose that Hilbert half-densities $\Psit$ and $\Psit'$ are continuous and that the $\Gamma_\R$-support of $\Psit$ is compact and slowly changing. Then the scalar density $\scl{\Psit'}{\Psit}$ is continuous.     
\label{PP-cont}
\end{lm}

\begin{proof}
Let us consider the scalar density $\scl{\Psit'}{\Psit}\equiv\Ft$ and a map $\varphi:U\to \R^{\dim\Mc}$ defining a coordinate system $(x^i)$ on $U\subset\Mc$. Using \eqref{hden-pair} and \eqref{den-prod} we obtain
\[
\Ft\big(x,(\partial_{x^i})\big)=\scl{\Psit'(x)}{\Psit(x)}_x(\partial_{x^i})=\bigscal{\Psit'\big(x,(\partial_{x^i})\big)}{\Psit\big(x,(\partial_{x^i})\big)}_x,
\]         
where $\scal{\cdot}{\cdot}_x$ is the inner product on the Hilbert space $H_x$. Consequently, 
\begin{equation}
\Ft\big(x,(\partial_{x^i})\big)=\int_{\Gamma_x}\overline{\Psit'\big(x,(\partial_{x^i})\big)}\Psit\big(x,(\partial_{x^i})\big)\,d\mu_x=c\int_{\Gamma_x}\overline{\Psit'\big(x,(\partial_{x^i})\big)}\Psit\big(x,(\partial_{x^i})\big)\,d\mu_{Qx}, 
\label{int}
\end{equation}
where in the second step we used the fact that every measure field \eqref{diff-inv} is of the form \eqref{c-dmuQx}. 

For further transformation of $\Ft\big(x,(\partial_{x^i})\big)$ we would like to use Equation \eqref{natme}. In order to do this we have to show that (for fixed $x$) the integrand in \eqref{int} is a continuous function of compact support. To this end let us note that if $\chi:\Gamma_x\to \Gamma_\R$ is a map \eqref{map-lin-coor} given by the basis $(\partial_{x^i})$ of $T_{x}\Mc$, and if $\varphi(x)=(x^i)$, then 
\begin{equation}
\chi^{-1\star}\Big[\overline{\Psit'\big(\varphi^{-1}(x^i),(\partial_{x^i})\big)}\Psit\big(\varphi^{-1}(x^i),(\partial_{x^i})\big)\Big]= \overline{{\psi'}\!_{(x^i)}}\psi_{(x^i)},  
\label{PPpp}
\end{equation}
where ${\psi'}\!_{(x^i)}$ and $\psi_{(x^i)}$ are continuous functions related to, respectively, $\Psit'$ and $\Psit$ via the formulas \eqref{psi_xi} and \eqref{psi-cont}. Thus the integrand in \eqref{int} is continuous.

Let us fix a point $x_0\in U$ and suppose that $U_0\subset U$ is an open neighborhood of $x_0$ introduced in Definition \ref{slowly-df} for the half-density $\Psit$. Then it follows from the assumptions imposed on $\Psit$ that {\em for every} $(x^i)\in\varphi(U_0)$, the support of $\overline{{\psi'}\!_{(x^i)}}\psi_{(x^i)}$ is contained in a {\em compact} set $K\subset \Gamma_\R$ and thereby the support is compact as well. This together with \eqref{PPpp} mean that, indeed, the integrand in \eqref{int} is of compact support. 

Since the integrand in \eqref{int} is a continuous compactly supported function on $\Gamma_x$, we can use \eqref{natme} to get
\[
\Ft\big(x,(\partial_{x^i})\big)=c\int_{\Gamma_\R}\chi^{-1\star}\Big[\overline{\Psit'\big(x,(\partial_{x^i})\big)}\Psit\big(x,(\partial_{x^i})\big)\Big]\,\Delta \,d\mu_L.
\]
Using \eqref{PPpp} once again we obtain
\begin{equation}
f(x^i)=\Ft\big(\varphi^{-1}(x^i),(\partial_{x^i})\big)=c\int_{\Gamma_\R} \overline{{\psi'}\!_{(x^i)}}\psi_{(x^i)} \,\Delta \,d\mu_L
\label{f-ppC}
\end{equation}
for every value $(x^i)\in \varphi(U_0)$.

Now, it follows from assumed continuity of $\Psit'$ and $\Psit$ that the function
\begin{equation}
\varphi(U_0)\times \Gamma_\R\ni (x^i,\gamma_{ij})\mapsto \overline{{\psi'}\!_{(x^i)}}(\gamma_{ij})\psi_{(x^i)}(\gamma_{ij})\Delta (\gamma_{ij})\in \C
\label{ppD}
\end{equation}
is continuous (note that the function $\Delta $ is continuous and independent of $(x^i)$, which follows from the properties of the natural metrics described in Section \ref{inv-dmu}). Therefore the function \eqref{ppD} becomes a {\em bounded} one once restricted to a {\em compact} set $\varphi(U_1)\times K$, where $U_1$ is a compact set of non-empty interior $\Int U_1\ni x_0$. 

Let then
\[
s\equiv\sup_{(x^i,\gamma_{ij})\in \varphi(U_1)\times K}\big|\overline{{\psi'}\!_{(x^i)}}(\gamma_{ij})\psi_{(x^i)}(\gamma_{ij})\Delta (\gamma_{ij})\big|
\]
and
\[
\Gamma_\R\ni (\gamma_{ij})\mapsto h(\gamma_{ij}):=
\begin{cases}
s & \text{if $(\gamma_{ij})\in K$, }\\
0 & \text{otherwise}
\end{cases}.
\]      
Since the set $K$ is compact and the Lebesgue measure $d\mu_L$ is regular \cite{cohn}, the function $h$ is integrable over $\Gamma_\R$ with respect to $d\mu_L$. Moreover, for every value $(x^i)\in\varphi(U_1)$, $|\overline{{\psi'}\!_{(x^i)}}\psi_{(x^i)} \,\Delta |\leq h$ (recall that the support of $\overline{{\psi'}\!_{(x^i)}}\psi_{(x^i)}$ is contained in $K$). On the other hand, $\overline{{\psi'}\!_{(x^i)}}\psi_{(x^i)}$ converges pointwise to $\overline{{\psi'}\!_{(x^i_0)}}\psi_{(x^i_0)}$ as $(x^i)\to (x^i_0)=\varphi(x_0)$.    

These three facts allows us to apply  the Lebesgue's dominated convergence theorem to conclude that $f$ given by \eqref{f-ppC}, is continuous at $(x^i_0)=\varphi(x_0)$. But $x_0$ is an arbitrary point, and $(x^i)$ an arbitrary local coordinate system. Thus $\scl{\Psit'}{\Psit}$ is a continuous scalar density on $\Mc$.       
\end{proof}

\subsection{The Hilbert space $\Hc_1$}

Let $\Hc_1^c$ be the set of all continuous Hilbert half-densities on $\Mc$ of compact $\Mc$-support and of compact and slowly changing $\Gamma_\R$-support. Any (finite) linear combination of elements of $\Hc^c_1$ is again a Hilbert half-density of compact $\Mc$-support. This fact and Lemma \ref{PP-lin-comb} guarantee that $\Hc_1^c$ is a complex vector space. By virtue of Lemma \ref{PP-cont} for any two elements $\Psit$ and $\Psit'$ of $\Hc_1^c$, the scalar density $\scl{\Psit'}{\Psit}$ on $\Mc$ is continuous. This density is also compactly supported and therefore it can be naturally integrated\footnote{Let us emphasize that the assumption that the $\Gamma_\R$-support of each element of $\Hc^c_1$ is slowly changing, is essential here---without it one can get a non-integrable scalar density as seen in the following example. In the case $\Mc=\R$ and the signature $(1,0)$ the set $\Gamma_\R$ is the set $\R_+$ of all positive real numbers. Let $x^1$ be the canonical coordinate on $\Mc=\R$ and $A:=\{\ (x^1,\gamma_{11})\in\R\times\Gamma_\R\ | \ x^1\geq 0,\ \gamma_{11}-1 \geq 0,\ 1-x^1\gamma_{11} \geq 0 \ \}$. Define a function $\psi$ on $\R\times\Gamma_\R$ as follows: $\psi$ is zero outside $A$ and $\psi(x^1,\gamma_{11}):=\sqrt{x^1}(\gamma_{11}-1)(1-x^1\gamma_{11})$ on $A$. It is clear that $\psi$ is continuous and for every $x^1\in\R$, $\supp\psi_{(x^1)}$ is compact. But for every open neighborhood $U$ of $x^1=0$, $\bigcup_{x^i\in U}\supp\psi_{(x^1)}=[1,\infty[$. Thus the Hilbert half-density $\Psit$ on $\Mc$ obtained from $\psi$ by means of the reconstruction formula \eqref{P-from-p}, is continuous and of compact $\Gamma_\R$-support everywhere on $\Mc$ (and of compact $\Mc$-support), but the $\Gamma_\R$-support is not slowly changing around $x^1=0$. The measure $\Delta d\mu_L$ on $\Gamma_\R$ is here of the form $(\gamma_{11})^{-1}d\gamma_{11}$ \cite{qs-metr}. Using this fact and \eqref{f-ppC} it is not difficult to realize  that the coordinate representation of $\scl{\Psit}{\Psit}$ in the coordinate $x^1$, diverges to infinity as $(x^1)^{-1}$, as $x^1$ goes to zero from the right. This means that the density $\scl{\Psit}{\Psit}$ is not continuous at the point $x^1=0$ and is also non-integrable over $\Mc$.} over $\Mc$ being a paracompact manifold. The following map
\begin{equation}
\Hc_1^c\times \Hc_1^c \ni (\Psit',\Psit)\mapsto \scal{\Psit'}{\Psit}:=\int_\Mc \scl{\Psit'}{\Psit}\in\C,
\label{H0-scal}
\end{equation}
where the integral at the r.h.s. is the integral of the scalar density $\scl{\Psit'}{\Psit}$, is an inner product on $\Hc_1^c$.

Indeed, it is clear that the map \eqref{H0-scal} is linear in the second argument and that it satisfies the Hermitian (or conjugate) symmetry condition. For every $\Psit\in\Hc^c_1$, the scalar density $\scl{\Psit}{\Psit}$ is continuous and non-negative i.e., for every $x\in \Mc$, $\scl{\Psit}{\Psit}(x)\geq 0$ (see the formula \eqref{w'>w}). Therefore $\scal{\Psit}{\Psit}\geq 0$. Suppose that $\Psit(x)\neq 0$ for a point $x\in\Mc$. Then, by continuity of $\scl{\Psit}{\Psit}$ its support contains a non-empty open set. Consequently, $\scal{\Psit}{\Psit}>0$ and \eqref{H0-scal} is positive definite. 

The completion $\Hc_1$ of $\Hc_1^c$ in the norm induced by the inner product \eqref{H0-scal}, is a Hilbert space built over the set $\Qc(\Mc)$. 

{Let us end this section by a remark concerning the standard Hilbert space $\Hc_{QM}$ of the ordinary quantum mechanics. This space is usually defined as $L^2(\R^3,d\mu_L)$, where $d\mu_L$ is the Lebesgue measure. Note, however, that wave functions in $\Hc_{QM}$ can be viewed as Hilbert half-densities of special sort defined on $\R^3$. Indeed, the set $\C$ of complex numbers equipped with the map
\[
\C^2\ni(z,z')\mapsto \bar{z}z'\in\C
\]       
is a one-dimensional complex Hilbert space. If $\Ct^{1/2}_x$, ($x\in\R^3$), is the set of all half-densities over $T_x\R^3$ valued in the Hilbert space $\C$, then a section of $\bigcup_{x\in\R^3}\Ct^{1/2}_x$ is a sort of a Hilbert half-density. It is clear that each wave function in $\Hc_{QM}$, if expressed in a Cartesian coordinate system on $\R^3$, can be understood as a coordinate representation of such a half-density. It is also not difficult to convince oneself that the inner product of two wave functions in $\Hc_{QM}$, can be expressed in terms of \emi pairing of corresponding Hilbert half-densities to a scalar density and \emii integrating of the resulting scalar density over $\R^3$. Thus, counterintuitively, the Hilbert spaces $\Hc_1$ and $\Hc_{QM}$ are fairly similar.}

\subsection{Uniqueness of $\Hc_1$}

Given manifold $\Mc$ and signature $(p,p')$, the only choice we have to make in order to obtain the Hilbert space $\Hc_1$, is the choice of a diffeomorphism invariant field \eqref{diff-inv} of invariant measures. However, since all such fields are unique up to a positive multiplicative constant (see Equation \eqref{mu-cmu}), the freedom to choose the measure field is actually not relevant. 

Indeed, if $x\mapsto d\mu_{x}$ and $x\mapsto d\check{\mu}_{x}$ are two such measure fields on $\Mc$, and $\Hc_1$ and $\check{\Hc}_1$ the resulting Hilbert spaces, then it follows from \eqref{mu-cmu} that
\begin{equation}
{\Hc_1}\ni\Psit\mapsto\frac{\Psit}{\sqrt{c}}\in\check{\Hc}_1 
\label{U-c}
\end{equation}
is a unitary map.

Thus all Hilbert spaces constructed according to the prescription presented in Section \ref{constr-H} are isomorphic. Moreover, there exists a distinguished or natural isomorphism \eqref{U-c} between each pair of such Hilbert spaces.

We conclude then that the Hilbert space $\Hc_1$ is unique up to natural isomorphisms.

\section{Action of diffeomorphisms on the Hilbert space $\Hc_1$}

\subsection{Action of diffeomorphisms on scalar densities}

Let $\Ft$ be a scalar density on $\Mc$, and $\theta:\Mc\to\Mc$ a diffeomorphism. The diffeomorphism $\theta$ acts on the density $\Ft$ by means of the following pull-back:
\begin{equation}
(\theta^*\Ft)(x,e_x):=\Ft\big(\theta(x),\theta^t e_x\big),
\label{chi-F}
\end{equation}
where $e_x\equiv(e_{xi})_{i=1,\ldots,\dim\Mc}$ is a basis of $T_x\Mc$, $\theta^t:T_x\Mc\to T_{\theta(x)}\Mc$ is the tangent map generated by $\theta$ and 
\[
\theta^t e_x\equiv\big(\theta^t(e_{xi})\big)
\]
is a basis of $T_{\theta(x)}\Mc$.

The pull-back $\theta^*\Ft$ is again a scalar density on $\Mc$. To see this let us calculate
\[
(\theta^*\Ft)(x,\Lambda e_x)=\Ft\big(\theta(x),\theta^t\Lambda e_x\big),
\]
where $\Lambda\equiv(\Lambda^j{}_i)$ is any non-singular matrix. If $e_x=(e_{xi})$, then $\Lambda e_x=(\Lambda^j{}_ie_{xj})$. Hence by virtue of linearity of $\theta^t$
\begin{equation}
\theta^t\Lambda e_x =\big(\theta^t(\Lambda^j{}_ie_{xj})\big)=\big(\Lambda^j{}_i\theta^t(e_{xj})\big)=\Lambda\big(\theta^t(e_{xj})\big)=\Lambda\theta^te_x.
\label{chi'-L}
\end{equation}
Consequently,
\[
(\theta^*\Ft)(x,\Lambda e_x)=\Ft(\theta(x),\Lambda \theta^te_x)=|\det\Lambda|\Ft(\theta(x),\theta^t e_x)=|\det\Lambda|(\theta^*\Ft)(x,e_x).
\]      

Let $\Ft$ be a scalar density integrable over $\Mc$. Then for every diffeomorphism $\theta$ of $\Mc$ \cite{heat}\footnote{In \cite{heat} scalar densities are defined in a different way than in the present paper, but it is not difficult to realize that both definitions are equivalent.}
\begin{equation}
\int_\Mc \theta^*\Ft=\int_\Mc \Ft.
\label{int-chi-F}
\end{equation}

\subsection{Action of diffeomorphisms on $\Hc_1^c$}

Let us consider a Hilbert half-density $\Psit$ being an element of $\Hc_1^c$. If $\theta$ is a diffeomorphism of $\Mc$, $\theta^t:T_x\Mc\to T_{\theta(x)}\Mc$ the corresponding tangent map, and $e_x$ a basis of $T_x\Mc$, then 
\begin{equation}
\Psit(\theta(x),\theta^t e_x)
\label{Psi-chi}
\end{equation}
is an element of the Hilbert space $H_{\theta(x)}$, that is, an equivalence class of functions on $\Gamma_{\theta(x)}$.
Since $\Psit\in\Hc_1^c$, the equivalence class \eqref{Psi-chi} can be represented by a unique continuous compactly supported function on $\Gamma_{\theta(x)}$ (see Lemma \ref{lm-psi-zero}).
 
If $(\theta^{-1})^t:T_{\theta(x)}\Mc\to T_x\Mc$ is the inverse of $\theta^t$, then the pull-back $(\theta^{-1})^{t *}:\Gamma_x\to\Gamma_{\theta(x)}$ is a diffeomorphism. Thus $(\theta^{-1})^{t *}$ can be used to pull-back functions on $\Gamma_{\theta(x)}$ to ones on $\Gamma_x$. 

To pull-back the equivalence class \eqref{Psi-chi} of functions on $\Gamma_{\theta(x)}$ to an equivalence class of functions on $\Gamma_x$ being an element of $H_x$, we will proceed as follows. First we will pull-back by means of $(\theta^{-1})^{t *}$ the unique continuous representative of $\eqref{Psi-chi}$, obtaining thereby  a continuous compactly supported function on $\Gamma_x$. Then we will find an element of $H_x$ defined by this resulting function, denote it by\footnote{The standard symbol for the pull-back of a metric $q$ under a diffeomorphism $\theta$ is $\theta^*q$. To be consistent with the standard notation we should use in the formula \eqref{chi'**Psi} the simpler symbol $(\theta^{-1})^{ *\star}$ instead of $(\theta^{-1})^{t *\star}$. However, here we work with a multi-level construction: the baseline level is the manifold $\Mc$, and the next levels are in turn: the tangent space $T_x\Mc$, the space $\Gamma_x$ of scalar products on the tangent space and finally the functions on $\Gamma_x$ constituting $H_x$. In the symbol $(\theta^{-1})^{t *\star}$ and the like, each superscript $t$, $*$ and $\star$ corresponds to a level above the baseline one and makes it easier to keep track where we are.} 
\begin{equation}
(\theta^{-1})^{t *\star}\,\Psit(\theta(x),\theta^t e_x)
\label{chi'**Psi}
\end{equation}
and treat it as the desired pull-back\footnote{Applying this procedure we avoid to prove that if $\int_{\Gamma_{\theta(x)}}\bar{\Psi}\Psi\,d\mu_{\theta(x)}=0$, then $\int_{\Gamma_{x}}(\theta^{-1})^{t *\star}(\bar{\Psi}\Psi)\,d\mu_{x}=0$.} of \eqref{Psi-chi}.

Using \eqref{chi'-L} and linearity of the pull-back $(\theta^{-1})^{t *\star}$ one easily shows that
\[
(\theta^{-1})^{t *\star}\,\Psit(\theta(x),\theta^t \Lambda e_x)=|\det\Lambda|^{1/2}(\theta^{-1})^{t *\star}\,\Psit(\theta(x),\theta^t e_x),
\] 
which means that the following map
\[
e_x\to (\theta^{-1})^{t *\star}\,\Psit(\theta(x),\theta^t e_x)
\]
defined on the set of all bases of $T_x\Mc$ and valued in $H_x$ is a half-density.  

We thus see that on the manifold $\Mc$ there exists a Hilbert half-density $\theta^*\Psit$, such that for every $x\in\Mc$ and for every basis $e_x$ of $T_x\Mc$,      
\begin{equation}
(\theta^*\Psit)(x,e_x)=(\theta^{-1})^{t *\star}\,\Psit(\theta(x),\theta^t e_x).
\label{chi*-P}
\end{equation}
We will say that $\theta^*\Psit$ is {\em pull-back of $\Psit$ under the diffeomorphisms $\theta$}.  

\begin{lm}
The Hilbert half-density $\theta^*\Psit$ is an element of $\Hc_1^c$.   
\label{pull-H^c}
\end{lm}

\begin{proof}
We have to show that $\theta^*\Psit$ is \emi of compact $\Mc$-support, \emii continuous and \emiii of compact and slowly changing $\Gamma_\R$-support. Let us recall that defining $\theta^*\Psit$ we assumed that $\Psit\in\Hc_1^c$.

Regarding \emi: if $\supp \Psit$ is the $\Mc$-support of $\Psit$, then $\theta^{-1}(\supp \Psit)$ is the $\Mc$-support of $\theta^*\Psit$. Since $\supp \Psit$ is compact and $\theta^{-1}$ a continuous map,  $\theta^{-1}(\supp \Psit)$ is also compact.

Regarding \emii: suppose that $U$ is an open subset of $\Mc$ and that a map $\varphi:U\mapsto \R^{\dim\Mc}$ defines a coordinate system $(x^i)$ on $U$. Then the diffeomorphism $\theta$ can be used to pull-back $(x^i)$ to a coordinate system  $(\bar{x}^i)$ on $\theta^{-1}(U)$ defined by the map $\varphi\circ\theta:\theta^{-1}(U)\to\R^{\dim\Mc}$. 

Now let us find a relation between the (continuous) coordinate representation $\psi$ of $\Psit$ in the system $(x^i)$ (see \eqref{psi-cont}), and a coordinate representation $\theta^*\psi$ of $\theta^*\Psit$ defined on the set $(\varphi\circ\theta)(\theta^{-1}(U))\times\Gamma_\R$ by the system $(\bar{x}^i)$.       

Let $e_x$ be a basis of $T_x\Mc$. It follows from \eqref{chi*-P} and the procedure, which defines \eqref{chi'**Psi}, that $(\theta^*\Psit)(x,e_x)\in H_x$ contains a {continuous} representative such that for every $\gamma_x\in\Gamma_x$ its value reads
\begin{equation}
(\theta^*\Psit)(x,e_x,\gamma_x)=\Psit\big(\theta(x),\theta^t e_x, (\theta^{-1})^{t *}\gamma_x\big),
\label{chiPP}
\end{equation}
where the number at the r.h.s. is a value of the {\em continuous} representative of $\Psit(\theta(x),\theta^t e_x)$. 

Assume now that $x\in \theta^{-1}(U)$. Then $x=(\theta^{-1}\circ\varphi^{-1})(\bar{x}^i)$ for some value $(\bar{x}^i)\in\varphi(U)$ and if the basis $e_x=(\partial_{\bar{x}^i})$, then $\theta^t e_x=(\partial_{x^i})$. Moreover, 
\[
\big((\theta^{-1})^{t *}\gamma_x\big)(\partial_{x^i},\partial_{x^j})=\gamma_x\big((\theta^{-1})^t \partial_{x^i},(\theta^{-1})^{t}\partial_{x^j}\big)=\gamma_x(\partial_{\bar{x}^i},\partial_{\bar{x}^j})
\]
and, consequently, if
\[
\gamma_x=\bar{\gamma}_{ij}\,d\bar{x}^i\ot d\bar{x}^j,
\]
then
\[
(\theta^{-1})^{t *}\gamma_x=\bar{\gamma}_{ij}\,d{x}^i\ot d{x}^j.
\]

Using all these results we can transform \eqref{chiPP} obtaining
\[
(\theta^*\Psit)\big((\theta^{-1}\circ\varphi^{-1})(\bar{x}^i),(\partial_{\bar{x}^i}),\bar{\gamma}_{ij}\,d\bar{x}^i\ot d\bar{x}^j\big)=\Psit\big(\varphi^{-1}(\bar{x}^i),(\partial_{x^i}), \bar{\gamma}_{ij}\,d{x}^i\ot d{x}^j\big).
\]
This equality together with \eqref{psi-cont} mean that on $\varphi(U)\times \Gamma_\R=(\varphi\circ\theta)(\theta^{-1}(U))\times \Gamma_\R$ 
\begin{equation}
(\theta^*\psi)(\bar{x}^i,\bar{\gamma}_{ij})=\psi(\bar{x}^i,\bar{\gamma}_{ij}).
\label{chip-p}
\end{equation}
Note that $\psi$ at the r.h.s. of the equation above, is the {continuous} coordinate representation of $\Psit$ in the system $(x^i)$---to be continuous, for every value $(\bar{x}^i)\in\varphi(U)$ the representation must come  from the unique continuous representative of $\Psit(\varphi^{-1}(\bar{x}^i),(\partial_{x^i}))$ and, indeed, $\psi$ on $\varphi(U)\times\Gamma_\R$ does come from these representatives (see the remark just below Equation \eqref{chiPP}).

Now an immediate implication of Equation \eqref{chip-p} is that the coordinate representation  $\theta^*\psi$ of $\theta^*\Psit$ in the system $(\bar{x}^i)$ is continuous---its continuity follows obviously from continuity of $\psi$. This is sufficient for $\theta^*\Psit$ to be continuous, since $(\bar{x}^i)$ is an arbitrary local coordinate system on $\Mc$.

Regarding \emiii: let us consider now a point $\theta(x_0)\in U$. Let $U_0\subset U$ be an open neighborhood of $\theta(x_0)$ and $K$ a compact subset of $\Gamma_\R$ such that for every value $(\bar{x}^i)\in \varphi(U_0)$, the support of the function $\psi_{(\bar{x}^i)}$ related to $\psi$  via the formula \eqref{psi_xi}, is contained in $K$---the existence of $U_0$ and $K$ follows from the assumption, that $\Psit$ is an element of $\Hc_1^c$. 

By virtue of \eqref{chip-p}, for every value $(\bar{x}^i)\in (\varphi\circ\theta)(\theta^{-1}(U_0))$, the support of the function $(\theta^*\psi)_{(\bar{x}^i)}$ related to $\theta^*\psi$  via the formula \eqref{psi_xi}, is contained in $K$. This means that the $\Gamma_\R$-support of the pull-back $\theta^*\Psit$ around $x_0$ is compact and slowly changing in the coordinate system $(\bar{x}^i)$. But $x_0$ is an arbitrary point of $\Mc$ and $(\bar{x}^i)$ an arbitrary local coordinate system on the manifold. Therefore $\theta^*\Psit$ is a Hilbert half-density on $\Mc$ of compact and slowly changing $\Gamma_\R$-support.
\end{proof}

\subsection{Unitary representation of diffeomorphisms on $\Hc_1$}

\begin{lm}
Suppose that $\Psit,\Psit'\in\Hc_1^c$ and $\theta$ is a diffeomorphism of $\Mc$. Then the pull-back $\theta^*$ preserves the inner product \eqref{H0-scal} on $\Hc_1^c$:  
\[
\scal{\theta^*\Psit'}{\theta^*\Psit}=\scal{\Psit'}{\Psit}.
\]    
\label{in-prod-pres}
\end{lm}

\begin{proof}
Consider first the scalar density $\scl{\theta^*\Psit'}{\theta^*\Psit}$ (see \eqref{hden-pair} for the definition). If $e_x$ is a basis of $T_x\Mc$, then
\begin{multline*}
\scl{\theta^*\Psit'}{\theta^*\Psit}(x,e_x)=
\scal{(\theta^*\Psit')(x,e_x)}{(\theta^*\Psit)(x,e_x)}_x=\\=\big\langle\,(\theta^{-1})^{t *\star}\,\Psit'(\theta(x),\theta^t e_x)\,\big|\,(\theta^{-1})^{t *\star}\,\Psit(\theta(x),\theta^t e_x)\,\big\rangle_x=\\=\int_{\Gamma_x}(\theta^{-1})^{t *\star}\big(\overline{\Psit'(\theta(x),\theta^t e_x)}\Psit(\theta(x),\theta^t e_x)\big)\,d\mu_x=\\=\int_{\Gamma_{\theta(x)}}\big(\overline{\Psit'(\theta(x),\theta^t e_x)}\Psit(\theta(x),\theta^t e_x)\big)\,\big((\theta^{-1})^{t *}\big)_\star d\mu_x=\\=\int_{\Gamma_{\theta(x)}}\big(\overline{\Psit'(\theta(x),\theta^t e_x)}\Psit(\theta(x),\theta^t e_x)\big)\,d\mu_{\theta(x)}
\end{multline*}
---here we used in turn: in the first step the definition \eqref{den-prod}, in the second step Equation \eqref{chi*-P}, in the third step we chosen the continuous (compactly supported) representatives of $\Psit'(\theta(x),\theta^t (e_x))$ and $\Psit(\theta(x),\theta^t (e_x))$, in the forth step we applied \eqref{*dmu}, finally in the last step we used the diffeomorphism invariance of the measure field $x\mapsto d\mu_x$ (see Equation \eqref{l*xx}). Thus
\[
\scl{\theta^*\Psit'}{\theta^*\Psit}(x,e_x)=\scal{\Psit'(\theta(x),\theta^t e_x)}{\Psit(\theta(x),\theta^t e_x)}_{\theta(x)}=\scl{\Psit'}{\Psit}(\theta(x),\theta^t e_x).
\]
Comparing this with \eqref{chi-F} we see that
\[
\scl{\theta^*\Psit'}{\theta^*\Psit}=\theta^*\scl{\Psit'}{\Psit}.
\]
By virtue of this result and Equation \eqref{int-chi-F} 
\[
\scal{\theta^*\Psit'}{\theta^*\Psit}=\int_\Mc \scl{\theta^*\Psit'}{\theta^*\Psit}=\int_\Mc\theta^*\scl{\Psit'}{\Psit}=\int_\Mc\scl{\Psit'}{\Psit}=\scal{\Psit'}{\Psit}.
\]
\end{proof}

Let us define an operator on $\Hc_1^c$:
\[
\Hc_1^c\ni \Psit\mapsto u_\theta\Psit:=(\theta^{-1})^*\Psit\in\Hc_1^c.
\] 
It follows from \eqref{chiPP} that $u_\theta$ is linear. Manifestly, for every $\Psit\in\Hc^c_1$  
\[
\Psi=u_\theta(\theta^*\Psi), 
\]
which means that $u_\theta$ is surjective. Lemma \ref{in-prod-pres} guarantees that $u_\theta$ preserves the inner product on $\Hc_1^c$ being a dense linear subspace of $\Hc_1$. Taking into account  all these facts, we see that the operator $u_\theta$ can be uniquely extended to a unitary operator $U_1(\theta)$ on $\Hc_1$.

It is not difficult to check that for two diffeomorphisms $\theta_1$ and $\theta_2$,
\[
U_1({\theta_1})\circ U_1({\theta_2})=U_1({\theta_1\circ\theta_2}).
\]          
Thus
\begin{equation}
\theta\mapsto U_1(\theta)
\label{diff-repr-U1}
\end{equation}
is a unitary representation of the group $\Diff(\Mc)$ of all diffeomorphisms of $\Mc$ on the Hilbert space $\Hc_1$.

\section{Construction of the Hilbert space $\Hf$}

\subsection{Motivation}

Suppose that the set $\Qc(\Mc)$ is non-empty. Then the following surjective \cite{qs-metr} map
\begin{equation}
\Qc(\Mc)\ni q\mapsto \kappa_x(q):=q_x\in\Gamma_x,
\label{dof}
\end{equation}
where $q_x$ is the value of the metric $q$ at $x\in\Mc$, can be treated as a degree of freedom (d.o.f.) on $\Qc(\Mc)$. Consequently, the Hilbert space $H_x$ given by \eqref{Hx} can be treated as a quantum counterpart of $\kappa_x$.

Evidently, for each pair $x\neq x'$, $\kappa_x$ and $\kappa_{x'}$ are independent d.o.f.. It seems therefore, that a Hilbert space being a quantum counterpart of the configuration space $\Qc(\Mc)$ should contain {\em tensor products} of the Hilbert spaces $\{H_x\}$. 

It is easy to realize that the structure of $\Hc_1$ does not meet this expectation. Indeed, the inner product \eqref{H0-scal} is an integral, that is, an ``uncountable sum'' of values of inner products on the Hilbert spaces $\{H_x\}$. Therefore $\Hc_1$ is more like a direct integral of Hilbert spaces \cite{g-hilb}:
\[
\int^\oplus_\Mc H_x.
\]  
This fact suggests that $\Hc_1$ (in the case of signature $(3,0)$) is not well suited for quantization of the ADM formalism. Fortunately, it is relatively easy to find a way around this problem.

Namely, consider the set $\Mc_N$ of all $N$-element subsets of $\Mc$:
\begin{equation}
\Mc_N:=\big\{ \ \{x_1,\ldots,x_N\}\equiv \{x_K\}\subset \Mc \ | \ \text{$x_I\neq x_J$ for $I\neq J$} \ \big\}.
\label{N_N-df}
\end{equation}
It is possible to define on $\Mc_N$ a differential structure in such a way that the set becomes a smooth paracompact manifold locally diffeomorphic to $\Mc^N$---for details see Appendix \ref{N_N}. We will associate with each point $y=\{x_K\}$ of $\Mc_N$ a Hilbert space $H^\ot_{y}$ naturally isomorphic to $H_{x_1}\ot\ldots \ot H_{x_N}$ (for every ordering of the factors in the tensor product),  
\begin{equation}
H^\ot_{y}\cong H_{x_1}\ot\ldots \ot H_{x_N},
\label{H^ot_y}
\end{equation}
and then construct a Hilbert space $\Hc_N$ using half-densities on $\Mc_N$ valued in the spaces $\{H^\ot_y\}$.

Next, we will merge all the spaces $\{\Hc_N\}$ into the desired Hilbert space $\Hf$. To do this we will follow some feature of states in the kinematical Hilbert space $H_{LQG}$ of Loop Quantum Gravity (see e.g. \cite{rev}): if a spin-network state $\Psi\in H_{LQG}$ depends non-trivially on a classical d.o.f., a spin-network state $\Psi'\in H_{LQG}$ is independent of, then $\Psi$ and $\Psi'$ are orthogonal. Note now, that if $N<N'$ and $H_y^\otimes$ and $H_{y'}^\otimes$ are the Hilbert spaces \eqref{H^ot_y} associated with, respectively, $y\in\Mc_N$ and $y'\in \Mc_{N'}$, then wave functions in $H^\otimes_{y'}$ depend on certain d.o.f. $\kappa_x$, wave functions in $H^\otimes_{y}$ are independent of. Therefore we would like $H_y^\otimes$ and $H_{y'}^\otimes$  to be orthogonal as building blocks of $\Hf$. To achieve this goal, we will define $\Hf$ as the orthogonal sum \begin{equation}
\mathfrak{H}:=\bigoplus_{N=1}^\infty \Hc_N.
\label{Hf-df}
\end{equation}
Note also that the structure of this $\Hf$ will resemble to a certain degree the structure of Fock spaces.

Obviously, the Hilbert space $\Hf$ will contain all finite tensor products of the Hilbert spaces $\{H_x\}$. Therefore the space $\mathfrak{H}$ constructed in the case of signature $(3,0)$, seems to be better suited for quantization of the ADM formalism than $\Hc_1$ alone.

\subsection{Remarks on the definition \eqref{Hf-df} of $\Hf$ }

A rigorous construction of $\Hc_N$, we are going to present below, is fairly long and technically involved (this concerns in particular the construction of the smooth atlas on the set $\Mc_N$). Let us note that an other method to take into account the tensor products of $\{H_x\}$, is to use in \eqref{Hf-df} the tensor product $\bigotimes^N \Hc_1$ of $N$ copies of $\Hc_1$, instead of $\Hc_N$. This may seem to be the simplest way to achieve the goal, which does not require any extra effort, but a closer look at this construction makes it clear that it is not the case.

Namely, it is easy to realize that, given $\{x_K\}\in\Mc_N$, one can find $N!$ tensor products of Hilbert spaces $\{H_{x_1},\ldots,H_{x_N}\}$ in the space $\bigotimes^N \Hc_1$---these tensor products differ from each other by the ordering of the factors and are given by all possible orderings. From a physical point of view the ordering of the factors in $H_{x_1}\ot\ldots\ot H_{x_N}$ is irrelevant---each such tensor product describes a space of quantum states of the same quantum system, obtained by quantization of a classical system, whose configuration space is given by the collection $\{\kappa_{x_1},\ldots,\kappa_{x_N}\}$ of independent classical d.o.f.. 

Consequently, a generic state in $\bigotimes^N \Hc_1$ seems to correspond to $N!$ distinct states in $\Hc_N$. It can be then expected that working with such generic states would be rather cumbersome. To avoid this, one would have to impose on elements of $\bigotimes^N \Hc_1$ some restrictions in order to isolate those states, which would correspond to single states in $\Hc_N$.

Moreover, for many purposes it would be convenient or even necessary to describe elements of $\bigotimes^N \Hc_1$ as fields on $\Mc^N$ (preferably, as half-densities on $\Mc^N$ valued in the tensor products of $\{H_x\}$) and to impose some regularity conditions on these fields. Therefore, in order to put the space $\bigotimes^N \Hc_1$ into practice, one would have to elaborate such a description. 

Thus the application of $\bigotimes^N \Hc_1$ does require some additional effort, which makes this space not as attractive as it seemed to be at the very beginning.          

In our opinion the advantage of $\Hc_N$ over $\bigotimes^N \Hc_1$ is conceptual simplicity of $\Hc_N$: the idea of its construction is simple and natural, and complies with the fact that the ordering of factors in $H_{x_1}\ot\ldots\ot H_{x_N}$ is physically irrelevant. Consequently, each state in $\Hc_N$ fits well our need to take into account the tensor products of $\{H_x\}$, without the necessity to impose on it any extra conditions.

Moreover, once the space $\Hc_N$ is rigorously built, one can work with it without the need to refer to many technical details of its construction like e.g. the construction of the smooth atlas on $\Mc_N$. For convenience of the reader not interested in such details, we will place a considerable part of these technicalities in the appendix to this paper.

\subsection{Construction of the Hilbert space $\Hc_N$}

The construction of the Hilbert space $\Hc_N$ will follow as closely as possible the construction of the Hilbert space $\Hc_1$ described in Section \ref{constr-H}.  

Let $(p,p')$ be the signature of metrics on $\Mc$ fixed at the very beginning of Section \ref{constr-H} for the sake of the construction of the Hilbert space $\Hc_1$. Let us fix additionally a natural number $N\geq 2$.

\subsubsection{Preliminaries \label{H_N-prel}}

 Suppose that $V^\oplus$ is a real vector space of dimension $N(p+p')$ and that a decomposition 
\begin{equation}
V^\oplus=\bigoplus_{I=1}^N V_I
\label{V^o+}
\end{equation}
is given, such that each $V_I$ is a linear subspace of $V^\oplus$ of dimension $p+p'$. 

Denote by $\Gamma_I$ the homogeneous space of all scalar products on $V_I$ of signature $(p,p')$. For each $I\in\{1,\ldots,N\}$ let us choose $\gamma_I\in\Gamma_I$ and define
\begin{equation}
V^\oplus\times V^\oplus\ni(v,\check{v})\mapsto \gamma^\oplus(v,\check{v}):=\sum_{I=1}^N\gamma_I(v_I,\check{v}_I)\in \R,
\label{g^o+}
\end{equation}
where  
\begin{align*}
v&=\sum_{I=1}^N v_I, & \check{v}&=\sum_{I=1}^N \check{v}_I, & v_I,\check{v}_I\in V_I.
\end{align*}   
Clearly, $\gamma^\oplus$ is a scalar product on $V^\oplus$ of signature $(Np,Np')$. We will use the symbol $\Gamma^\oplus$ to represent the set of all scalar products on $V^\oplus$ of the form \eqref{g^o+} (with the fixed decomposition \eqref{V^o+}). Let us emphasize that $\Gamma^\oplus$ is a {\em proper} subset of the set of all scalar products on $V^\oplus$ of signature $(Np,Np')$.   

Note that one can assign to $\gamma^\oplus$ a sequence of the scalar products used to define $\gamma^\oplus$ via \eqref{g^o+}. Obviously, this assignment, 
\begin{equation}
\Gamma^\oplus\ni\gamma^\oplus\mapsto b(\gamma^\oplus):=(\gamma_{1},\ldots,\gamma_{N})\in \Gamma_{1}\times\ldots\times \Gamma_{N}\equiv \bigtimes \Gamma_K,
\label{G^o+-GGG}
\end{equation}
is a bijection, which can be used to induce some structures on $\Gamma^\oplus$.

First, the bijection together with charts $\{(\Gamma_I,\chi_I)\}$ given by \eqref{map-lin-coor}, allow us to construct a map 
\begin{equation}
\Gamma^\oplus\ni\gamma^\oplus\mapsto(\chi_1\times\ldots\times\chi_N)(b(\gamma^\oplus))\in \Gamma^N_\R\subset\R^{N\dim\Gamma_I}
\label{o+-lin-coor}
\end{equation}
which define a global coordinate system on $\Gamma^\oplus$. This map can be used to ``pull-back'' the topology from $\R^{N\dim\Gamma_I}$ onto $\Gamma^\oplus$. Obviously, coordinate systems given by all maps \eqref{o+-lin-coor} form an analytic atlas on $\Gamma^\oplus$.

Suppose that $d\mu_I$ is an invariant measure on $\Gamma_I$. Since $d\mu_I$ is $\sigma$-finite (see Section \ref{inv-dmu}), the product
\[
d\mu_1\times\ldots \times d\mu_N\equiv \bigtimes d\mu_K
\]
is well-defined. This product is also a regular Borel measure on the Cartesian product $\Gamma_{1}\times\ldots\times \Gamma_{N}$---this is because \cite{cohn} each $d\mu_I$ is a regular Borel measure on a second countable l.c.H.  space $\Gamma_I$. We can use the bijection \eqref{G^o+-GGG} to push-forward this measure obtaining thereby a (regular Borel) measure on $\Gamma^\oplus$, which will be denoted by $d\mu^\times$: 
\[
d\mu^\times:=(b^{-1})_\star \big(\bigtimes d\mu_K\big).
\] 
Let
\[
H^\ot:=L^2(\Gamma^\oplus,d\mu^\times)\cong L^2(\Gamma_{1}\times\ldots\times \Gamma_{N},d\mu_1\times\ldots \times d\mu_N).
\]
Recall that each $L^2(\Gamma_I,d\mu_I)$ is separable \cite{qs-metr} and each measure $d\mu_I$ is $\sigma$-finite.  These properties of the space and the measure guarantee \cite{spectr-qm} that 
\begin{equation}
H^\ot\cong L^2(\Gamma_1,d\mu_1)\ot\ldots\ot L^2(\Gamma_N,d\mu_N).
\label{H-LLL}
\end{equation}

Let us emphasize that the Cartesian product $\Gamma_{1}\times\ldots\times \Gamma_{N}$, unlike the construction of $\Gamma^\oplus$, does require the spaces $\{\Gamma_I\}$ to be ordered. However, neither the topology nor the differential structure nor the measure $d\mu^\times$ induced on $\Gamma^\oplus$ depends on the ordering. Thus the Hilbert space $H^\ot$ is also independent of the ordering. 

Finally, let us state a fact, which concerns a push-forward of a product of measures:
\begin{lm}
Let $X$, $Y$, $X'$ and $Y'$ be second countable l.c.H. spaces and let $\alpha:X'\to X$, and $\beta:Y'\to Y$ be homeomorphisms. If $d\mu$ and $d\nu$ are regular Borel measures on, respectively, $X'$ and $Y'$, then
\begin{equation}
(\alpha_{\star}d\mu)\times(\beta_{\star}d\nu)=(\alpha\times\beta)_{\star}(d\mu\times d\nu).
\label{mu-x-nu}
\end{equation}
\label{lm-push-prod}   
\end{lm}
\noindent This lemma may seem to be obvious, but a strict proof of it, done in line with the definition \eqref{*dmu} of push-forward measure, requires some effort. Therefore we relegate the proof to Appendix \ref{app-proof}.
 
The lemma above can be easily generalized to a product of any finite number of suitable measures since \cite{cohn} \emi a product $X\times Y$ of second countable l.c.H. spaces $X$ and $Y$  is such a space again and \emii if $d\mu$ and $d\nu$ are regular Borel measures on, respectively, $X$ and $Y$, then $d\mu\times d\nu$ is such a measure on $X\times Y$.

Let us now apply Lemma \ref{lm-push-prod} to express an integral over $\Gamma^\oplus$ with respect to $d\mu^\times$  in terms of an integral over $\Gamma^N_\R$. Recall first that each $\Gamma_I$ in \eqref{G^o+-GGG} and $\Gamma_\R$ are second countable l.c.H. spaces. Each invariant measure $d\mu_I$ on $\Gamma_I$ used to define $d\mu^\times$ is regular and Borel. Moreover, $d\mu_I=c_I\,d\mu_{Q_I}$, where $c_I>0$ and $d\mu_{Q_I}$ is the invariant measure on $\Gamma_I$ given by the natural metric on this space (see \eqref{mu-cmuQ}). On the other hand, we can treat the r.h.s. of \eqref{natme} as the definition of a positive functional on $C^c(\Gamma_\R)$, which (by virtue of the Riesz representation theorem) allows us to regard $\Delta \,d\mu_L$ as a regular Borel measure on $\Gamma_\R$. Taking into account all these facts, Equation \eqref{natme} and Lemma \ref{lm-push-prod}, it is straightforward to make the following transformations: 
\begin{multline}
\int_{\Gamma^\oplus}(b^\star \Psi)\,d\mu^\times =\int_{\bigtimes \Gamma_K} \Psi\,\big(\bigtimes d\mu_K\big)=\mathbf{c}\int_{\bigtimes \Gamma_K} \Psi\,\big(\bigtimes d\mu_{Q_K}\big)=\\=\mathbf{c}\int_{\bigtimes \Gamma_K} \Psi\,\big(\bigtimes \chi^{-1}_{K\star}(\Delta  d\mu_{L})\big)=\mathbf{c}\int_{\bigtimes \Gamma_K}  \Psi\,\Big(\big(\bigtimes \chi^{-1}_K\big)_\star\big(\bigtimes \Delta  d\mu_{L}\big)\Big)=\\=\mathbf{c}\int_{\Gamma^N_\R} \Big((\bigtimes \chi^{-1}_K)^\star \Psi\Big)\,\big(\bigtimes \Delta  d\mu_{L}\big), 
\label{int_g^o+}
\end{multline}
where $b^\star \Psi\in C^c(\Gamma^\oplus)$, $\mathbf{c}=c_1\cdot\ldots\cdot c_N$, the product
\[
\bigtimes \chi^{-1}_K\equiv \chi^{-1}_1\times\ldots\times \chi^{-1}_N
\]
is given by the maps appearing in \eqref{o+-lin-coor}, and $\bigtimes \Delta \,d\mu_L$ is the product of $N$ copies of $\Delta \,d\mu_L$.   

Equation \eqref{int_g^o+} can be now used to prove the following generalization of Lemma \ref{lm-psi-zero}:
\begin{lm}
Let $\Psi:\Gamma^\oplus\to \C$ be continuous. If
\[
\int_{\Gamma^\oplus} \bar{\Psi}\Psi\,d\mu^\times = 0,
\]
then $\Psi=0$.  
\label{lm-psi-zero-o+}
\end{lm}

\subsubsection{Hilbert half-densities on $\Mc_N$ \label{Hhd-N_N}}

Here we will apply the construction of the Hilbert space $H^\ot$ just presented to associate the Hilbert space $H^\ot_y$ with every point $y\equiv\{x_K\}$ of the manifold $\Mc_N$ (let us recall that the set $\Mc_N$ is defined by the formula \eqref{N_N-df}, and the smooth atlas on $\Mc_N$ is introduced in Appendix \ref{N_N}). 

As shown in Appendix \ref{app-decomp}, for every $y\in\Mc_N$, there exists a distinguished decomposition of the tangent space $T_{y}\Mc_N$ into a direct sum of (linear subspaces naturally isomorphic to) the tangent spaces $\{T_{x_I}\Mc\}_{x_I\in\, y}$:
\begin{equation}
T_{y}\Mc_N=\bigoplus_{x_K\in \,y} T_{x_K}\Mc.
\label{T-dec}
\end{equation}
Suppose that $\{\gamma_{x_I}\}_{x_I\in \,y}$ is a collection of scalar products of signature $(p,p')$ on, respectively, $\{T_{x_I}\Mc\}$---in other words,  each $\gamma_{x_I}\in \Gamma_{x_I}$. This collection together with the decomposition \eqref{T-dec} define a scalar product $\gamma^\oplus_{ y}$ on $T_{ y}\Mc_N$ of signature $(Np,Np')$ according to the prescription \eqref{g^o+}. We will denote by $\Gamma^\oplus_{ y}$ the set of all such scalar products on $T_{ y}\Mc_N$.      

Let $x\mapsto d\mu_x$ be the diffeomorphism invariant field \eqref{diff-inv} of invariant measures on $\Mc$,  used in Section \ref{Hc-df} to construct the Hilbert space $\Hc_1$. Denote by $d\mu^\times_{ y}$ the measure on $\Gamma^\oplus_{ y}$ defined as the push-forward of the measure
\[
d\mu_{x_1}\times\ldots \times d\mu_{x_N}
\]
given by the inverse of the natural bijection (see \eqref{G^o+-GGG})
\begin{equation}
b_y:\Gamma^\oplus_{y}\to\Gamma_{x_1}\times\ldots\times\Gamma_{x_N}.
\label{1-to-1}
\end{equation}
This allows us to associate with the point $ y$ the following Hilbert space (see \eqref{H-LLL}):
\begin{equation}
H^\ot_{ y}:=L^2(\Gamma^\oplus_{ y},d\mu^\times_{ y})\cong L^2(\Gamma_{x_1},d\mu_{x_1})\ot\ldots\ot L^2(\Gamma_{x_N},d\mu_{x_N})=H_{x_1}\ot\ldots\ot H_{x_N}.
\label{H^o+_y}
\end{equation}
Let us emphasize that the measure $d\mu^\times_{ y}$  does not depend on the choice of the ordering of the spaces $\{\Gamma_{x_K}\}$ in \eqref{1-to-1}. Consequently, the Hilbert space $H^\ot_{y}$ does not distinguish any ordering of the Hilbert spaces $\{H_{x_K}\}$ in \eqref{H^o+_y}.    

\paragraph{Definition} Let $\tilde{H}^\ot_{ y}$ denote the pseudo-Hilbert space of all half-densities over $T_{ y}\Mc_N$ valued in $H^\ot_{ y}$. We will use the symbol $\scl{\cdot}{\cdot}_{ y}$ to represent the density product on $\tilde{H}^\ot_{ y}$. A map $\Psit$ from $\Mc_N$ to 
\begin{equation}
\tilde{\mathbf{H}}^\ot:=\bigcup_{ y\in\Mc_N}\tilde{H}^\ot_{ y}
\label{t-bfH^ot}
\end{equation}
such that $\Psit(y)\in \tilde{H}^\ot_{y}$, will be called Hilbert half-density on $\Mc_N$. Equivalently, one can think of $\Psit$ as of a section of the bundle-like set $\tilde{\mathbf{H}}^\ot$.       

\paragraph{Regularity conditions}

Again, we would like to impose on the Hilbert half-densities just defined, some regularity conditions, which \emi would be helpful while defining physical operators on $\Hc_N$ and \emii  will ensure that the half-densities paired by means of the density products $\{\scl{\cdot}{\cdot}_y\}_{y\in\Mc_N}$, give integrable scalar densities on $\Mc_N$. 

The regularity conditions, we are going to introduce here, will be analogous to those presented in Section \ref{reg-con}. However, some differences will be unavoidable. The reason is that each space $\Gamma^\oplus_y$ does not consist of all scalar products on $T_y\Mc_N$ of signature $(Np,Np')$, but contains only some special ones. Therefore, introducing and working with these new regularity conditions, we will restrict ourselves to some coordinate systems on $\Mc_N$, which, in a sense, are compatible with decompositions \eqref{T-dec} and, thereby, with the special form of the elements of $\Gamma^\oplus_y$.      

We showed in Appendix \ref{N_N} that for every $y\in\Mc_N$, there exist local charts $\{(U_K,\varphi_K)\}$ $_{K=1,\ldots,N}$ on $\Mc$ such that \emi the sets $\{U_K\}$ are pairwise disjoint and \emii there exists a distinguished diffeomorphism from an open neighborhood $Z_y$ of $y$ onto $U_1\times\ldots \times U_N$. The composition $\Phi$ of this diffeomorphism with the map $\varphi_1\times\ldots\times\varphi_N$ and the set $Z_y$ form a chart $(Z_y,\Phi)$ on $\Mc_N$ (see \eqref{coor-N}). All charts of this sort constitute a smooth atlas on the manifold denoted in the appendix by $\Ac$. Here we will extend this atlas by admitting all charts obtained by restricting domains of charts in $\Ac$. The extended atlas will be denoted by $\Ac'$.                

Let $\Sigma$ be the set of all permutation of the sequence $(1,\ldots,N)$. Every chart in $\Ac'$ defines a local coordinate system 
\begin{equation}
(x^{i_1}_1,\ldots, x^{i_N}_N)\equiv (x^\ab)
\label{xxx}
\end{equation}
on $\Mc_N$, which is compatible with the decomposition \eqref{T-dec} in the following sense: there exists $\sigma\in\Sigma$ such that the tangent vectors $(\partial_{x^i_I})$ (with fixed $I$)  form a basis of $T_{x_{\sigma(I)}}\Mc$ in \eqref{T-dec} (see Appendix \ref{app-decomp} for a justification of this claim). Then each element of $\Gamma^\oplus_y$, $y=\{x_K\}$,  reads    
\begin{equation}
\gamma^\oplus_y=(\gamma^\oplus_y)_{\ab\bb}\,dx^\ab\ot dx^\bb=\sum_{K=1}^N (\gamma_{x_{\sigma(K)}})_{i_Kj_K}\,dx^{i_K}_K\ot dx^{j_K}_K.
\label{g^o+-comp}
\end{equation}
Thus the map
\[
\Gamma^\oplus_y\ni \gamma^\oplus_y\mapsto \big((\gamma_{x_{\sigma(1)}})_{i_1j_1},\ldots,(\gamma_{x_{\sigma(N)}})_{i_Nj_N}\big)\in \Gamma^N_\R
\]
is of the sort of the map \eqref{o+-lin-coor}. 

The description \eqref{g^o+-comp} of $\gamma^\oplus_y$ in terms of the components $\big((\gamma_{x_\sigma(K)})_{i_Kj_K}\big)$ is more explicit than that in terms of $\big((\gamma^\oplus_y)_{\ab\bb}\big)$. However, the symbols $\big((\gamma_{x_\sigma(K)})_{i_Kj_K}\big)$ are fairly complex and thereby somewhat unreadable. Therefore we would like to use the components $\big((\gamma^\oplus_y)_{\ab\bb}\big)$ instead of them. To this end we will neglect each zero component $(\gamma^\oplus_y)_{\ab\bb}$ for which $\ab$ and $\bb$ refer to coordinates, respectively, $x^{i_I}_I$ and $x^{j_J}_J$ with $I\neq J$. This will allow us to treat the set $\big((\gamma^\oplus_y)_{\ab\bb}\big)$ (being in fact an element of $\R^{(N\dim\Mc)^2}$) as an element of $\Gamma^N_\R$ and identify the two sets of components under consideration:
\[
\big((\gamma^\oplus_y)_{\ab\bb}\big)\equiv \big((\gamma_{x_{\sigma(1)}})_{i_1j_1},\ldots,(\gamma_{x_{\sigma(N)}})_{i_Nj_N}\big).
\]

Let us now introduce a coordinate representation of a Hilbert half-density $\Psit$ on $\Mc_N$. To this end consider a chart $(Z,\Phi)\in\Ac'$ and the corresponding coordinate system \eqref{xxx}. Given $y\in Z$, $\Psit(y)$ is a half-density over $T_y\Mc_N$ valued in $H^\ot_y$, which means that the value of $\Psit(y)$ on the basis $(\partial_{x^{i_1}_1},\ldots,\partial_{x^{i_N}_N})\equiv (\partial_{x^{\ab}})$ of the tangent space, is an element of the Hilbert space:
\[
\Psit\big(y,(\partial_{x^\ab})\big)\in H^\ot_y.
\]       
Thus $\Psit\big(y,(\partial_{x^\ab})\big)$ is an equivalence class of a function 
\[
\Gamma^\oplus_y\ni \gamma^\oplus\mapsto \Psit\big(y,(\partial_{x^\ab}),\gamma^\oplus\big)\in\C. 
\]
Expressing $y$ by means of the coordinates and $\gamma^\oplus$ as in \eqref{g^o+-comp} we obtain a {\em coordinate representation} of $\Psit$ being the map
\begin{equation}
\Phi(Z)\times \Gamma^N_\R\ni (x^\ab,\gamma^\oplus_{\ab\bb})\mapsto \psi(x^\ab,\gamma^\oplus_{\ab\bb}):=\Psit\big(\Phi^{-1}(x^\ab),(\partial_{x^\ab}),\gamma^\oplus_{\ab\bb}\,dx^\ab\ot dx^\bb\big)\in\C.
\label{coor-repr-N_N}
\end{equation}

Now we can introduce the notion of {\em continuous Hilbert half-densities on $\Mc_N$ of compact and slowly changing $\Gamma^N_\R$-support} exactly as we did in Section \ref{reg-con} in the case of Hilbert half-densities on $\Mc$ with only three exceptions: 
\begin{enumerate}
\item the scalar product components $(\gamma^\oplus_{\ab\bb})$ in \eqref{coor-repr-N_N} do not describe arbitrary scalar products on $T_y\Mc_N$, $y=\Phi^{-1}(x^\ab)\in Z$, of signature $(Np,Np')$, but exclusively those in $\Gamma^\oplus_y$. Therefore  the components are restricted to be elements of $\Gamma^N_\R$. 
\item we do not allow ourselves to use arbitrary local coordinate systems on $\Mc_N$, but only those given by charts in $\Ac'$. 
\item Lemma \ref{lm-psi-zero-o+} should be used instead of Lemma \ref{lm-psi-zero}.
\end{enumerate}
This means in particular that \emi given an admissible coordinate system $(x^\ab)$, a {\em continuous} coordinate representation of every Hilbert half-density on $\Mc_N$ in the system is {\em unique} (provided it exists) and \emii  appropriate counterparts of Lemmas \ref{lm-cont-coor}, \ref{slow-coor} and \ref{PP-lin-comb} can be proven in the same way without any essential changes.

\subsubsection{Pairing of Hilbert half-densities into scalar densities}

Suppose that $\Psit$ and $\Psit'$ are Hilbert half-densities on  $\Mc_N$. Clearly, the map
\begin{equation}
\Mc_N\ni y\mapsto \scl{\Psit'}{\Psit}(y):=\scl{\Psit'(y)}{\Psit(y)}_y\in \tilde{\mathbf{C}}
\label{hden-pair-N}
\end{equation}
is a scalar density on $\Mc_N$ ($\tilde{\mathbf{C}}$ here is defined analogously to $\tilde{\mathbf{C}}$ in Section \ref{Hc-df}). 

As before, if we assume that both half-densities $\Psit$ and $\Psit'$ are continuous and that the $\Gamma^N_\R$-support of one of them is compact and slowly changing, then the scalar density \eqref{hden-pair-N} is continuous. This fact can be proven analogously\footnote{The regularity conditions imposed on elements of $\Hc^c_N$ are expressed in terms of coordinate systems defined by charts in $\Ac'$. Therefore the result of this analogous proof will be a conclusion that a coordinate representation of the density \eqref{hden-pair-N} in every coordinate system of this special sort, is continuous. But this is sufficient to claim that all coordinate representations of the density are continuous (see Equation \eqref{f'-f}) and thereby the density is continuous.} to the proof of Lemma \ref{PP-cont}. The only essential difference is a bit more complicated passage from the counterpart of Equation \eqref{int} to the counterpart of Equation \eqref{f-ppC}, where now Equation \eqref{int_g^o+} should be used.

\subsubsection{The Hilbert space $\Hc_N$}

Let $\Hc^c_{N}$ be the vector space of all continuous Hilbert half-densities on $\Mc_N$ of compact $\Mc_N$-support and of compact and slowly changing $\Gamma^N_\R$-support. For any two elements $\Psit$ and $\Psit'$ of $\Hc^c_{N}$, the scalar density $\scl{\Psit'}{\Psit}$ on $\Mc_N$ is continuous and of compact support and therefore the density can be naturally integrated over this paracompact manifold. The following map
\begin{equation}
\Hc^c_{N}\times \Hc^c_{N} \ni (\Psit',\Psit)\mapsto \scal{\Psit'}{\Psit}:=\int_{\Mc_N} \scl{\Psit'}{\Psit}\in\C,
\label{HNc-scal}
\end{equation}
where the integral at the r.h.s. is the integral of the scalar density $\scl{\Psit'}{\Psit}$, is an inner product on $\Hc^c_{N}$. 

By definition, the Hilbert space $\Hc_N$ is the completion of $\Hc^c_{N}$ in the norm induced by the inner product \eqref{HNc-scal}.

\subsubsection{Uniqueness of $\Hc_N$}

Let us recall that to construct the Hilbert spaces $\Hc_1$ and $\Hc_N$ we used the same diffeomorphism invariant field $x\mapsto d\mu_{x}$ of invariant measures. If for this purpose we used an other such field $x\mapsto d\check{\mu}_{x}$ instead, then we would obtain an other Hilbert space $\check{\Hc}_N$. It is not difficult to realize (see Equation \eqref{mu-cmu}) that there exists a positive number $c$ such that  
\begin{equation}
{\Hc_N}\ni\Psit\mapsto\frac{\Psit}{\sqrt{c^N}}\in\check{\Hc}_N 
\label{UN-c}
\end{equation}
is a unitary map.

Thus we conclude that the Hilbert space $\Hc_N$ is unique up to natural isomorphisms \eqref{UN-c}.

\subsection{Action of diffeomorphisms of $\Mc$ on the Hilbert space $\Hc_N$}

Let $\theta$ be a diffeomorphism of $\Mc$. It induces the following map:
\[
\Mc_N\ni \{x_1,\ldots,x_N\}\mapsto \Theta(\{x_1,\ldots,x_N\}):=\{\theta(x_1),\ldots,\theta(x_N)\}\in\Mc_N
\]   
being a diffeomorphism on $\Mc_N$---see Appendix \ref{diff-M-N}. Diffeomorphisms of this sort form a subgroup of the diffeomorphism group on $\Mc_N$. We will denote this subgroup by $\Diff_\Mc(\Mc_N)$. We showed in Appendix \ref{diff-M-N} that each $\Theta\in\Diff_\Mc(\Mc_N)$ preserves 
\begin{enumerate}
\item the atlas $\Ac$ on $\Mc_N$,  
\item the decompositions \eqref{T-dec},
\item the spaces $\{\Gamma^\oplus_y\}_{y\in\Mc_N}$. 
\end{enumerate}

Let $\Psit$ be a Hilbert half-density belonging to $\Hc^c_N$, and let $\Theta\in\Diff_\Mc(\Mc_N)$. Following the definition \eqref{chi*-P} of the pull-back of a Hilbert half-density on $\Mc$, we define the pull-back of $\Psit$        
\begin{equation}
(\Theta^*\Psit)(y,e_y):=(\Theta^{-1})^{t *\star}\,\Psit(\Theta(y),\Theta^te_y),
\label{pull-back-Th}
\end{equation}
where $e_y$ is a basis of $T_y\Mc_N$, and $\Theta^t$ denotes the tangent map given by the diffeomorphism $\Theta$.    

It can be shown that the pull-back $\Theta^*\Psit$ is again an element of $\Hc^c_N$ for every $\Theta\in\Diff_\Mc(\Mc_N)$. A proof of this fact is similar to the proof of Lemma \ref{pull-H^c}, one has only take into account that each element of $\Diff_\Mc(\Mc_N)$ preserves the atlas $\Ac'$ (since it preserves the atlas $\Ac$).     

Moreover, every diffeomorphism $\Theta\in\Diff_\Mc(\Mc_N)$ preserves the inner product \eqref{HNc-scal}. This can be proven analogously to Lemma \ref{in-prod-pres}. The only extra work, which has to be done, is 
\begin{enumerate}
\item to note that if $\Psit\in\Hc^c_N$, then for every $y\in\Mc_N$ and for every basis $e_y$ of $T_y\Mc_N$, the equivalence class $\Psit(y,e_y)\in H^\ot_y$ possesses a unique continuous compactly supported representative;    
\item to prove that the measure field on $\Mc_N$
\begin{equation}
y\to d\mu^\times_y
\label{dmuX-diff-inv}
\end{equation}
is invariant with respect to the action of all diffeomorphisms in $\Diff_\Mc(\Mc_N)$. 
\end{enumerate}

To this latter end recall that if $y=\{x_K\}\in\Mc_N$, then the measure $d\mu^\times_y$ is given by the push-forward of the measure $\bigtimes d\mu_{x_K}$ under the bijection $b^{-1}_y$ (see \eqref{1-to-1}). Let $y'=\{x'_K\}$ be a point of $\Mc_N$ such that $x_K=\theta(x'_K)$, which means that $y=\Theta(y')$. As shown in Appendix \ref{diff-M-N}, the bijections $b_y$ and $b_{y'}$ intertwine the pull-back $\Theta^{t *}:\Gamma^\oplus_y\to\Gamma^\oplus_{y'}$ and the pull-back $\bigtimes \theta^{t *}:\bigtimes \Gamma_{x_K}\to\bigtimes \Gamma_{x'_K}$ (see Equation \eqref{Xtheta} for definition of $\bigtimes \theta^{t *}$ and Equation \eqref{bb} for the relation between the bijections and the pull-backs). Consequently, 
\begin{multline*}
(\Theta^{t *})_\star d\mu^\times_y=(\Theta^{t *})_\star (b^{-1}_y)_{\star} \bigtimes d\mu_{x_K}=(b^{-1}_{y'})_\star \big(\bigtimes \theta^{t *}\big)_\star\bigtimes d\mu_{x_K}=\\=(b^{-1}_{y'})_\star\bigtimes\big((\theta^{t *})_\star\,d\mu_{x_K}\big)=(b^{-1}_{y'})_\star \bigtimes d\mu_{x'_K}=d\mu^\times_{y'}
\end{multline*}
---here in the third step we used Lemma \ref{lm-push-prod}, and in the forth step the diffeomorphism invariance of the measure field $x\mapsto d\mu_x$ on $\Mc$ (see Equation \eqref{l*xx}). We thus conclude that, indeed, the measure field \eqref{dmuX-diff-inv} is invariant with respect to the action of elements of $\Diff_\Mc(\Mc_N)$. 

As in the case of the Hilbert space $\Hc_1$, the pull-back \eqref{pull-back-Th} given by a diffeomorphism $\Theta\in \Diff_{\Mc}(\Mc_N)$ corresponding to $\theta\in\Diff(\Mc)$, can be unambiguously extended to a unitary operator on the Hilbert space $\Hc_N$. It is convenient to denote this operator by  $U_N(\theta^{-1})$---then 
\begin{equation}
\theta\mapsto U_N(\theta)
\label{diff-repr-HN}
\end{equation}
is a unitary representation of $\Diff(\Mc)$ on $\Hc_N$.

\subsection{Uniqueness of $\Hf$ }

We assumed that each space $\Hc_N$ for $N\geq 1$ is built using the same diffeomorphism invariant field \eqref{diff-inv} of invariant measures on $\Mc$. Consequently, the resulting Hilbert space $\Hf$ defined by the orthogonal sum \eqref{Hf-df} stems from this field. 

Suppose now, that a Hilbert space $\check{\Hf}$ is constructed in the same way from an other such a measure field---the other field is related to the former one by the formula \eqref{mu-cmu}. Then taking into account the distinguished unitary maps \eqref{U-c} and \eqref{UN-c} we see that there exists a distinguished unitary map (isomorphism) between $\Hf$ and $\check{\Hf}$. If $\Psit_N\in\Hc_N$, then this unitary map is given by the following formula:
\[
\Hf\ni (\Psit_N)\mapsto \Big(\frac{\Psit_N}{\sqrt{c^N}}\Big)\in\check{\Hf}.
\]            

We are then allowed to state that the Hilbert space $\Hf$ is unique up to natural (or distinguished) isomorphisms.

\subsection{Unitary representation of $\Diff(\Mc)$ on $\Hf$ }

The unitary representations \eqref{diff-repr-U1} and \eqref{diff-repr-HN} of $\Diff(\Mc)$ on, respectively, $\Hc_1$ and $\Hc_N$, $N\geq 2$, can be used to define the following unitary representation of the diffeomorphism group on the Hilbert space $\Hf$:
\[
\theta\mapsto \bigoplus_{N=1}^\infty U_N(\theta), 
\]      
where $\theta\in\Diff(\Mc)$.

\subsection{Hilbert spaces $\{\Hf\}$ built over $\Mc=\R$ \label{Hf-McR}}

In this section we will consider the Hilbert space $\Hf$ built in the case $\Mc=\R$ for signature either $(1,0)$ or $(0,1)$ and will show that this space is {\em separable} (regardless of signature). To this end we will show first that for every $N\geq 1$ the Hilbert space $\Hc_N$ constructed over $\Mc=\R$ is separable.

In Appendix \ref{exampl-N_N} we considered the manifold $\Mc_N$, $N\geq 2$, constituted of points of $\Mc=\R$ and constructed a bijection $\iota:\Mc_N\to\R^N_>$, where
\[
\R^N_>=\{\ (x_1,\ldots,x_N)\in\R^N \ | \ x_1>x_2>\ldots>x_{N-1}>x_N\ \}
\]
is an open subset of $\R^N$. We have further demonstrated that $\iota$ defines a global coordinate system on $\Mc_N$. Setting $\Mc\equiv \Mc_1$, $\R^1_>\equiv \R$ and $\iota\equiv\id$ on $\Mc_1\equiv\R^1_>$ will allow us to treat the case $N=1$ together with all the cases $N\geq 2$ in the considerations below.            
  
Let us then fix $N\geq 1$ and a Hilbert half-density $\Psit\in\Hc^c_N$. As shown in Appendix \ref{exampl-N_N}, for every $y\in \Mc_N$ there exists its open neighborhood $Z$ such that the chart $(Z,\iota|_Z)$ belongs to the atlas $\Ac$ on $\Mc_N$ and thereby to the atlas $\Ac'$. This together with continuity of $\Psit$ mean that if a coordinate system is given by $(Z,\iota|_Z)$, then there exists a continuous coordinate representation of $\Psit$ in this system. Moreover, this continuous representation is unique (see the last sentence of Section \ref{Hhd-N_N}). This uniqueness allows us to merge all such continuous coordinate representations of $\Psit$ into one {\em continuous} representation
\begin{equation}
\R^N_>\times \Gamma^N_\R\ni (x^\ab,\gamma^\oplus_{\ab\bb})\mapsto \psi(x^\ab,\gamma^\oplus_{\ab\bb})\in\C
\label{psi-iota}
\end{equation}
of $\Psit$ in the global coordinate system $(x^\ab)$ defined by the map $\iota$ (see \eqref{psi-cont} and \eqref{coor-repr-N_N})\footnote{If $N=1$, then the existence of the continuous representation \eqref{psi-iota} follows directly from the definition of continuous Hilbert half-densities on $\Mc$. Moreover, for $N=1$ the superscript $\oplus$ in $\gamma^\oplus$ in \eqref{psi-iota}, is superfluous.}. 

Furthermore, $\Psit$ under consideration is of compact and slowly changing $\Gamma^N_\R$-support. Using this fact we can conclude in an analogous way that for every $y\in\Mc_N$ there exists its open neighborhood $Z_y$ and a compact set $K_y\in \Gamma^N_\R$ such that for every $(x^\ab)\in\iota(Z_y)$ the support of $\psi_{(x^\ab)}$ is contained in $K_y$---$\psi_{(x^\ab)}$ is related to the map \eqref{psi-iota} by an obvious generalization of \eqref{psi_xi}.

It is easy to convince oneself that, the other way round, if a Hilbert half-density $\Psit$ of compact $\Mc_N$-support is \emi continuous in the coordinate system $(x^\ab)$ and \emii of compact and slowly changing $\Gamma^N_\R$-support in the same system, then $\Psit\in\Hc^c_N$.    

\begin{lm}
The map
\begin{equation}
\Hc^c_N\ni\Psit\mapsto \psi,
\label{Psit->psi}
\end{equation}
where $\psi$ is the function \eqref{psi-iota}, is a linear bijection from $\Hc^c_N$ onto the linear space $C^c(\R^N_>\times\Gamma^N_\R,\C)$ of all complex compactly supported continuous functions on $\R^N_>\times\Gamma^N_\R$.     
\end{lm}

\begin{proof}
By reasoning similar to that used in the proof of Lemma \ref{PP-lin-comb} one can show that the map \eqref{Psit->psi} is linear.

Let us fix $\Psit$ and $\psi$ related by the map \eqref{Psit->psi}. We know already that $\psi$ is continuous. Let us then show that $\psi$ is compactly supported.

If $\supp\Psit$ denotes the $\Mc_N$-support of $\Psit$, then obviously
\[
\supp\Psit\subset \bigcup_{y\in\supp \Psit} Z_y,
\]
where the sets $\{Z_y\}$ are defined in the paragraph just above the lemma. In other words, $\{Z_y\}_{y\in\supp\Psit}$ is an open cover of $\supp\Psit$. By definition of $\Hc^c_N$, the support of every its element is compact. Therefore the cover $\{Z_y\}_{y\in\supp\Psit}$ contains a finite open subcover $\{Z_{y_n}\}_{n=1,\ldots,m}$:
\begin{equation}
\supp\Psit\subset \bigcup_{n=1}^m Z_{y_n}.
\label{supp-Z}
\end{equation}

Suppose now that $\psi(x^\ab,\gamma^\oplus_{\ab\bb})\neq 0$ for some $(x^\ab,\gamma^\oplus_{\ab\bb})\in\R^N_>\times\Gamma^N_\R$. By continuity of $\psi$, $\Psit(y,\partial_{x^\ab})$ is a non-zero element of $H^\ot_{y}$, where $y=\iota^{-1}(x^\ab)$. Consequently, $(x^\ab)\in \iota(\supp\Psit)$.

On the other hand, if $\psi(x^\ab,\gamma^\oplus_{\ab\bb})\neq 0$, then $(\gamma^\oplus_{\ab\bb})\in \supp \psi_{(x^\ab)}$. But since $(x^\ab)\in\iota(\supp\Psit)$, by virtue of \eqref{supp-Z} there exists $n\in\{1,\ldots,m\}$ such that $(x^\ab)\in\iota(Z_{y_n})$. Then $\supp \psi_{(x^\ab)}\subset K_{y_n}$ (the sets $\{K_y\}$ are introduced just above the lemma). Consequently, $(\gamma^{\oplus}_{\ab\bb})\in K_{y_n}$.

We are then allowed to state that if $\psi(x^\ab,\gamma^\oplus_{\ab\bb})\neq 0$, then
\begin{align}
&(x^\ab)\in\iota(\supp\Psit), & & (\gamma^{\oplus}_{\ab\bb})\in \bigcup_{n=1}^m K_{y_n}.
\label{in-in}
\end{align}

Note now that $\iota(\supp\Psit)$ is compact, because it is the image of a compact set under a continuous map. The set $\bigcup_{n=1}^m K_{y_n}$ is compact being a union of a finite number of compact sets. The Cartesian product of $\iota(\supp\Psit)$ and $\bigcup_{n=1}^m K_{y_n}$ is then a {\em compact} subset of $\R^N_>\times \Gamma^N_\R$ and therefore it is closed. Because it is closed and $\psi(x^\ab,\gamma^\oplus_{\ab\bb})\neq 0$ implies \eqref{in-in}, then 
\begin{equation}
\supp \psi \subset \Big(\iota(\supp\Psit)\times \bigcup_{n=1}^m K_{y_n}\Big).
\label{supp-psi}
\end{equation}
We thus see that $\supp\psi$ is a closed subset of a compact set. Therefore $\supp\psi$ is compact. 

We just proved that the map \eqref{Psit->psi} is valued in $C^c(\R^N_>\times \Gamma^N_\R,\C)$. To show that the map is injective note that if $\psi$ is the value of the map at $\Psit$, then $\Psit$ can be unambiguously reconstructed from $\psi$ by means of the obvious generalization of the formula \eqref{P-from-p}.     

To finish the proof it remains to show that the map \eqref{Psit->psi} is surjective. To this end assume that $\psi\in C^c(\R^N_>\times \Gamma^N_\R,\C)$. If 
\begin{align*}
\pi_1:\R^N_>\times \Gamma^N_\R&\to\R^N_>, & \pi_2:\R^N_>\times \Gamma^N_\R&\to\Gamma^N_\R,\\
(x^\ab,\gamma^\oplus_{\ab\bb})&\mapsto(x^\ab), & (x^\ab,\gamma^\oplus_{\ab\bb})&\mapsto(\gamma^\oplus_{\ab\bb}),
\end{align*}
are canonical projections, then both sets $\pi_1(\supp\psi)$ and $\pi_2(\supp\psi)$ are compact being images of a compact set under continuous maps. It is not difficult to show that for every $(x^\ab)\in\R^N_>$, $\supp\psi_{(x^\ab)}$ is a subset of $\pi_2(\supp\psi)$---see the reasoning concerning the support of a function $h_x$ in Appendix \ref{app-proof}. This means that the Hilbert half-density $\Psit$ defined on $\Mc_N$ by $\psi$ with the help of the generalization of \eqref{P-from-p}, is of compact and slowly changing $\Gamma^N_\R$-support.            

On the other hand, if $(x^\ab)\not\in\pi_1(\supp\psi)$, then $\psi_{(x^\ab)}=0$. Therefore, if $\Psit$ is defined by $\psi$ as above, then $\Psit(y)\neq 0$ implies $y\in \iota^{-1}\big(\pi_1(\supp\psi)\big)$. $\supp\Psit$ is then contained in the compact (and closed) set $\iota^{-1}\big(\pi_1(\supp\psi)\big)$ and $\Psit$ is a half-density of compact $\Mc_N$-support.

Thus every $\psi\in C^c(\R^N_>\times \Gamma^N_\R,\C)$ defines by means of the generalization of \eqref{P-from-p} an element $\Psit\in\Hc^c_N$. The map \eqref{Psit->psi} is then surjective.
\end{proof}

By definition of $\Hc_N$, the space $\Hc^c_N$ is dense in $\Hc_N$. We will show now that $\Hc_N$ is isomorphic to a Hilbert space, which contains the space $C^c(\R^N_>\times \Gamma^N_\R,\C)$ as its dense subset.

Let us begin by expressing the inner product \eqref{HNc-scal} in terms of an iterated integral over $\R^N_>\times \Gamma^N_\R$. To this end consider an integrable scalar density $\Ft$ on $\Mc_N$. If $f$ is its coordinate representation in the global coordinate system $(x^\ab)$ defined by $\iota$ (see \eqref{F-cont}), then
\[
\int_{\Mc_N}\Ft=\int_{\R^N_>} f\,d\mu^N_L,
\]    
where $d\mu^N_L$ is the Lebesgue measure on $\R^N$. If $\Ft=\scl{\Psit'}{\Psit}$, where $\Psit',\Psit\in\Hc^c_N$, then combining Equations \eqref{PPpp}, \eqref{f-ppC} and \eqref{int_g^o+} we get 
\[
f(x^\ab)=\Ft(\iota^{-1}(x^\ab),\partial_{x^\ab})=\int_{\Gamma^N_\R} \overline{{\psi'}\!_{(x^\ab)}}\psi_{(x^\ab)} \,c^N\big(\bigtimes \Delta  d\mu_L\big).
\]  
Here $\psi'_{(x^\ab)}$ is the function on $\Gamma^N_\R$ related by the obvious generalization of \eqref{psi_xi} to the coordinate representation $\psi'$ of $\Psit'$ in the coordinates $(x^\ab)$, and $\bigtimes \Delta \,d\mu_L$ is the product of $N$ copies of $\Delta \,d\mu_L$. Taking into account the last two equations and the definition \eqref{HNc-scal}, we see that the inner product
\begin{equation}
\scal{\Psit'}{\Psit}=\int_{\R^N_>}\Big(\int_{\Gamma^N_\R} \overline{{\psi'}\!_{(x^\ab)}}\psi_{(x^\ab)} \,c^N\big(\bigtimes \Delta  d\mu_L\big)\Big)\,d\mu^N_L.
\label{scal-int}
\end{equation}

$\R^N_>\times \Gamma^N_\R$ is an open subset of $\R^{2N}$, and thereby a l.c.H. space. Therefore by virtue of the Riesz representation theorem, there exists a regular Borel measure $d\nu$ on $\R^N_>\times \Gamma^N_\R$ such that integrals defined by both $d\nu$ and $d\mu^N_L\times c^N(\bigtimes \Delta  d\mu_L)$ coincide on $C^c(\R^N_>\times \Gamma^N_\R)$. The measure $d\nu$ defines the Hilbert space $L^2(\R^N_>\times \Gamma^N_\R,d\nu)$, and $C^c(\R^N_>\times \Gamma^N_\R,\C)$ turns out to be a dense subset \cite{cohn}\footnote{Proposition 7.4.3 in \cite{cohn} concerns real functions. But by separating a complex function into its real and imaginary parts one can show that the proposition holds true in the complex case as well. \label{foot-R-C}} of this Hilbert space.

Let us denote by $\scal{\cdot}{\cdot}_{d\nu}$ the inner product  on $L^2(\R^N_>\times \Gamma^N_\R,d\nu)$. If $\psi',\psi$ are related to, respectively, $\Psit',\Psit\in\Hc^c_N$ by the map \eqref{Psit->psi}, then
\begin{multline*}
\scal{\psi'}{\psi}_{d\nu}=\int_{\R^N_>\times \Gamma^N_\R}\overline{\psi'}{\psi}\,d\nu=\int_{\R^N_>\times \Gamma^N_\R}\overline{\psi'}{\psi}\,\big(d\mu^N_L\times c^N(\bigtimes \Delta  d\mu_L)\big)=\\=\int_{\R^N_>}\Big(\int_{\Gamma^N_\R} \overline{{\psi'}\!_{(x^\ab)}}\psi_{(x^\ab)} \,c^N\big(\bigtimes \Delta  d\mu_L\big)\Big)\,d\mu^N_L,
\end{multline*}
where in the last step we used the Fubini-Tonelli theorem. Comparing the result above with \eqref{scal-int} we see that for every $\Psit',\Psit\in \Hc^c_N$ 
\[
\scal{\Psit'}{\Psit}=\scal{\psi'}{\psi}_{d\nu}.
\]      

We conclude that the map \eqref{Psit->psi} \emi is a linear bijection between linear dense subspaces of $\Hc_N$ and $L^2(\R^N_>\times \Gamma^N_\R,d\nu)$ and \emii preserves the inner products. Therefore the map can be unambiguously extended to a unitary map from $\Hc_N$ onto $L^2(\R^N_>\times \Gamma^N_\R,d\nu)$.

$\R^N_>\times \Gamma^N_\R$ is a second countable l.c.H. space being an open subset of $\R^{2N}$. As each regular measure on such a space is $\sigma$-finite \cite{cohn}, so is $d\nu$. On the other hand, each $\sigma$-finite Borel measure on a second-countable space defines a separable $L^2$ space \cite{spectr-qm}. This means that $L^2(\R^N_>\times \Gamma^N_\R,d\nu)$ is separable. $\Hc_N$ being isomorphic to the former Hilbert space, is separable as well. Consequently, $\Hf$ built over $\Mc=\R$ is separable, since it is defined as the countable orthogonal sum \eqref{Hf-df} of separable Hilbert spaces.  

There is also another important conclusion, which can be drawn from the results just obtained. Let us recall that on elements of $\Hc^c_N$ there are imposed seemingly strong conditions of compact $\Mc_N$-support and of compact and slowly changing $\Gamma^N_\R$-support. This fact may raise concerns about whether the resulting Hilbert space $\Hc_N$ is ``large enough'' from a physical point of view. But since $\Hc_N$ built over $\Mc=\R$ turned out to be isomorphic to $L^2(\R^N_>\times \Gamma^N_\R,d\nu)$, then at least in this case we can regard these concerns to be unfounded.

\section{Construction of the Hilbert space $\Kf$}


Let us fix a manifold $\Mc$, a metric signature $(p,p')$ such that $p+p'=\dim\Mc$ and a diffeomorphism invariant field \eqref{diff-inv} of invariant measures on $\Mc$. This field defines via \eqref{H^o+_y} the Hilbert space $H^\otimes_y$ for every $y\in\Mc_N$, where $N\geq 2$. Note however that the definition of $\Mc_N$ applied to the case $N=1$, gives the original manifold $\Mc$ (provided we identify $y=\{x\}$ with $x\in\Mc$). Then $H^\otimes_y$ coincides with $H_y$ defined by \eqref{Hx} (under the same identification). This observation allows us to simplify the presentation below by considering spaces $\{\Mc_N\}$ and corresponding Hilbert spaces for all $N\geq 1$.         

Let us then fix an integer $N\geq 1$ and consider a bundle-like set 
\begin{equation}
\mathbf{H}^\otimes:=\bigcup_{y\in\Mc_N}H^\otimes_y.
\label{bf-H^ot}
\end{equation}
Let $\Kc_N$ be a set, which consists of some special sections of $\mathbf{H}^\otimes$: a section $\Psi$ of $\mathbf{H}^\otimes$ belongs to $\Kc_N$ if 
\begin{enumerate}
\item the set
\begin{equation}
\{\ y\in\Mc_N \ | \ \Psi(y)\neq 0 \ \}
\label{Psi-neq-0}
\end{equation}
is countable;
\item the sum 
\begin{equation}
\sum_{y\in\Mc_N}||\Psi(y)||^2_y\equiv ||\Psi||^2,
\label{KN-norm}
\end{equation}
where $||\cdot||_y$ is the norm on $H^\otimes_y$, is finite (note that by virtue of the previous assumption, the uncountable sum above reduces to a sum of countable number of positive terms).
\end{enumerate}      

\begin{lm}
The map
\begin{equation}
\Kc_N\times \Kc_N\ni(\Psi',\Psi)\mapsto \scal{\Psi'}{\Psi}:=\sum_{y\in\Mc_N}\scal{\Psi'(y)}{\Psi(y)}_y \in \C
\label{KN-inner}
\end{equation}
is well-defined. $\Kc_N$ equipped with this map is a Hilbert space.  
\label{K1-Hilb}
\end{lm}
\noindent Let us recall that $\scal{\cdot}{\cdot}_y$ in \eqref{KN-inner} is the inner product on $H^\ot_y$.  Note also that the map \eqref{KN-inner} can be expressed alternatively as
\[
\scal{\Psi'}{\Psi}=\int_{\Mc_N}\scal{\Psi'(y)}{\Psi(y)}_y\,d\mu_0,
\]
where $d\mu_0$ is the counting measure on $\Mc_N$, being a diffeomorphism invariant measure on the manifold.  

A proof of Lemma \ref{K1-Hilb}, as following the well-known case of the Hilbert sequence space $l^2$, is relegated to Appendix \ref{K1-proof}.

For many practical purposes it would be convenient to have a dense linear subspace of $\Kc_N$, which would contain sufficiently regular elements of the Hilbert space. Denote by $\Kc^c_N$ a set, which consists of all elements $\{\Psi\}$ of $\Kc_N$ of the following property: for every $y\in\Mc_N$, the value $\Psi(y)\in H^\otimes_y$ is (an equivalence class of) an element of $C^c(\Gamma^\oplus_y,\C)$, i.e., a complex continuous function on $\Gamma^\oplus_y$ of compact support. {Let $\Kc^{cf}_N$ be the set consisting of all elements of $\Kc^c_N$, for which the set \eqref{Psi-neq-0} is finite.}   

{\begin{lm}
$\Kc^{cf}_N$ is a dense linear subspace of $\Kc_N$.   
\label{Kc^c_N}
\end{lm}}

\begin{proof}
{Since every set $C^c(\Gamma^\oplus_y,\C)$ is a linear space, then both $\Kc^c_N$ and $\Kc^{cf}_N$ are linear subspaces of $\Kc_N$. It remains then to prove that $\Kc^{cf}_N$ is dense in $\Kc_N$. We will first show that $\Kc^c_N$ is dense in $\Kc_N$}.      

Let us fix $\Psi\in\Kc_N$. Then the set \eqref{Psi-neq-0} is countable and all its elements can be ordered to form a sequence $(y_n)$.  

We know that every $\Gamma^\oplus_y\cong \Gamma^N_\R$ is l.c.H. space, and the measure $d\mu^\times_y$ is regular and Borel. These imply that the space $C^c(\Gamma^\oplus_y,\C)$ is a dense subset of $H^\otimes_y=L^2(\Gamma^\oplus_y,d\mu^\times_y)$ \cite{cohn}\footnote{See Footnote \ref{foot-R-C}.}.      

This means that if $\psi$ is a non-zero element of $H^\otimes_y$, then there exists a sequence $(\psi'_m)$ of non-zero elements of $C^c(\Gamma^\oplus_y,\C)$, which converges to $\psi$. Then the sequence $(||\psi'_m||_y)$, where $||\cdot||_y$ is the norm on $H^\otimes_y$, converges to $||\psi||_y$. Consequently, functions 
\[
\psi_m=\frac{||\psi||_y}{||\psi'_m||_y}\psi'_m\in C^c(\Gamma^\oplus_y,\C)
\]          
form a sequence, which converges to $\psi$ and for every $m$, $||\psi_m||_y=||\psi||_y$ (in other words, all elements of the sequence $(\psi_m)$ belong to the sphere of radius $||\psi||_y$ centered at zero of $H^\ot_y$).   

Let us fix a natural number $m>0$. The conclusions above allow us to choose for every $y_n$ a function $\psi_{nm}\in C^c(\Gamma^\oplus_{y_n},\C)$ such that   
\begin{align*}
||\psi_{nm}-\Psi(y_n)||_{y_n}&<\frac{1}{\sqrt{2^n}m}, & ||\psi_{nm}||_{y_n}&=||\Psi(y_n)||_{y_n}.
\end{align*}

Define now a section $\Psi_m$ of $\mathbf{H}^\otimes$:  
\[
\Psi_m(y)=
\begin{cases}
\psi_{nm} & \text{if $y=y_n$}\\
0 & \text{otherwise}
\end{cases}.
\]
Obviously, for every $\Psi_m$ the set \eqref{Psi-neq-0} is countable and
\[
||\Psi_m||^2=\sum_{y\in\Mc_N}||\Psi_m(y)||^2_y=\sum_{y\in\Mc_N}||\Psi(y)||^2_y=||\Psi||^2<\infty. 
\] 
Thus $\Psi_m$ is an element of $\Kc^c_N$.  

On the other hand, the norm of $\Psi_m-\Psi$ in the Hilbert space $\Kc_N$, can be bounded from above as follows:
\[
||\Psi_m-\Psi||^2=\sum_{y\in \Mc_N}||\Psi_m(y)-\Psi(y)||^2_y=\sum_{n=1}^\infty||\psi_{nm}-\Psi(y_n)||^2_{y_n}<\sum_{n=1}^\infty \frac{1}{2^nm^2}=\frac{1}{m^2}.
\]

{This result means that $\Kc^c_N$ is dense in $\Kc_N$. Note now that for every $\Psi\in \Kc^c_N $ and for every $\eps>0$, there exists $\Psi_f\in \Kc^{cf}_N$ such that $||\Psi-\Psi_f||<\eps$---if the set \eqref{Psi-neq-0} for $\Psi$ is infinite, then to obtain the desired $\Psi_f$ it is enough to zero out values of $\Psi$ at appropriately chosen points in $\Mc_N$. Consequently, $\Kc^{cf}_N$ is also dense in $\Kc_N$. }
\end{proof}

Taking into account experiences gained from the study of the spaces $\{\Hc_N\}$, it is easy to realize that
\begin{enumerate}
\item if the Hilbert spaces $\Kc_N$ and $\check{\Kc}_N$ are constructed as above, starting from two distinct diffeomorphism invariant measure fields, then there exists a number $c>0$ such that 
\begin{equation}
{\Kc_N}\ni\Psi\mapsto\frac{\Psi}{\sqrt{c^N}}\in\check{\Kc}_N 
\label{KN-c}
\end{equation}
is a unitary map.
\item if $\theta\in\Diff(\Mc)$ and $\Theta\in\Diff_\Mc(\Mc_N)$ are related diffeomorphisms, then the pull-back  
\begin{align*}
\Psi&\mapsto\Theta^*\Psi,\\
\big(\Theta^*\Psi\big)(y)&:=(\Theta^{-1})^{t *\star}\,\Psi\big(\Theta(y)\big)
\end{align*}
defined on $\Kc^{cf}_N$ is a linear bijection onto $\Kc^{cf}_N$  and preserves the inner product \eqref{KN-inner}. Consequently, with help of Lemma \ref{Kc^c_N} the pull-back can be uniquely extended to a unitary map $U_N(\theta^{-1})$ on $\Kc_N$. Moreover, $\theta\mapsto U_N(\theta)$ is a unitary representation of $\Diff(\Mc)$ on the Hilbert space. 
\end{enumerate} 

Each Hilbert space $H^\ot_y$ is separable, because every Hilbert space $H_x$ is separable \cite{qs-metr}. Let $\{\psi_{yn}\}_{n\in\N}$ be a basis of $H^\ot_{y}$ and let $\Psi_{yn}$ be an element of $\Kc_N$ such that 
\[
\Psi_{yn}(y')=
\begin{cases}
\psi_{yn} & \text{if $y'=y$},\\
0 & \text{otherwise}
\end{cases}.
\]  
It is not difficult to demonstrate that $\{\Psi_{yn}\}_{y\in\Mc_N,\,n\in\N}$ is an orthonormal basis\footnote{To prove that the linear span of $\{\Psi_{yn}\}_{y\in\Mc_N,\,n\in\N}$ is dense in $\Kc_N$, one can use a reasoning similar to that used {in the proof of Lemma \ref{Kc^c_N}}.} of $\Kc_N$, which thereby is a non-separable Hilbert space.      

Now we are able to merge the spaces $\{\Kc_N\}$ into the Hilbert space $\Kf$ in the same way, the spaces $\{\Hc_N\}$ were merged into $\Hf$, that is, by means of an orthogonal sum:
\[
\Kf:=\bigoplus_{N=1}^\infty \Kc_N.
\]  
Let us recall that all the Hilbert spaces $\{\Kc_N\}$ above stem from the same diffeomorphism invariant field \eqref{diff-inv} of invariant measures.

It is now a simple exercise 
\begin{enumerate}
\item to show that $\Kf$ is unique up to natural isomorphisms built of the unitary maps  \eqref{KN-c}; 
\item to construct a unitary representation of $\Diff(\Mc)$ on $\Kf$ from the representations $\{U_N\}$ just defined.   
\end{enumerate}
Note also that $\Kf$ is a non-separable Hilbert space being built from non-separable Hilbert spaces $\{\Kc_N\}$.  

\section{Summary and outlook}

In this paper we constructed two Hilbert spaces $\Hf$ and $\Kf$ over the set $\Qc(\Mc)$ of all metrics of arbitrary signature $(p,p')$, defined on a (smooth connected paracompact) manifold $\Mc$ \footnote{Note, however, that both spaces $\Hf$ and $\Kf$ do exist, even if the set $\Qc(\Mc)$ is empty.}. Each space was obtained by merging the tensor products $\{H_{x_1}\ot\ldots\ot H_{x_N}\}_{N=1,2,\ldots}$---every state in $\Hf$ was built of an uncountable number of elements of these products, while every state in $\Kf$ from a countable number of them. 

The Hilbert spaces $\{H_x\}$ were defined by means of a diffeomorphism invariant field $x\mapsto d\mu_x$ of invariant measures. The diffeomorphism invariance of this measure field resulted in existence of a unitary representation of the diffeomorphism group $\Diff(\Mc)$ on each Hilbert space $\Hf$ and $\Kf$. On the other hand, the measure field is unique up to a multiplicative constant, which resulted in uniqueness of each of the Hilbert spaces $\Hf$ and $\Kf$ up to distinguished isomorphisms. The two Hilbert spaces $\{\Hf\}$ built over $\Mc=\R$ turned out to be separable, while all the Hilbert spaces $\{\Kf\}$ to be non-separable.       

{The Hilbert spaces $\Hf$ and $\Kf$ are constructed on the basis of the {\em pointwise} d.o.f. \eqref{dof}. Since we are going to use the spaces obtained for signature $(3,0)$ to canonical quantization of the ADM formalism, there arises a question whether such pointwise d.o.f. are suitable for this purpose. At the present moment we are unable to justify this choice on physical grounds. On the other hand, from a mathematical point of view, the map \eqref{dof} seems to be very simple and natural. Let us also note that the application of the d.o.f. \eqref{dof} allowed to construct the two Hilbert spaces, both equipped with diffeomorphism invariant inner products, while (to the best of our knowledge) similar constructions based on other d.o.f. defined on $\Qc(\Mc)$, have not been known so far.}

Let us now present an outlook to future research.

The most important question is whether either $\Hf$ or $\Kf$ {constructed} in the case of signature $(3,0)$, can be used for {the} quantization of the ADM formalism. As emphasized in the introduction to this paper, there is no guarantee that the answer to this question is in affirmative. The first step to be done to clarify this issue, is an attempt to define on the Hilbert spaces operators \cite{prep} related to the ADM canonical variables. As it seems, the fact that each measure $d\mu_x$ is an invariant measure on the homogeneous space $\Gamma_x$, should result in self-adjointness of operators related to the momentum variable.  

Other issues we left open here are: \emi the relation between each  Hilbert space $\Hc_N$ (being a building block of $\Hf$) and the set  of all ``square integrable'' Hilbert half-densities on $\Mc_N$, \emii the question whether each space $\Hc_N$ generated by the space $\Hc^c_N$ of special Hilbert half-densities, is ``large enough'' from a physical point of view\footnote{Recall that in the case $\Mc=\R$ the answer to this question is in affirmative.} and \emiii the question whether all the Hilbert spaces $\{\Hf\}$ are separable.

In this paper we considered the bundle-like sets $\tilde{\mathbf{H}}$, $\tilde{\mathbf{H}}^\ot$ and $\mathbf{H}^\ot$ defined, respectively, by the formulas \eqref{t-bfH}, \eqref{t-bfH^ot} and \eqref{bf-H^ot}. It is interesting, at least from a mathematical point of view, whether these spaces can be endowed with local trivializations, which would make them genuine bundles. In particular, it is interesting, whether the set $\mathbf{H}^\ot$ is a Hilbert bundle (see e.g. \cite{hilb_bund}) over $\Mc_N$.

Let us emphasize that the Hilbert spaces $\Hf$ and $\Kf$ are distinctly different from the space $\Sc$ of quantum states built in \cite{qs-metr} by means of the Kijowski's projective method over the same set $\Qc(\Mc)$ of metrics. To construct $\Sc$, we extended each Hilbert space in $\{H_{x_1}\ot\ldots\ot H_{x_N}\}_{N=1,2,\ldots}$ to a larger space $\Sc_\la$, where $\la\equiv\{x_1,\ldots,x_N\}$. Namely, the space $\Sc_\la$ was defined as the set of all algebraic states on the $C^*$-algebra $\Bc_\la$ of all bounded operators on $H_\la\equiv H_{x_1}\ot\ldots\ot H_{x_N}$. Since all the sets $\{\Sc_\la\}$ form naturally a projective family, the space $\Sc$ were obtained as the projective limit of the family. As a result, the space $\Sc$ is not a Hilbert space, but it is rather a convex set of all algebraic states on a ``large'' $C^*$-algebra, obtained by merging all the algebras $\{\Bc_\la\}$ \cite{mod-proj}. 

Despite these differences, the spaces of quantum states: $\Hf$, $\Kf$ and $\Sc$, are constructed of the same building blocks (being the Hilbert spaces $\{H_{x_1}\ot\ldots\ot H_{x_N}\}_{N=1,2,\ldots}$) and in our opinion it is worthwhile to explore more closely the relations between these three spaces.

\paragraph{Acknowledgments} I am very grateful to Jerzy Kijowski, {Jerzy Lewandowski} and Piotr So{\l}tan for valuable discussions and help. This work was partially supported by the Polish National Science Centre grant No. 2018/30/Q/ST2/00811.

\appendix

\section{The set $\Mc_N$ as a manifold }

Let us fix an integer $N\geq 2$ and a smooth connected paracompact manifold $\Mc$ and define
\[
\Mc_N:=\big\{ \ \{x_1,\ldots,x_N\}\equiv\{x_K\}\subset \Mc \ | \ \text{$x_I\neq x_J$ for $I\neq J$} \ \big\}.
\]

\subsection{Smooth atlas on $\Mc_N$ \label{N_N}}

Our goal in this section is to define a smooth atlas on $\Mc_N$, which will allow us to treat this set a smooth manifold.   

\paragraph{A submanifold of $\Mc^N$} To this end let us consider the following set:
\[
\Mc^N_0:=\{\ (x_1,\ldots,x_N)\equiv(x_K)\in\Mc^N \ | \ \text{$x_I\neq x_J$ for $I\neq J$} \ \}.
\] 
$\Mc^N_0$ is an open subset of $\Mc^N$. Indeed, if $(x_K)$ is an arbitrary point in $\Mc^N_0$, then for each $x_I\in(x_K)$ there exists an open neighborhood $U_I\subset\Mc$ of $x_I$ such that 
\begin{equation}
\forall_{I\neq J} \quad {U}_I\cap {U}_J=\varnothing. 
\label{clo-dis}
\end{equation}
Then 
\[
U_1\times\ldots \times U_N\equiv \bigtimes U_K
\]
is an open subset of $\Mc^N$, contains the point $(x_K)$ and is contained in $\Mc^N_0$. 

Consequently, $\Mc^N_0$ is a smooth manifold being an open subset of $\Mc^N$.

\paragraph{Action of permutations on $\Mc^N_0$} Let $\Sigma$ be the group of all permutations of the finite sequence $(1,2,3,\ldots,N)$. Given $\sigma\in\Sigma$, the following map                
\begin{equation}
\Mc^N_0\ni (x_K)\mapsto \bar{\sigma}(x_K):=(x_{\sigma(1)},\ldots,x_{\sigma(N)})\equiv (x_{\sigma(K)})\in\Mc^N_0
\label{bar-si}
\end{equation}
is a diffeomorphism on $\Mc^N_0$. 

To see this let us fix a point $(x^0_K)\in\Mc^N_0$ and for every $x^0_I\in (x^0_K)$ choose its open neighborhood $U_I$ is such a way that \emi the neighborhoods $\{U_I\}$ satisfy \eqref{clo-dis} and \emii each $U_I$ is a domain of an $\R^{\dim\Mc}$-valued map $\varphi_I$, which defines a coordinate system $(x^i_I)_{i=1,\ldots,\dim\Mc}$ on $U_I$. Then the map
\begin{equation}
\bigtimes\varphi_K\equiv\varphi_1\times\ldots\times\varphi_N:U_1\times\ldots \times U_N\to \R^{N\dim\Mc}
\label{map-Mind}
\end{equation}
defines a local coordinate system $(x^{i_1}_1,\ldots,x^{i_N}_N)\equiv (x^{i_K}_K)$ on $\bigtimes U_K$.

Clearly, the following map        
\begin{equation}
(\varphi_{\sigma(1)}\times\ldots\times\varphi_{\sigma(N)})\circ\bar{\sigma}\circ(\varphi^{-1}_1\times\ldots\times\varphi^{-1}_N):(x^{i_1}_1,\ldots,x^{i_N}_N)\mapsto (x^{i_1}_{\sigma(1)},\ldots,x^{i_N}_{\sigma(N)}) 
\label{sm-map}
\end{equation}
between appropriate open subsets of $\R^{N\dim\Mc}$ is smooth. Since $\varphi_{\sigma(1)}\times\ldots\times\varphi_{\sigma(N)}$ is a map on $U_{\sigma(1)}\times\ldots\times U_{\sigma(N)}$ of the same sort as \eqref{map-Mind}, smoothness of \eqref{sm-map} means that the map \eqref{bar-si} is smooth as well. Consequently, the inverse map $\bar{\sigma}^{-1}$ is also smooth, since it is given by the inverse permutation $\sigma^{-1}$. Thus we see that, indeed, \eqref{bar-si} is a diffeomorphism.

\paragraph{Natural projection from $\Mc^N_0$ onto $\Mc_N$}

The map 
\begin{equation}
\Mc^N_0\ni (x_K)\mapsto \pi(x_K):=\{x_K\}\in \Mc_N,
\label{pi}
\end{equation}
is a natural surjection (projection) from $\Mc^N_0$ onto $\Mc_N$, which ``forgets'' about the ordering of points in $(x_K)$. This map  will be used to define the smooth atlas on $\Mc_N$. 

It follows immediately\footnote{If $(x_K)\in\Mc^N_0$, then the set
\[
\{\ \bar{\sigma}(x_K) \ | \ \sigma\in\Sigma\ \}\subset \Mc^N_0
\]  
represents the {\em unordered} set $\{x_K\}$ of pairwise distinct points of $\Mc$. Under this identification, $\Mc_N$ is the set of all orbits of the action \eqref{bar-si} of the group $\Sigma$ on $\Mc^N_0$.} from \eqref{pi} that for every $\{x_K\}\in \Mc_N$ and for every subset $U$ of $\Mc^N_0$,
\begin{align}
\pi^{-1}(\{x_K\})&=\{ \ \bar{\sigma}(x_K) \ | \ \sigma\in\Sigma \ \}, & \pi^{-1}\big(\pi(U)\big)&=\bigcup_{\sigma\in\Sigma}\bar{\sigma}(U).
\label{pi--1}
\end{align}

\paragraph{Topology on $\Mc_N$} In order to define the smooth atlas on $\Mc_N$, let us first equip the set with a suitable topology: we will say that a set $Z\subset \Mc_N$ is open, if its preimage under $\pi$, $\pi^{-1}(Z)$, is an open subset of $\Mc^N_0$. This immediately means that $\pi$ becomes a continuous map. Moreover, the image $\pi(U)$ of every open set $U\subset \Mc^N_0$ is open---this is because by virtue of the second of Equations \eqref{pi--1}, the preimage of $\pi(U)$ is a union of open sets in $\Mc^N_0$. 

To show that the topology just introduced is Hausdorff, consider two distinct elements $\{x_K\}$ and $\{x'_K\}$ of $\Mc_N$. Since $\Mc$ is Hausdorff, for every $x\in\{x_K\}\cup\{x'_K\}$ there exists its open neighborhood $U_x\subset\Mc$ such that $U_x\cap U_{\check{x}}=\varnothing$ if $x\neq\check{x}$. Because $\{x_K\}\neq \{x'_K\}$, there exists $x'\in\{x'_K\}$, which does not belong to $\{x_K\}$. Consequently, $U_{x'}$ is disjoint with every $U_{x_I}$, $x_I\in\{x_K\}$,  and therefore for each two permutation $\sigma,\sigma'\in\Sigma$, 
\begin{equation}
\bar{\sigma}\big(\bigtimes U_{x_K}\big)\cap\bar{\sigma}'\big(\bigtimes U_{x'_K}\big)=\varnothing.
\label{Ux-Ux'}
\end{equation}
 
It is clear that $\pi\big(\bigtimes U_{x_K}\big)$ and $\pi\big(\bigtimes U_{x'_K}\big)$ are open neighborhoods of, respectively, $\{x_K\}$ and $\{x'_K\}$. Suppose that the neighborhoods are not disjoint. Then their preimages under the surjection $\pi$ are also not disjoint:
\begin{multline*}
\varnothing\neq \pi^{-1}\big(\pi\big(\bigtimes U_{x_K}\big)\big)\cap \pi^{-1}\big(\pi\big(\bigtimes U_{x'_K}\big)\big)=\\=\Big(\bigcup_{\sigma\in\Sigma}\bar{\sigma}\big(\bigtimes U_{x_K}\big)\Big)\cap\Big(\bigcup_{\sigma'\in\Sigma}\bar{\sigma}'\big(\bigtimes U_{x'_K}\big)\Big)
\end{multline*}
(here we used the second of Equations  \eqref{pi--1}). But this contradicts Equation \eqref{Ux-Ux'}, which means that the neighborhoods under consideration are disjoint. The topology on $\Mc_N$ is thus Hausdorff.

\paragraph{The atlas on $\Mc_N$} Consider again the domain of the map \eqref{map-Mind}, keeping in mind that for $I\neq J$ the sets ${U}_I$ and ${U}_J$ are disjoint. It turns out that the map $\pi|_{\bigtimes U_K}$, that is, the map $\pi$ restricted to $\bigtimes U_K$, is a bijection onto its image. Indeed, it follows from the first of Equations \eqref{pi--1} that 
\begin{equation}
\pi(x_K)=\pi(x'_K)
\label{pipi}
\end{equation}
if and only if 
\begin{equation}
(x'_K)=\bar{\sigma}(x_K)
\label{x'-pix}
\end{equation}
for some $\sigma\in\Sigma$. But since the sets $\{U_I\}$ are pairwise disjoint, the intersection 
\begin{equation}
\big(\bigtimes U_K\big)\cap\bar{\sigma}\big(\bigtimes U_K\big)=\varnothing
\label{U-sigU}
\end{equation}
for every $\sigma\in\Sigma$ except the identity permutation. This means that if \eqref{pipi} holds for two elements of $\bigtimes U_K$, then the elements coincide.             

Denote by $\pi^{-1}_{\!\bigtimes\!U_K}$ the map from $\pi(\bigtimes U_K)$ onto $\bigtimes U_K$ inverse to the map $\pi|_{\bigtimes U_K}$. Let $U$ be an open subset of the (open) set $\bigtimes U_K$. Then the preimage of $U$ under $\pi^{-1}_{\!\bigtimes\!U_K}$
\[
(\pi^{-1}_{\!\bigtimes\!U_K})^{-1}(U)=\pi(U)
\]   
and therefore is an open subset of $\Mc_N$ (see the second  of Equations \eqref{pi--1}). This means that $\pi^{-1}_{\!\bigtimes\!U_K}$ is a {continuous} map. Since the map $\pi$ is also continuous, $\pi^{-1}_{\!\bigtimes\!U_K}$ is a homeomorphism.
  
This property of $\pi^{-1}_{\!\bigtimes\!U_K}$ allows us to define a local coordinate system on $\Mc_N$: given a map \eqref{map-Mind}, the composition
\begin{equation}
\Phi\equiv\big(\bigtimes\varphi_K\big)\circ \pi^{-1}_{\!\bigtimes\!U_K}: \pi\big(\bigtimes U_K\big)\to \R^{N\dim\Mc}
\label{coor-N}   
\end{equation}
is a homeomorphism onto its image and defines thereby a local coordinate system on $\Mc_N$. Domains of all maps of the form \eqref{coor-N} constitute an open cover of $\Mc_N$. Therefore charts given by all maps \eqref{coor-N} and their domains, form a continuous atlas $\Ac$ on $\Mc_N$. 

\paragraph{$\Ac$ is smooth} Let us show now that $\Ac$ is also smooth. To this end consider the map \eqref{coor-N} and an other one of this sort:    
\begin{equation}
\Phi'\equiv\big(\bigtimes\varphi'_K\big)\circ \pi^{-1}_{\bigtimes U'_K}: \pi\big(\bigtimes U'_K\big)\to \R^{N\dim\Mc},
\label{coor'-N}
\end{equation}
and suppose that the domains of $\Phi$ and $\Phi'$ are not disjoint:
\begin{equation}
Z\equiv\pi\big(\bigtimes U_K\big)\cap \pi\big(\bigtimes U'_K\big)\neq \varnothing.
\label{piU-piU'}
\end{equation}
Our goal now is to show that the transition function $\Phi'\circ \Phi^{-1}$ related to $Z$ is smooth. 

Fix $(x'_K)\in\bigtimes U'_K$. Then $\pi(x'_K)\in Z$ if and only if there exists $(x_K)\in\bigtimes U_K$ satisfying Equation \eqref{pipi}. We know already that \eqref{pipi} is equivalent to Equation \eqref{x'-pix} holding for some $\sigma\in\Sigma$. We are then allowed to state that $\pi(x'_K)\in Z$ if and only if there exists $\sigma\in\Sigma$ such that  
\[
(x'_K)\in \big(\bigtimes U'_K\big)\cap \bar{\sigma}\big(\bigtimes U_K\big)\equiv U_\sigma,
\]      
or, equivalently, if and only if
\[
(x'_K)\in \bigcup_{\sigma\in\Sigma} U_\sigma\subset \bigtimes U'_K.
\]
This fact together with the inclusion $Z\subset \pi(\Mc^N_0)$ (see \eqref{piU-piU'}) mean that
\begin{equation}
Z=\pi\Big(\bigcup_{\sigma\in\Sigma} U_\sigma\Big)=\bigcup_{\sigma\in\Sigma} \pi(U_\sigma).
\label{Z-piU}
\end{equation}
Applying \eqref{U-sigU} we see that if $\sigma\neq\sigma'$, then
\begin{equation}
U_\sigma\cap U_{\sigma'}=\varnothing.
\label{Us-U's}
\end{equation}
But $\pi$ restricted to $\bigtimes U'_K$ is injective and therefore the sets $\{\pi(U_\sigma)\}$ appearing at the r.h.s. of \eqref{Z-piU} are pairwise disjoint.   

Now it is enough to find the transition function $\Phi'\circ\Phi^{-1}$ on each non-empty set $\pi(U_\sigma)$. Consider then 
\begin{equation}
(x'_K)=\bar{\sigma}(x_K)\in U_\sigma,
\label{xx'}
\end{equation}
where $(x'_K)\in \bigtimes U'_K$ and $(x_K)\in \bigtimes U_K$. Then $\{x'_K\}=\pi(x'_K)=\pi(x_K)=\{x_K\}$ and  
\begin{align*}
\Phi'(\{x'_K\})&=\big(\bigtimes \varphi'_K\big)(x'_K)=(x^{\prime i_K}_K)\in\R^{N\dim\Mc},\\
\Phi(\{x_K\})&=(\bigtimes \varphi_K\big)(x_K)=(x^{i_K}_K)\in\R^{N\dim\Mc}.
\end{align*}
Thus $(x^{\prime i_K}_K)$ and $(x^{i_K}_K)$ are values of coordinates defined by, respectively, $\Phi'$ and $\Phi$, of the same point $\{x'_K\}=\{x_K\}$. Using \eqref{xx'} we obtain the following relation between the values:
\begin{equation}
(x^{\prime i_K}_K)=\big(\bigtimes \varphi'_K\big)(x'_K)=\Big(\big(\bigtimes \varphi'_K\big)\circ \bar{\sigma}\Big)(x_K)=\Big(\big(\bigtimes \varphi'_K\big)\circ \bar{\sigma}\circ\big(\bigtimes \varphi_K\big)^{-1}\Big)(x^{i_K}_K)
\label{sig-expr}
\end{equation}
Obviously, this relation is nothing else but the value of the transition function $\Phi'\circ\Phi^{-1}$ at $(x^{i_K}_K)$. We thus conclude that $\Phi'\circ\Phi^{-1}$ on the set $\Phi\big(\pi(U_\sigma)\big)$ is smooth, since it coincides with the coordinate expression \eqref{sig-expr} of the diffeomorphism $\bar{\sigma}$---given $I$,
\begin{equation}
(x^{\prime i_I}_I)=\varphi^{\prime }_I\circ \varphi^{-1}_{\sigma(I)}(x^{j_I}_{\sigma(I)}).
\label{x'i-xi}
\end{equation}

Thus the transition map is smooth on its whole domain $\Phi(Z)=\bigcup_{\sigma\in\Sigma}\Phi\big(\pi(U_\sigma)\big)$.

\paragraph{Conclusions} Consequently, the atlas $\Ac$ is smooth and the set $\Mc_N$ is a smooth manifold. But since we are going to integrate densities over this manifold we extend $\Ac$ to the maximal smooth atlas on $\Mc_N$, since restricting ourselves to the atlas $\Ac$ would be inconvenient.

For every collection $\{U'_K\}_{K=1,\ldots,N}$ of pairwise disjoint open subsets of $\Mc$, the map
\[
\pi^{-1}_{\!\bigtimes\!U'_K}:\pi\big(\bigtimes U'_K\big)\to\bigtimes U'_K
\]    
is a diffeomorphism. Indeed, let us choose a collection of charts $\{(U_K,\varphi_K)\}_{K=1,\ldots,N}$ on $\Mc$ such that $U_K\subset U'_K$. Using the maps \eqref{map-Mind} and \eqref{coor-N} we see that the following coordinate expression for $\pi^{-1}_{\!\bigtimes\!U'_K}$:
\[
\big(\bigtimes\varphi_K\big)\circ \pi^{-1}_{\!\bigtimes\!U'_K} \circ \Phi^{-1}
\]
is the identity on $(\bigtimes\varphi_K(U_K))\subset\R^{N\dim\Mc}$. 

This fact allows us to state that \emi $\Mc_N$ is locally diffeomorphic to $\Mc^N$ and \emii the projection $\pi$ is smooth.

\paragraph{$\Mc_N$ is paracompact} A manifold is paracompact if and only if each connected component of the manifold is second countable \cite{ma-diff-geo}. We assumed that $\Mc$ is paracompact and connected, which means that $\Mc$ is second countable. Thus $\Mc^N$ and $\Mc^N_0$ are second countable as well. 

Let $\mathbb{B}$ be a countable base for the topology of $\Mc^N_0$. Suppose that $Z$ is an open subset of $\Mc_N$. Then, by definition of the topology on $\Mc_N$, the preimage $\pi^{-1}(Z)$ is open subset of $\Mc^N_0$. Consequently, the preimage is a union of open sets $\{U_\alpha\}\subset\mathbb{B}$. We know already that $\pi$ maps open subsets of $\Mc^N_0$ onto open ones in $\Mc_N$. Thus  
\[
Z=\pi\big(\pi^{-1}(Z)\big)=\pi\Big(\bigcup_\alpha U_\alpha\Big)= \bigcup_\alpha \pi(U_\alpha).
\]
The conclusion is that every open subset of $\Mc_N$ is a union of open sets being images of elements of $\mathbb{B}$ under $\pi$---these images form a countable base for the topology of $\Mc_N$.

$\Mc_N$ is thus second countable and thereby paracompact.    
\subsection{An example of $\Mc_N$ \label{exampl-N_N}}

Here we will find an explicit description of $\Mc_N$ for $\Mc=\R$ by means of a global coordinate system on $\Mc_N$.  

Let us fix an integer $N\geq 2$. If $\Mc=\R$, then for every $y\equiv\{x_K\}\in \Mc_N$ it is possible to form a {\em decreasing} sequence from all elements of $y$ i.e. there exists a permutation $\sigma\in\Sigma$ such that
\[
x_{\sigma(1)}>x_{\sigma(2)}>\ldots>x_{\sigma(N-1)}>x_{\sigma(N)}.
\]
This observation allows us to define the following map
\[
\Mc_N\ni \{x_K\}\mapsto \iota(\{x_K\}):=\text{the decreasing sequence of elements of $\{x_K\}$ }\in\Mc^N_0\subset \R^N.
\]
It is obvious that the map is a bijection onto its image such that
\begin{align}
\pi\circ\iota&=\id,& \iota\circ\pi&=\id,
\label{pi-i}
\end{align}
where the last equation holds on $\iota(\Mc_N)\subset \Mc^N_0$. Let
\[
\R^N_>:=\{\ (x_1,\ldots,x_N)\in\R^N \ | \ x_1>x_2>\ldots>x_{N-1}>x_N \ \}\subset \Mc^N_0.
\]
It is clear that $\iota(\Mc_N)=\R^N_>$. 

Fix an arbitrary $z\equiv (x_1,\ldots,x_N)\in\R^N_>$ and define 
\[
2\eps:={\rm min}\{\ x_{K+1}-x_K \ | \ x_{K+1},x_K\in z\ \}.
\]
If 
\[
U_K:=\,]x_K-\eps,x_K+\eps[,
\]
then $\bigtimes U_K\ni z$ is an open subset of $\R^N$. Since $\bigtimes U_K\subset \R^N_>$, the latter set is an open subset of $\R^N$.    

Moreover, since the sets $\{U_K\}$ just defined satisfy \eqref{clo-dis} it is not difficult to realize with help of Equations \eqref{pi-i} that $\iota$ restricted to $\pi\big(\bigtimes U_K\big)$ is a map of the sort of the map \eqref{coor-N} (with $\varphi_K$ being the identity on $U_K$). This means that every $y\in\Mc_N$ possesses an open neighborhood $Z$ such that $(Z,\iota|_Z)$ is a chart belonging to the atlas $\Ac$ on the manifold. This is sufficient to conclude that the map $\iota$ defines a global coordinate system on the manifold $\Mc_N$.     

Thus if $\Mc=\R$, then $\Mc_N$ can be identified with $\R^N_>$ being an open subset of $\R^N$.    

\subsection{Natural decomposition of tangent spaces to $\Mc_N$ \label{app-decomp}} 

Let us fix a point $y\equiv\{x_K\}\in\Mc_N$, the numbering of elements of $y$ and open subsets $\{U_K\}$ of $\Mc$ satisfying \eqref{clo-dis} such that $x_K\in U_K$. Let $y_I$ be a subset of $y$ obtained by removing from $y$ the point $x_I$: $y_I=\{x_1,\ldots,x_{I-1},x_{I+1},\ldots,x_N\}$. Define a map
\begin{equation}
U_I\ni x \mapsto \xi_{y_I}(x):=\{x\}\cup y_I\in\Mc_N.
\label{xi}
\end{equation}
To demonstrate that this map is smooth consider the following smooth map
\[
U_I\ni x \mapsto \xi^0_{y_I}(x):=(x_1,\ldots,x_{I-1},x,x_{I+1},\ldots,x_N)\in\Mc^N_0.
\]
Clearly, 
\begin{equation}
\xi_{y_I}=\pi\circ\xi^0_{y_I},
\label{xi-pixi}
\end{equation}
which means that $\xi_{y_I}$ is a composition of two smooth maps.  
  
If $\xi^t_{y_I}$ is the tangent map defined by $\xi_{y_I}$, then $\xi^t_{y_I}(T_{x_I}\Mc)$ is a linear subspace of $T_y\Mc_N$. This subspace is generated by all curves in $\Mc_N$ of the following form:      
\begin{equation}
t\mapsto \xi_{y_I}\big(x(t)\big)=\{x(t)\}\cup y_I\in\Mc_N,
\label{curves}
\end{equation}
where $t\mapsto x(t)$, $x(0)=x_I$, is a differentiable curve in $U_I\subset \Mc$.
   
It is evident that
\[
T_{(x_K)}\Mc^N_0=\bigoplus_{K=1}^N \xi^{0t}_{y_K}(T_{x_K}\Mc),
\]
where $\xi^{0t}_{y_K}$ denotes the tangent map given by $\xi^{0}_{y_K}$. Let us now act on both sides of this equation by the tangent map $\pi^t$ defined by $\pi$. Since $\pi$ restricted to $\bigtimes U_K$ is a (local) diffeomorphism\footnote{This is because $\pi^{-1}_{\!\bigtimes\! U_K}$ is a local diffeomorphism as proven at the end of Appendix \ref{N_N}.} to $\Mc_N$, we thus obtain 
\begin{equation}
T_y\Mc_N=\pi^t\big(T_{(x_K)}\Mc^N_0\big)=\bigoplus_{K=1}^N \pi^t\xi^{0t}_{y_K}(T_{x_K}\Mc)=\bigoplus_{K=1}^N \xi^t_{y_K}(T_{x_K}\Mc),
\label{decomp-x_0}
\end{equation}
where in the last step we used \eqref{xi-pixi}.

{Suppose now that the product $\bigtimes U_K$ is the domain of the map \eqref{map-Mind}, which is used to define the map $\Phi$ via \eqref{coor-N}. Denote by $(x_K^{i_K})$ the value at $y=\{x_K\}$ of the coordinates defined by $\Phi$. Then the following curve      
\begin{multline*}
t\mapsto \big(\pi\circ(\bigtimes \varphi^{-1}_K)\big)(x^1_1,\ldots,x^{i-1}_I,x^{i}_I+t,x^{i+1}_I,\ldots,x^{\dim\Mc}_N)=\\=\big\{\varphi^{-1}_I(x^1_I,\ldots,x^{i-1}_I,x^{i}_I+t,x^{i+1}_I,\ldots,x^{\dim\Mc}_I)\big\}\cup\{y_I\}\in \Mc_N 
\end{multline*}
defines the tangent vector $\partial_{x^i_I}\in T_y\Mc_N$. Taking into account the formula \eqref{curves} we conclude that $\partial_{x^i_I}\in \xi^t_{y_I}(T_{x_I}\Mc)$. Hence the vectors $(\partial_{x^k_I})_{k=1,\ldots,\dim\Mc}$ (with the fixed index $I$) form a basis of $\xi^t_{y_I}(T_{x_I}\Mc)$.} 

To simplify the notation, in the main body of the paper we will identify 
\begin{equation}
\xi^t_{y_I}(T_{x_I}\Mc)\equiv T_{x_I}\Mc 
\label{T-id-T}
\end{equation}
and write 
\[
T_y\Mc_N=\bigoplus_{K=1}^N T_{x_K}\Mc.
\]  

\subsection{Diffeomorphisms of $\Mc_N$ induced by those of $\Mc$ \label{diff-M-N}} 

\paragraph{Definition} Let $\theta$ be a diffeomorphism on $\Mc$. It induces a map on $\Mc_N$ as follows: 
\begin{equation}
\Mc_N\ni \{x_K\}\mapsto \Theta(\{x_K\}):=\{\theta(x_K)\}\in \Mc_N
\label{diff-ind}
\end{equation}
(if $\{x_K\}$ consists of pairwise distinct points of $\Mc$, then the points in $\big\{\theta(x_K)\big\}$ are pairwise distinct too). Let us show now that $\Theta$ is a diffeomorphism on $\Mc_N$.  

To this end consider the following map
\[
\Mc^N_0\ni (x_K)\mapsto \thetabf(x_K):=\big(\theta(x_K)\big)\in \Mc^N_0
\]
Clearly,
\begin{equation}
\Theta\circ\pi=\pi\circ \thetabf.
\label{chi-pi}
\end{equation}
Consider now maps \eqref{coor-N} and \eqref{coor'-N} assuming that $U'_K=\theta(U_K)$ and $\varphi'_K=\varphi_K\circ\theta^{-1}$. Then using \eqref{chi-pi} we obtain
\begin{multline}
\Phi' \circ \Theta\circ \Phi^{-1}=\Phi' \circ \Theta\circ(\pi^{-1}_{\!\bigtimes\!U_K})^{-1}\circ \big(\bigtimes\varphi_K\big)^{-1}=\Phi' \circ \Theta\circ\pi\circ \bigtimes\varphi^{-1}_K =\\=\big(\bigtimes\varphi'_L\big)\circ \pi^{-1}_{\!\bigtimes\!U'_K}\circ \pi\circ \thetabf\circ \bigtimes\varphi^{-1}_K=\big(\bigtimes(\varphi_L\circ \theta^{-1})\big)\circ \bigtimes(\theta\circ\varphi^{-1}_K)=\id
\label{f-chi-f}
\end{multline}
on $\big(\bigtimes\varphi_L\big)\bigtimes U_K$. This means that for every diffeomorphism $\theta$ of $\Mc$, the map $\Theta$ is smooth. But $\Theta^{-1}$ exists and is smooth, since it is given via \eqref{diff-ind} by $\theta^{-1}$. $\Theta$ is thus a diffeomorphism of $\Mc_N$.      

If $\Theta$ is induced by $\theta$ via \eqref{diff-ind}, then 
\[
\Diff(\Mc)\ni \theta\mapsto \Theta\in \Diff(\Mc_N)
\]
is a homomorphism. Its image is a subgroup of the diffeomorphism group of $\Mc_N$. This subgroup will be denoted by $\Diff_\Mc(\Mc_N)$.    

Note also that it follows from \eqref{f-chi-f} that $\Phi\circ \Theta^{-1}$ coincides with $\Phi'$. This means that $\Theta^{-1}$ maps a chart in $\Ac$ to an other one in $\Ac$. In other words, the atlas $\Ac$ is preserved by all diffeomorphisms in $\Diff_\Mc(\Mc_N)$.

\paragraph{Diffeomorphisms in $\Diff_\Mc(\Mc_N)$ preserve the decompositions \eqref{decomp-x_0}} 

Let $\theta\in\Diff(\Mc)$ generates $\Theta\in\Diff_\Mc(\Mc_N)$. Using the notation introduced in Appendix \ref{app-decomp}, we have for every $x\in U_I$: 
\[
(\Theta\circ \xi_{y_I})(x)=\{\theta(x)\}\cup \Theta(y)_I=(\xi_{\Theta(y)_I}\circ\theta)(x),
\]
where $\Theta(y)_I$ is the set obtained by removing the point $\theta(x_I)$ from $\Theta(y)$. If $\Theta^t$ denotes the tangent map given by $\Theta$, then by virtue of the equation above,
\begin{equation}
\Theta^t\xi^t_{y_I}(T_{x_I}\Mc)=\xi^t_{\Theta(y)_I}\theta^t(T_{x_I}\Mc)=\xi^t_{\Theta(y)_I}(T_{\theta(x_I)}\Mc).
\label{Theta'-TxM}
\end{equation}
This means that $\Theta$ maps the decomposition \eqref{decomp-x_0} at $y$ into the decomposition \eqref{decomp-x_0} at $\Theta(y)$.  
    
\paragraph{Diffeomorphisms in $\Diff_\Mc(\Mc_N)$ preserve the spaces $\{\Gamma^\oplus_y\}$} Let $\theta\in\Diff(\Mc)$, and $\Theta$ be the corresponding element of $\Diff_\Mc(\Mc_N)$. To simplify the notation, in line with the identification \eqref{T-id-T}, we will denote elements of both $\xi_{y_I}(T_{x_I}\Mc)$ and $T_{x_I}\Mc$ by the same symbols $v_I,\check{v}_I$. Similarly, taking into account Equation \eqref{Theta'-TxM} we will identify $\Theta^t v_I$ being an element of $\Theta^t\xi^t_{y_I}(T_{x_I}\Mc)$ with $\theta^tv_I$ being an element of $T_{\theta(x_I)}\Mc$. Then for every $v\in T_y\Mc_N$    
\[
\Theta^tv=\Theta^t\Big(\sum_{I=1}^N v_I\Big)=\sum_{I=1}^N \theta^tv_I.
\] 

Let $y=\{x_K\}\in\Mc_N$ and let $x'_K=\theta^{-1}(x_K)$. Then $y'=\{x'_K\}=\Theta^{-1}(y)$. Suppose that $\{\gamma_{x_K}\}$ are scalar products (of the same signature) such that $\gamma_{x_I}\in \Gamma_{x_I}$ and that $\gamma^\oplus_y\in\Gamma^\oplus_y$ is constructed of these scalar products according to Equation \eqref{g^o+}. Consider now the pull-back $\Theta^{t *} \gamma^{\oplus}_y$:
\begin{equation*}
\big(\Theta^{t *}\gamma^\oplus_y\big)(v,\check{v})=\gamma^\oplus_y\big(\Theta^t v,\Theta^t\check{v}\big)=\sum_{I=1}^N\gamma_{x_I}(\theta^t v_I,\theta^t \check{v}_I)=\sum_{I=1}^N(\theta^{t *}\gamma_{x_I})(v_I,\check{v}_I),
\end{equation*}
where $v,\check{v}\in T_{y'}\Mc_N$. We see that the pull-back $\Theta^{t *}\gamma^\oplus_y$ is a scalar product on $T_{y'}\Mc_N$ constructed of scalar products $\{\theta^{t*}\gamma_{x_I}\}$ via \eqref{g^o+}. Therefore $\Theta^{t *}\gamma^\oplus_y\in \Gamma^\oplus_{y'}$.

We are then allowed to conclude that diffeomorphisms in $\Diff_\Mc(\Mc_N)$ preserve the spaces $\{\Gamma^\oplus_y\}_{y\in\Mc_N}$. Moreover, regarding the correspondence given by the bijection \eqref{1-to-1}, we see that the pull-back $\Theta^{t *}:\Gamma^\oplus_y\to\Gamma^\oplus_{y'}$ corresponds to the pull-back
\begin{equation}
\bigtimes \theta^{t *}\equiv \theta^{t *}\times\ldots\times \theta^{t *}:\Gamma_{x_1}\times\ldots\times\Gamma_{x_N}\to \Gamma_{x'_1}\times\ldots\times\Gamma_{x'_N}.
\label{Xtheta}
\end{equation}
More precisely,
\begin{equation}
b_{y'}\circ \Theta^{t *}=\big(\bigtimes \theta^{t *}\big)\circ b_y
\label{bb}
\end{equation}
provided the ordering of the spaces of scalar products $\{\Gamma_{x_K}\}$ and $\{\Gamma_{x'_K}\}$ is chosen as in \eqref{Xtheta}.

\section{Proof of Lemma \ref{lm-push-prod} \label{app-proof}}

Since $X$, $Y$, $X'$ and $Y'$ are second countable l.c.H. spaces, so are the products $X\times Y$ and $X'\times Y'$ \cite{cohn}. If two regular Borel measures are defined on second countable l.c.H. spaces, then their product is well defined and again is a regular Borel measure \cite{cohn}. Therefore both product measures, which appear in \eqref{mu-x-nu} are regular and Borel. That being the case, by virtue of the Riesz representation theorem, it is enough to show that for every $h\in C^c(X\times Y)$,  
\begin{equation}
\int_{X\times Y}h\,\big((\alpha_{\star}d\mu)\times(\beta_{\star}d\nu)\big)=\int_{X\times Y}h\,\big((\alpha\times\beta)_{\star}(d\mu\times d\nu)\big).
\label{int-int}
\end{equation}

Each second countable l.c.H. space is metrizable \cite{cohn}. Let then $\delta_X$ and $\delta_Y$ be corresponding metrics on $X$ and $Y$. Then 
\[
\delta\big((x,y),(\check{x},\check{y})\big):=\sqrt{\delta^2_X(x,\check{x})+\delta^2_Y(y,\check{y})}
\]   
is a metric on $X\times Y$ compatible with the product topology. The canonical projections $\pi_X:X\times Y\to X$, $(x,y)\mapsto x$,  and $\pi_Y:X\times Y\to Y$, $(x,y)\mapsto y$, are continuous maps.

Given $x\in X$, let us define 
\[
Y\ni y\mapsto h_x(y):=h(x,y)\in\R.
\]
The support of $h_x$, if non-empty, can be characterized as follows: $y\in \supp h_x$ if and only if for every $\eps>0$ there exists $\check{y}\in Y$ such that $h_x(\check{y})\neq 0$ and $\delta_Y(y,\check{y})<\eps$, or, equivalently, if and only if for every $\eps>0$ there exists a pair $(x,\check{y})\in X\times Y$ such that $h(x,\check{y})\neq 0$ and $\delta\big((x,y),(x,\check{y})\big)<\eps$. This last statement implies that $(x,y)\in \supp h$. We thus showed that if $y\in \supp h_x$, then $(x,y)\in \supp h$. But if  $(x,y)\in \supp h$, then $y=\pi_Y(x,y)\in \pi_Y(\supp h)$. Thus $\supp h_x \subset \pi_Y(\supp h)$. This inclusion holds also if $\supp h_x$ is empty. 

Note now that $\pi_Y(\supp h)$ is compact being the image of the compact set $\supp h$ under the continuous map $\pi_Y$. We see that $\supp h_x$ is a closed subset of a compact set and therefore is compact as well\footnote{More precisely, $\supp h_x\subset\pi_Y(\supp h)$ is a closed subset of $Y$ and therefore is a closed subset of $\pi_Y(\supp h)$ (i.e. is closed with respect to the subspace topology on $\pi_Y(\supp h)$). This means that $\supp h_x$ is a compact subset of $\pi_Y(\supp h)$ and therefore it is a compact subset of $Y$.}. 

On the other hand, continuity of $h$ implies continuity of $h_x$.

We conclude that for every $x\in X$, $h_x\in C^c(Y)$ and accordingly to \eqref{*dmu}
\begin{equation}
\int_Y h_x\,(\beta_\star d\nu)=\int_{Y'}(\beta^\star h_x)\,d\nu.
\label{int-hx}
\end{equation}

Consider now the following function
\begin{equation}
X\ni x\mapsto \check{h}(x):=\int_Y h_x\,(\beta_\star d\nu)\in\R.
\label{df-ch}
\end{equation}
Suppose that $x\not\in \pi_X(\supp h)$. This means that for every $y\in Y$, $h_x(y)=h(x,y)=0$ and consequently $\check{h}(x)=0$. Thus if $\check{h}(x)\neq 0$, then $x\in \pi_X(\supp h)$. Therefore $\supp \check{h}\subset\pi_X(\supp h)$, since $ \pi_X(\supp h)$ is compact and thereby closed. We conclude then that $\supp \check{h}$ is compact being a closed subset of a compact set.         

Let
\[
s=\sup_{({x},{y})\in X\times Y}|h({x},{y})|
\]
and
\[
Y\ni y \mapsto \mathbf{h}(y):=
\begin{cases}
s & \text{if $y\in \pi_{Y}(\supp h)$, }\\
0 & \text{otherwise}
\end{cases}.
\]      
Since $\pi_{Y}(\supp h)$ is compact and the measure $\beta_\star d\nu$ is regular, $\mathbf{h}$ is integrable with respect to the measure. Moreover, for every $x\in X$, $|h_x|\leq \mathbf{h}$. These two facts allow us to use the Lebesgue's dominated convergence theorem to conclude that the function $\check{h}$ is continuous\footnote{Note that since $X$ is metrizable we can apply here the criterion of continuity formulated in terms of sequences of arguments and values of a function.}.

We thus showed that $\check{h}\in C^c(X)$ and consequently by virtue of \eqref{*dmu} 
\begin{equation}
\int_X \check{h} \,(\alpha_\star d\mu)=\int_{X'}(\alpha^\star \check{h})\,d\mu.
\label{int-ch}
\end{equation}
  
Now we are ready to show that Equation \eqref{int-int} holds. To this end let us transform the l.h.s. of this equation using the Fubini-Tonelli theorem and the formulas \eqref{int-hx}, \eqref{df-ch} and \eqref{int-ch}:
\begin{multline*}
\int_{X\times Y}h\,\big((\alpha_{\star}d\mu)\times(\beta_{\star}d\nu)\big)=\int_{X}\Big[\int_Y h_x\,(\beta_{\star}d\nu)\Big](\alpha_{\star}d\mu)=\int_{X} \check{h}\,(\alpha_{\star}d\mu)=\\=\int_{X'} (\alpha^\star\check{h})\,d\mu=\int_{X'}\Big[\int_Y h_{\alpha(x')}\,(\beta_\star d\nu)\Big]d\mu=\\=\int_{X'}\Big[\int_{Y'} (\beta^\star h_{\alpha(x')})\,d\nu\Big]d\mu 
=\int_{X'\times Y'}\big((\alpha\times\beta)^\star h\big)\,(d\mu\times d\nu)=\\=\int_{X\times Y}h\, \big((\alpha\times\beta)_\star (d\mu\times d\nu)\big)
\end{multline*}
(note also that $\alpha\times \beta$ is a homeomorphism and therefore $(\alpha\times\beta)^\star h$ is continuous and compactly supported.)

\section{Proof of Lemma \ref{K1-Hilb} \label{K1-proof}}

To prove Lemma \ref{K1-Hilb} we have to show that \emi $\Kc_N$ is a linear space, \emii the map \eqref{KN-inner} is an inner product on $\Kc_N$ and \emiii $\Kc_N$ is complete in the norm defined by the inner product \eqref{KN-inner}. The proof of the lemma we are going to present here, follows a proof of an analogous lemma concerning the Hilbert sequence space $l^2$ (see e.g. \cite{foll,kreyszig}). 

\paragraph{$\Kc_N$ is a linear space} If $\Psi,\,\Psi'\in\Kc_N$ and $z\in\C$, then
\begin{align*}
(z\Psi)(y)&:=z\,\Psi(y), & (\Psi+\Psi')(y)&:=\Psi(y)+\Psi'(y).
\end{align*}
It is now obvious, that $z\Psi\in \Kc_N$ for every $z\in\C$. Regarding the sum $\Psi+\Psi'$ note first that the set
\begin{equation}
\{\ y\in\Mc_N \ | \ \Psi(y)\neq 0 \ \text{or}\ \Psi'(y)\neq 0\ \}
\label{yyy}
\end{equation}
is countable. Therefore all its elements can be ordered into a sequence $(y_n)$. The value $||\Psi+\Psi'||^2$ can be now bounded from above  as follows:
\begin{multline*}
||\Psi+\Psi'||^2=\sum_{n=1}^\infty||\Psi(y_n)+\Psi'(y_n)||^2_{y_n}\leq \sum_{n=1}^\infty \big(||\Psi(y_n)||_{y_n}+||\Psi'(y)||_{y_n}\big)^2\\
=\sum_{n=1}^\infty\big(||\Psi(y_n)||^2_{y_n}+2||\Psi(y_n)||_{y_n}||\Psi'(y_n)||_{y_n}+||\Psi(y_n)||^2_{y_n}\big),
\end{multline*}   
where $||\cdot||_y$ is the norm on $H^\otimes_y$. Note now that $2ab\leq a^2+b^2$ for every real numbers $a$ and $b$. Therefore
\[
||\Psi+\Psi'||^2\leq \sum_{n=1}^\infty \big(2||\Psi(y_n)||^2_{y_n}+2||\Psi'(y_n)||^2_{y_n}\big)=2\big(||\Psi||^2+||\Psi'||^2\big)<\infty.
\]   
Thus $\Psi+\Psi'\in\Kc_N$ and $\Kc_N$ is a linear space.

\paragraph{The map \eqref{KN-inner} is an inner product} To prove that the map \eqref{KN-inner} is well defined, consider again the same elements $\Psi,\,\Psi'\in\Kc_N$ and the same sequence $(y_n)$ of points in \eqref{yyy}. Then 
\[
\scal{\Psi'}{\Psi}=\sum_{y\in\Mc_N}\scal{\Psi'(y)}{\Psi(y)}_y:=\sum_{n=1}^\infty\scal{\Psi'(y_n)}{\Psi(y_n)}_{y_n},
\]   
where $\scal{\cdot}{\cdot}_y$ is the inner product on $H^\otimes_y$. Let us show now that the series above is absolutely convergent---then its sum does not depend on the ordering of points in \eqref{yyy} into a sequence and, consequently, $\scal{\Psi'}{\Psi}$ is well defined. 

To this end we will apply the Schwarz inequality to every term in the following series:
\begin{multline*}
\sum_{n=1}^\infty\big|\scal{\Psi'(y_n)}{\Psi(y_n)}_{y_n}\big|\leq \sum_{n=1}^\infty ||\Psi'(y_n)||_{y_n}||\Psi(y_n)||_{y_n}\leq\\ \leq\frac{1}{2}\sum_{n=1}^\infty \big(||\Psi'(y_n)||^2_{y_n}+||\Psi(y_n)||^2_{y_n}\big)=\frac{1}{2}\big(||\Psi'||^2+||\Psi||^2\big)<\infty 
\end{multline*}
(here in the second step we again used the inequality $2ab\leq a^2+b^2$). 

Thus the map \eqref{KN-inner} is well-defined. It is now an easy exercise to show that it is an inner product on $\Kc_N$ and that the norm defined on the set by the inner product coincides with \eqref{KN-norm}.   

\paragraph{$\Kc_N$ is complete} It remains to show that $\Kc_N$ equipped with the norm is a complete space. Let us then suppose that $(\Psi_m)_{m=1,2,\ldots}$ is a Cauchy sequence of elements of $\Kc_N$. Then the set
\[
\{\ y\in\Mc_N \ | \ \text{$\exists$  $m$ such that $\Psi_m(y)\neq 0$} \ \}
\]  
is countable and we can form a sequence $(y_n)$ using all elements of this set. Thus for every $\eps>0$, there exists $m_0$ such that for every $m,m'>m_0$,     
\begin{equation}
||\Psi_m-\Psi_{m'}||^2=\sum_{n=1}^\infty ||\Psi_m(y_n)-\Psi_{m'}(y_n)||^2_{y_n}<\eps^2.
\label{psim-psim'}
\end{equation}
Consequently, for every $n$ and  every $\eps>0$, there exists $m_0$ such that for each $m,m'>m_0$,
\[
||\Psi_m(y_n)-\Psi_{m'}(y_n)||^2_{y_n}<\eps^2.
\]
This implies that for every (fixed) $n$, the sequence $\big(\Psi_m(y_n)\big)$ has a limit in (the complete space) $H^\otimes_{y_n}$---this limit will be denoted by $\psi_n$. Let then $\Psi$ be a section of $\mathbf{H}^\otimes$ such that  
\[
\Psi(y):=
\begin{cases}
\psi_n & \text{if $y=y_n$}\\
0 & \text{otherwise}
\end{cases}.
\]
We will show now that \emi the sequence $(\Psi_m)$ converges to $\Psi$ in the norm \eqref{KN-norm} and \emii $\Psi\in\Kc_N$.    
 
It follows from \eqref{psim-psim'} that for every $l$ and for every $m,m'>m_0$,  
\[
\sum_{n=1}^l ||\Psi_m(y_n)-\Psi_{m'}(y_n)||^2_{y_n}<\eps^2.
\]
Therefore for every $l$ and for every $m>m_0$,
\[
\lim_{m'\to\infty}\sum_{n=1}^l ||\Psi_m(y_n)-\Psi_{m'}(y_n)||^2_{y_n}=\sum_{n=1}^l ||\Psi_m(y_n)-\Psi(y_n)||^2_{y_n}\leq\eps^2.
\]
Consequently, for $m>m_0$, passing to the limit as $l$ tends to the infinity, we obtain  
\begin{equation}
\sum_{n=1}^\infty ||\Psi_m(y_n)-\Psi(y_n)||^2_{y_n}=\sum_{y\in\Mc_N}||\Psi_m(y)-\Psi(y)||^2_{y}=||\Psi_m-\Psi||^2\leq\eps^2.
\label{last}
\end{equation}
We conclude then that $\Psi$ is the limit of $(\Psi_m)$. 

Evidently, the set \eqref{Psi-neq-0} for $\Psi-\Psi_m$ is countable. It follows from \eqref{last} that $\Psi-\Psi_m$ is of finite norm \eqref{KN-norm}. Thus $\Psi-\Psi_m$ is an element of $\Kc_N$. But because $\Psi_m\in\Kc_N$ and $\Kc_N$ is a linear space (as proven above), $\Psi$ belongs to $\Kc_N$.          

We thus showed that every Cauchy sequence of elements of $\Kc_N$ converges to an element of this space. The space is then complete.



\end{document}